\documentclass[12pt]{article}

\usepackage{verbatim,color,amssymb}
\usepackage{amsmath}					
\usepackage{amsthm}					
\usepackage{natbib}
\usepackage{multirow}
\usepackage{setspace}
\usepackage[mathscr]{euscript}
\usepackage{fancyhdr}
\usepackage{enumitem}
\usepackage{graphicx}
\usepackage{lineno}
\usepackage{array,booktabs}

\usepackage{sectsty}

\usepackage{multibib}
\newcites{latex}{References}

\usepackage{lscape}

\usepackage{subfig}
\usepackage{setspace}
\usepackage{tikz}
\usetikzlibrary{arrows}
\usepackage{multirow}

\newtheorem{Thm}{\underline{\bf Theorem}}
\newtheorem{Assmp}{\underline{\bf Assumptions}}

\newtheorem*{Proof*}{Proof}

\newtheorem{Lem}{\underline{\bf Lemma}}

\def\cC{\mathbb{C}}

\def\rR{\mathbb{R}}
\def\eE{\mathbb{E}}

\def\Z{{\cal Z}}

\def\diag{\hbox{diag}}
\def\Ind{\hbox{I}}
\def\wh{\widehat}
\def\wt{\widetilde}

\def\diag{\hbox{diag}}
\def\log{\hbox{log}}

\def\var{\hbox{var}}
\def\cov{\hbox{cov}}

\def\newlog{\hbox{newlog}}

\def\Bernoulli{\hbox{Bernoulli}}

\def\Dir{\hbox{Dir}}

\def\IG{\hbox{Inv-Ga}}

\def\Laplace{\hbox{Laplace}}
\def\LN{\hbox{LN}}
\def\MVN{\hbox{MVN}}
\def\MVLN{\hbox{MVLN}}

\def\Normal{\hbox{Normal}}
\def\TN{\hbox{TN}}
\def\Unif{\hbox{Unif}}

\def\P_25_ICML{{\it Proceedings of the 25th international conference on Machine learning}}

\def\bse{\begin{eqnarray*}}
\def\ese{\end{eqnarray*}}
\def\be{\begin{eqnarray}}
\def\ee{\end{eqnarray}}
\def\bq{\begin{equation}}
\def\eq{\end{equation}}

\def\wh{\widehat}

\def\trans{^{\rm T}}

\def\th{^{th}}

\def\b1e{{\mathbf e}}

\def\bB{{\mathbf B}}

\def\bC{{\mathbf C}}

\def\bD{{\mathbf D}}
\def\bI{{\mathbf I}}

\def\bL{{\mathbf L}}

\def\bP{{\mathbf P}}
\def\bR{{\mathbf R}}

\def\bS{{\mathbf S}}
\def\bt{{\mathbf t}}
\def\bT{{\mathbf T}}
\def\bu{{\mathbf u}}
\def\bU{{\mathbf U}}
\def\bV{{\mathbf V}}

\def\bW{{\mathbf W}}
\def\bx{{\mathbf x}}
\def\bX{{\mathbf X}}
\def\by{{\mathbf y}}
\def\bY{{\mathbf Y}}

\def\bZ{{\mathbf Z}}
\def\bS{{\mathbf S}}
\def\bzero{{\mathbf 0}}

\def\tmu{\widetilde{\mu}}

\newcommand{\bmu}{\mbox{\boldmath $\mu$}}

\newcommand{\bpi}{\mbox{\boldmath $\pi$}}

\newcommand{\bxi}{\mbox{\boldmath $\xi$}}
\newcommand{\bvartheta}{\mbox{\boldmath $\vartheta$}}
\newcommand{\bepsilon}{\mbox{\boldmath $\epsilon$}}
\newcommand{\btheta}{\mbox{\boldmath $\theta$}}
\newcommand{\bbeta}{\mbox{\boldmath $\beta$}}

\newcommand{\bzeta}{\mbox{\boldmath $\zeta$}}
\newcommand{\bsigma}{\mbox{\boldmath $\sigma$}}
\newcommand{\bSigma}{\mbox{\boldmath $\Sigma$}}

\newcommand{\bLambda}{\mbox{\boldmath $\Lambda$}}

\newcommand{\abs}[1]{\left\vert#1\right\vert}

\renewcommand\footnoterule{\kern-3pt \hrule \textwidth 2in \kern 2.6pt}

\def\boxit#1{\vbox{\hrule\hbox{\vrule\kern6pt \vbox{\kern6pt \textcolor{blue}{#1}\kern6pt}\kern6pt\vrule}\hrule}}

\def\authorfootnote#1{{\let\thefootnote\relax\footnotetext{#1}}}

\allowdisplaybreaks

\begin{document}
\thispagestyle{empty}
\baselineskip=28pt

\begin{center}
{\LARGE{\bf Bayesian Copula Density Deconvolution for Zero-Inflated Data in Nutritional Epidemiology}}
\end{center}
\baselineskip=12pt

\vskip 2mm
\begin{center}
Abhra Sarkar\\
abhra.sarkar@utexas.edu \\
Department of Statistics and Data Sciences,
The University of Texas at Austin\\
2317 Speedway D9800, Austin, TX 78712-1823, USA\\
\hskip 5mm \\
Debdeep Pati and Bani K. Mallick\\
debdeep@stat.tamu.edu and 
bmallick@stat.tamu.edu\\
Department of Statistics, Texas A\&M University\\ 
3143 TAMU, College Station, TX 77843-3143, USA\\
\hskip 5mm \\
Raymond J. Carroll\\
carroll@stat.tamu.edu\\
Department of Statistics, Texas A\&M University\\ 
3143 TAMU, College Station, TX 77843-3143, USA\\
School of Mathematical and Physical Sciences, University of Technology Sydney\\ 
Broadway NSW 2007, Australia\\
\end{center}

\vskip 8mm
\begin{center}
{\Large{\bf Abstract}} 
\end{center}
Estimating the marginal and joint densities of the long-term average intakes of different dietary components is an important problem in nutritional epidemiology. 
Since these variables cannot be directly measured, data are usually collected in the form of 24-hour recalls of the intakes, which show marked patterns of conditional heteroscedasticity. 
Significantly compounding the challenges, the recalls for episodically consumed dietary components also include exact zeros. 
The problem of estimating the density of the latent long-time intakes from their observed measurement error contaminated proxies 
is then a problem of deconvolution of densities with zero-inflated data. 
We propose a Bayesian semiparametric solution to the problem, building on a novel hierarchical latent variable framework that translates the problem to one involving continuous surrogates only. 
Crucial to accommodating important aspects of the problem, 
we then design a copula based approach to model the involved joint distributions, adopting different modeling strategies for the marginals of the different dietary components. 
We design efficient Markov chain Monte Carlo algorithms for posterior inference and illustrate the efficacy of the proposed method through simulation experiments. 
Applied to our motivating nutritional epidemiology problems, compared to other approaches, our method provides more realistic estimates of the consumption patterns of episodically consumed dietary components. 
\baselineskip=12pt

\vskip 8mm
\baselineskip=12pt
\noindent\underline{\bf Some Key Words}: Copula, Density deconvolution, 
Measurement error, Nutritional epidemiology, Zero inflated data.

\par\medskip\noindent
\underline{\bf Short/Running Title}: Deconvolution for Zero Inflated Data

\par\medskip\noindent
\underline{\bf Corresponding Author}: Abhra Sarkar (abhra.sarkar@utexas.edu) 

\clearpage\pagebreak\newpage
\pagenumbering{arabic}
\newlength{\gnat}
\setlength{\gnat}{16pt}
\baselineskip=\gnat


\section{Introduction}

{\bf Problem Statement:} Dietary habits are important for our general health and well-being, having been known to play important roles in the etiology of many chronic diseases. 
Estimating the long-term average intakes of different dietary components $\bX$ and their marginal and joint distributions 
is thus a fundamentally important problem in nutritional epidemiology. 

The dietary component may be a nutrient, like sodium, vitamin A etc., or a food group, like milk, whole grains etc. 
In any case, by the very nature of the problem, $\bX$ can never be observed directly. 
Data are thus often collected in the form of 24-hour recalls of the intakes. 
Many of the dietary components of interest are daily consumed.
Examples include total grains, sodium, etc., the recalls for which are all continuous, comprising only strictly positive intakes.
Compounding the challenge, interest may additionally lie in episodically consumed components 
whose long-term average intake is assumed to be strictly positive 
but the recalls are semicontinuous,
comprising positive recalls for consumption days and exact zero recalls for non-consumption days.
Examples include milk, whole grains etc.

Since dietary patterns often vary with energy levels, measured in total caloric intake, 
adjustments with energy provide a way of standardizing the dietary assessments. 
The recalls for energy are always continuous. From a statistical viewpoint they can thus be treated just like the regular components, 
and hence, with some abuse, will be referred to as such. 

When the recalls are recorded within a relatively short span of time, 
it may be assumed that the participants' dietary patterns $\bX$ will not have changed significantly over this period. 
Treating the recalls $\bY$, like the ones shown in Table \ref{tab: EATS data structure}, to be surrogates for the latent $\bX$, contaminated by measurement errors $\bU$, 
the problem of estimating the joint and marginal distributions of $\bX$ from the recalls $\bY$ then translates to a problem of multivariate deconvolution of densities with exact zero surrogates for some of the components. 

Throughout we adopt the following generic notation for marginal, joint and conditional densities, respectively. 
For random vectors $\bS$ and $\bT$, we denote the marginal density of $\bS$, 
the joint density of $(\bS,\bT)$, 
and the conditional density of $\bS$ given $\bT$, 
by the generic notation $f_{\bS}, f_{\bS,\bT}$ and $f_{\bS\vert \bT}$, respectively.
Likewise, for univariate random variables $S$ and $T$, the corresponding densities are denoted by $f_{S},f_{S,T}$ and $f_{S\mid T}$, respectively. 
Additional summaries of the variables and notations used can be found in Table \ref{tab: var description} below. 

{\bf The EATS Data Set and Its Prominent Features:} 
The main motivation behind the research being reported here comes from the Eating at America's Table Study (EATS) \citep{Subar2001}, 
a large scale epidemiological study conducted by the National Cancer Institute 
in which $i=1,\dots,n = 965$ participants were interviewed $j=1,\dots,m_{i}=4$ times over the course of a year and their 24-hour dietary recalls were recorded. 

Data on many different dietary components were recorded in the EATS study, 
including episodic components milk and whole grains, whose recalls involved approximately $21\%$ and $37\%$ exact zeros, respectively.
Table \ref{tab: EATS data structure} shows the general structure of this data set for one regularly consumed and one episodically consumed dietary component.

\begin{table}[!ht]\footnotesize
\begin{center}
\begin{tabular}{|c|c c c c|c c c c|}
\hline
\multirow{2}{35pt}{Subject} & \multicolumn{8}{|c|}{24-hour recalls} \\ \cline{2-9}
   & \multicolumn{4}{|c|}{Episodic Component}  					&  \multicolumn{4}{|c|}{Regular Component}	\\ \hline\hline
1 & $Y_{e,1,1}$ 	& $Y_{e,1,2}$ 	& $Y_{e,1,3}$ 	& $Y_{e,1,4}$ 		& $Y_{r,1,1}$ 	& $Y_{r,1,2}$ 	& $Y_{r,1,3}$ 	& $Y_{r,1,4}$\\ \hline
2 & $0$ 		 	& $Y_{e,2,2}$ 	& $Y_{e,2,3}$ 	& $Y_{e,2,4}$ 		& $Y_{r,1,1}$ 	& $Y_{r,2,2}$ 	& $Y_{r,2,3}$ 	& $Y_{r,2,4}$\\ \hline
3 & $Y_{e,3,1}$ 	& $0$ 		& $Y_{e,3,3}$ 	& $Y_{e,3,4}$ 		& $Y_{r,3,1}$ 	& $Y_{r,3,2}$ 	& $Y_{r,3,3}$ 	& $Y_{r,3,4}$\\ \hline
4 & $0$ 			& $Y_{e,4,2}$ 	& $Y_{e,4,3}$ 	& $0$ 			& $Y_{r,4,1}$ 	& $Y_{r,4,2}$ 	& $Y_{r,4,3}$ 	& $Y_{r,4,4}$ \\ \hline
$\cdots$ & $\cdots$ 	& $\cdots$ 	& $\cdots$ 	& $\cdots$ 		& $\cdots$ 	& $\cdots$ 	& $\cdots$ 	& $\cdots$ \\ \hline
n & $Y_{e,n,1}$ 	& $Y_{e,n,2}$ 	& $0$ 		& $0$ 			& $Y_{r,n,1}$ 	& $Y_{r,n,2}$ 	& $Y_{r,n,3}$ 	& $Y_{r,n,4}$ \\
\hline
\end{tabular}
\caption{\baselineskip=10pt The general structure of the EATS data set showing the recalls for one episodically consumed and one regularly consumed dietary component. 
Here $Y_{\ell,i,j}$ is the reported intake for the $j\th$ recall of the $i\th$ individual for the $\ell\th$ dietary component. 
\vspace{-20pt}
}
\label{tab: EATS data structure}
\end{center}
\end{table}

Patterns of conditional heteroscedasticity are also generally very prominent in dietary recall data. 
See, for example, the right panels of Figure \ref{fig: EATS Sodium & Energy} 
which shows the plot of subject-specific means $\overline{Y}_{\ell,i}=\sum_{j=1}^{m_{i}}Y_{\ell,i,j}/4$ vs subject-specific variances $S_{Y,\ell,i}^{2}=\sum_{j=1}^{m_{i}}(Y_{\ell,i,j}-\overline{Y}_{\ell,i})^{2}/3$ 
for the 24-hour recalls of sodium and energy, 
which provide crude estimates of the underlying true intakes $X_{\ell,i}$ and the conditional measurement error variances $\var(U_{\ell,i,j} \vert X_{\ell,i})$, respectively, 
suggesting strongly that $\var(U \vert X)$ increases as $X$ increases.  
Similar observation can also be made for positive recalls of episodic components from the middle panels of Figure \ref{fig: EATS Milk & Whole Grains}.

\begin{figure}[!ht]
\begin{center}
\includegraphics[height=10cm, width=11cm, trim=0cm 0cm 10.5cm 0cm, clip=true]{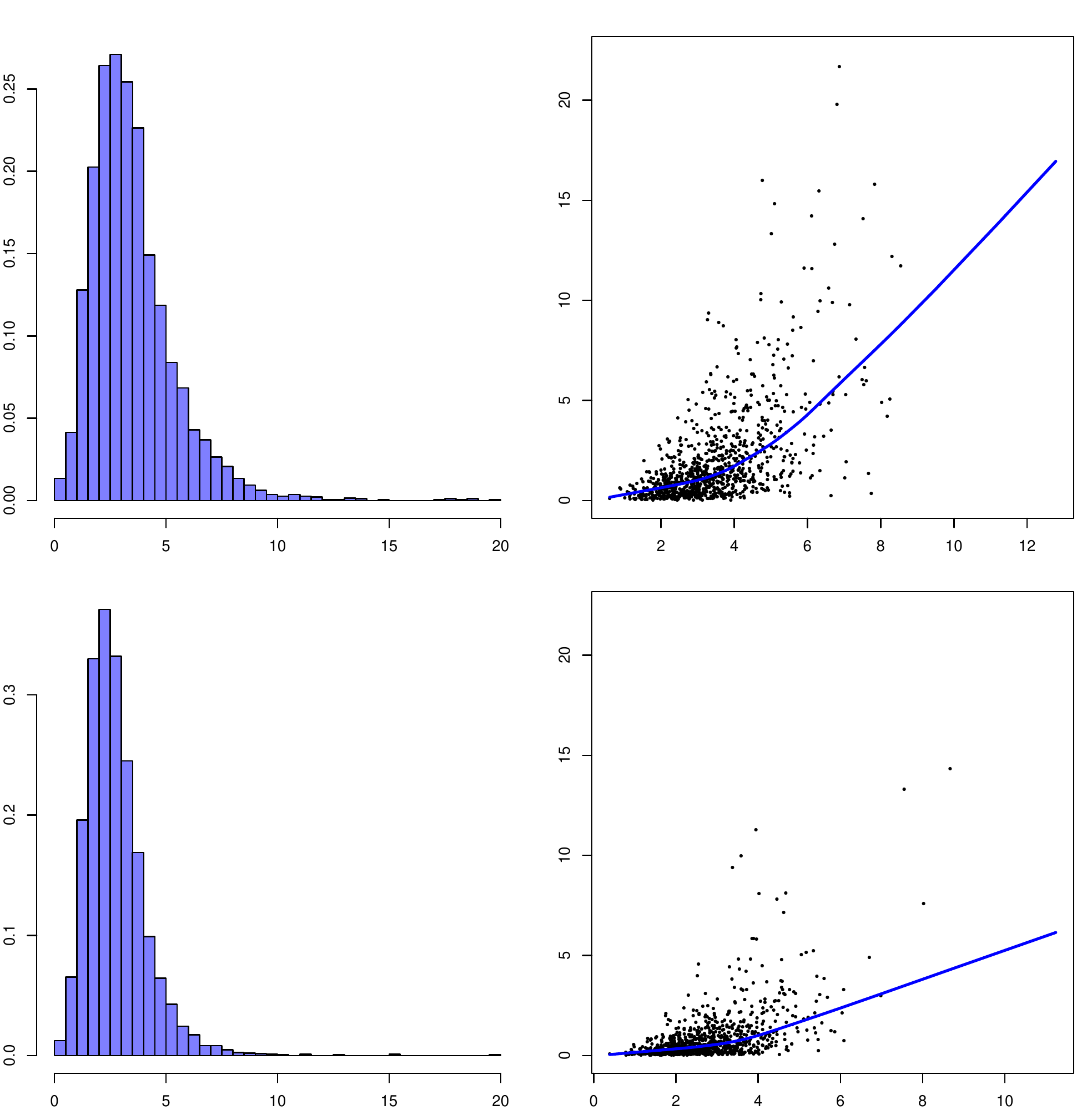}
\end{center}
\caption{\baselineskip=10pt Exploratory plots for sodium (top row) and energy (bottom row). 
Left panels: histogram of recalls $Y_{\ell,i,j}$; 
right panels: subject-specific means $\overline{Y}_{\ell,i}$ vs variances $S_{Y,\ell,i}^{2}$. 
}
\label{fig: EATS Sodium & Energy}
\end{figure}

As can be seen from the right and middle panels in Figures \ref{fig: EATS Sodium & Energy} and \ref{fig: EATS Milk & Whole Grains}, respectively, 
for both regular and episodic components, 
the variability of the positive recalls naturally decreases to zero as the average intake on consumption days decreases to zero. 
For all regularly consumed components, the histograms of the recalls are mildly right skewed bell shaped.
The histograms for the episodically consumed components are, however, reflected J-shaped - 
the frequencies of the bins start with their largest value at the left end and then rapidly decrease as we move to the right. 
These imply that, 
for regularly consumed components, the distributions of the true long-term average intakes smooth out near both ends, 
whereas, for episodically consumed components, the distributions of the true long-term average intakes have discontinuities at zero.    
The right panels of Figure \ref{fig: EATS Milk & Whole Grains} also show that, as expected, individuals consuming an episodic component in smaller amounts also consume it less often on average.

{\bf Existing Methods and Their Limitations:} 
The literature on univariate density deconvolution for continuous surrogates, 
in which context we denote the variable of interest by $X$ and the measurement errors by $U$, is massive. 
The early literature, however, focused on scenarios with restrictive assumptions, 
such as known measurement error distribution, homoscedasticity of the errors, their independence from $X$ etc, 
which are all highly unrealistic, 
especially in nutritional epidemiology applications like ours. 
Reviews of these early methods can be found in \cite{Carroll2006} and \cite{Buonaccorsi2010}. 
We cite below some relatively recent ideas that are directly relevant to our proposed solution.  


Bayesian frameworks can accommodate measurement errors through natural hierarchies, 
providing powerful tools for solving complex deconvolution problems, 
including scenarios when the measurement errors can be conditionally heteroscedastic. 
Taking such a route, \cite{Staudenmayer_etal:2008} assumed the measurement errors to be normally distributed but allowed the variability of $U$ to depend on $X$, 
utilizing mixtures of B-splines to estimate $f_{X}$ as well the conditional variability $\var(U \vert X)$.
 \cite{Sarkar_etal:2014} relaxed the assumption of normality of $U$, 
employing flexible mixtures of normals \citep{Escobar_West:1995, fruhwirth2006finite} to model both $f_{X}$ and $f_{U\vert X}$. 
\cite{sarkar2018bayesian} extended the methods to multivariate settings, 
modeling $f_{\bX}$ and $f_{\bU\vert \bX}$ using mixtures of multivariate normals. 

\begin{figure}[!ht]
\begin{center}
\includegraphics[height=10cm, width=17cm, trim=0cm 0cm 0cm 0cm, clip=true]{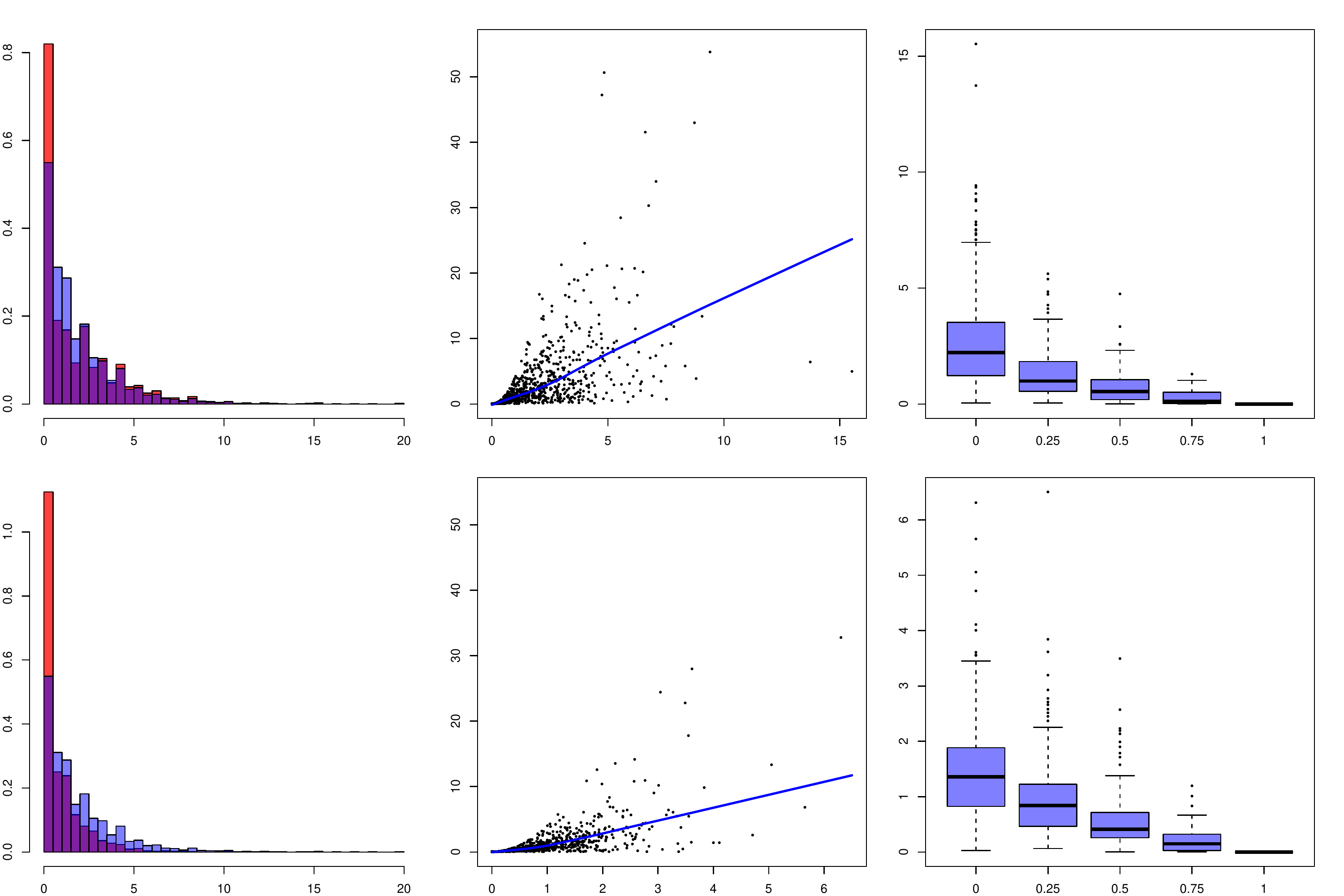}
\end{center}
\caption{\baselineskip=10pt Exploratory plots for milk (top row) and whole grains (bottom row). 
Left panels: histogram of recalls $Y_{\ell,i,j}$ (red) and histogram of strictly positive recalls $Y_{\ell,i,j}(>0)$ (blue) superimposed on each other; 
middle panels: subject-specific means $\overline{Y}_{\ell,i}$ vs subject-specific variances $S_{Y,\ell,i}^{2}$ when multiple strictly positive recalls are available; 
right panels: box plots of proportion of zero recalls vs corresponding subject-specific means $\overline{Y}_{\ell,i}$.
}
\label{fig: EATS Milk & Whole Grains}
\end{figure}

While \cite{Staudenmayer_etal:2008} and \cite{Sarkar_etal:2014, sarkar2018bayesian} provided progressively flexible frameworks for 
univariate and multivariate deconvolution with continuously measured surrogates, 
they can not directly handle multivariate zero-inflated dietary recall data. 
There are several restrictive aspects of their approaches that also do not allow them to be straightforwardly extended to deconvolution problems with zero-inflated surrogates, 
as we outline shortly while describing our proposed approach.  


%
The problem of estimating long-term nutritional intakes of a single episodic dietary component 
from zero-inflated recall data has previously been considered in \cite{Tooze2002, Tooze2006, Kipnis_etal:2009, Zhang2011a}. 
The work was extended to multivariate settings with both episodic and regular components in \cite{Zhang2011b}. 
These approaches all worked with component-wise Box-Cox transformed \citep{box1964analysis} positive recalls which were then assumed to decompose into a subject specific random effect component and an error or pseudo-error component.  
Assumed independent and homoscedastic, these components were then both modeled using single component multivariate normal distributions.  
Estimates of the long-term consumption day intakes were then obtained via individual transformations back to the original scale. 
Long-term episodic consumptions were finally defined combining these estimates with probabilities of reporting non-consumptions. 
As shown in \cite{Sarkar_etal:2014}, Box-Cox transformations for surrogate observations have severe limitations, 
including almost never being able to produce transformed surrogates that conform to normality, homoscedasticity, and independence. 
Transformation-retransformation based methods are thus highly restrictive, even for univariate regularly consumed components. 

Despite the limitations, to our knowledge, \cite{Zhang2011b} is the only available method that can handle multivariate zero-inflated dietary recall data. 
It is thus also our main and only competitor.

{\bf Outline of Our Proposed Method:}
In this article, we develop a Bayesian semiparametric density deconvolution approach specifically designed to address problems with zero-inflated surrogates, 
carefully accommodating all prominent features of the EATS data set described above. 
We build on an augmented latent variable framework 
which introduces, for each recall of the episodically consumed component, one or two latent continuous proxies, 
depending on whether the recall was positive or exact zero, 
effectively translating a deconvolution problem with zero-inflated data to one with all continuous surrogates, albeit some latent ones. 
This requires modeling an additional pseudo-error distribution for each episodically consumed component,
but returns, as potentially useful by-products, 
estimates of the probabilities of reporting zero recalls for the episodically consumed dietary components. 
As the right panels of Figure \ref{fig: EATS Milk & Whole Grains} suggest, 
individuals who consume an episodic component less often (in other words, report more zero recalls) 
naturally also consume the component in smaller amounts in the long run.
The probabilities of reporting zero consumptions are thus informative about the true long-term consumption amounts and conversely. 
Our proposed latent variable framework appropriately recognizes these features.  

Even though the multivariate latent consumptions $\bX$ and the associated multivariate errors and pseudo-errors $\bU$ 
become all strictly continuous in our augmented latent variable framework, 
the approach of \cite{sarkar2018bayesian} to model their distributions using mixtures of multivariate normals is still fraught with serious practical drawbacks 
as it does not allow much flexibility in modeling the univariate marginals $f_{X_{\ell}}$ and $f_{U_{\ell}\vert X_{\ell}}$, 
especially the marginals of the episodic components which have discontinuities at zero.  
The issue becomes more critical when inference is based on samples drawn from the posterior using Markov chain Monte Carlo (MCMC) algorithms. 
The latent $\bX_{i}$'s are also sampled in the process 
and the specific parametric form of the assumed multivariate mixture kernel may influence this step in ways 
that result in density estimates closely resembling its parametric form 
even when the shape of the true density departs from it. 

As opposed to \cite{sarkar2018bayesian} who focused on modeling the joint distributions $f_{\bX}$ and $f_{\bU\mid\bX}$ first and then deriving the marginals from those estimates, 
we take the opposite approach of modeling the marginals $f_{X_{\ell}}$ and $f_{U_{\ell}\vert X_{\ell}}$ first 
and then build the joint distributions $f_{\bX}$ and $f_{\bU\mid\bX}$ by modeling the dependence structures separately using Gaussian copulas. 
This approach allows us adopt different strategies for modeling the different components of $f_{\bX}$ and $f_{\bU \mid \bX}$ 
which proved crucial in accommodating the important features of our motivating data sets. 
Following \cite{Sarkar_etal:2014}, we use flexible mixtures of mean restricted normals 
and mixtures of B-splines to model $f_{U_{\ell}\vert X_{\ell}}$'s and the associated conditional heteroscedasticity functions.
Mixtures of normal kernels, as in \cite{Sarkar_etal:2014}, are, however, not suitable for modeling $f_{X_{\ell}}$'s. 
We use normalized mixtures of B-splines and mixtures of truncated normal kernels instead which are well suited to model densities with bounded supports and discontinuities at the boundaries. 

The literature on copula models in measurement error free scenarios is vast. See, for example, \cite{nelsen2007introduction, joe2015dependence, shemyakin2017introduction} and the references therein. 
We are, however, unaware of any published work in the context of measurement error problems. 

In contrast to \cite{Zhang2011b}, we model the densities of the latent consumptions and the error and pseudo-errors more directly 
using flexible models that can accommodate widely varying shapes with discontinuous boundaries as well as conditional heteroscedasticity. 
In our latent variable framework, the probability of reporting zero recalls depends directly on the latent true consumption day intake, hence informing each other. 
Applied to our motivating nutritional epidemiology problems, our method thus provides more realistic estimates of the intakes of the episodically consumed dietary components. 
Additional detailed comparisons of our method with previous approaches for zero-inflated data are presented in Section \ref{sec: mvt copula comparison with NCI} in the supplementary material. 

Compared to all previously existing density deconvolution methods, 
including traditional methods for strictly continuous data as well as methods designed specifically for zero-inflated data, 
our proposed approach is thus fundamentally novel 
while also being broadly applicable to both scenarios.

{\bf Outline of the Article:} The rest of the article is organized as follows. 
Section \ref{sec: models} details the proposed Bayesian hierarchical framework.  
Simulation studies comparing the proposed method to its main competitor are presented in Section \ref{sec: simulation studies}. 
Section \ref{sec: applications} presents results of our proposed method applied to the motivating nutritional epidemiology problems. 
Section \ref{sec: discussion} concludes with a discussion.
%
A brief review of copula, a detailed comparison of our method with previous approaches to zero-inflated data, 
a Markov chain Monte Carlo (MCMC) algorithm to sample from the posterior and some additional results are included in the supplementary material. 

\section{Deconvolution Models} \label{sec: models}

\subsection{Latent Variable Framework} \label{sec: mvt copula latent variable framework}

Our goal is to estimate the marginal and joint consumption patterns of $q+p$ dietary components 
of which the first $q$ are episodically consumed and the latter $p$ are regularly consumed, including energy.
There are a total of $n$ subjects with $m_{i}$ 24-hour recalls recorded for the $i\th$ subject.
We let $\bY_{i,j} = (Y_{1,i,j},\dots,Y_{2q+p,i,j})\trans$ denote the observed data for the $j\th$ recall of the $i\th$ individual. 
For $\ell=1\dots,q$, $Y_{\ell,i,j}$ is the indicator of whether the $\ell\th$ episodic component is reported to have been consumed. 
For $\ell=q+1,\dots,2q$, $Y_{\ell,i,j}$ is the reported intake of the $\ell\th$ episodically consumed component, 
and for $\ell=2q+1,\dots,2q+p$, $Y_{\ell,i,j}$ is the reported intake of the $\ell\th$ regularly consumed component.
Let $\bW_{i,j} = (W_{1,i,j},\dots,W_{2q+p,i,j})\trans$ denote a vector with all continuous components that are related to the observed data $\bY_{i,j}$ by the relationships
\vspace{-6ex}\\
\be
Y_{\ell,i,j} &=& \Ind(W_{\ell,i,j}>0),~~~~~\hbox{for}~\ell=1,\dots,q, \nonumber\\
Y_{\ell,i,j} &=& Y_{\ell-q,i,j} W_{\ell,i,j},~~~~~\hbox{for}~\ell=q+1,\dots,2q,   \label{eq: episodic Y2}\\
Y_{\ell,i,j} &=& W_{\ell,i,j},~~~~~~~~~~~~~\hbox{for}~\ell=2q+1,\dots,2q+p.   \nonumber
\ee
\vspace{-4ex}

For $\ell=1,\dots,q$, $W_{\ell,i,j}$ indicates whether the $\ell\th$ episodic component is reported to have been consumed in the $j\th$ recall of the $i\th$ individual 
and is always latent except that we know whether it is positive or negative. 
That is, for $\ell=1,\dots,q$, $W_{\ell,i,j}$ is always latent with $W_{\ell,i,j}<0$ if $Y_{\ell,i,j}=0$ and $Y_{q+\ell,i,j}=0$, 
and $W_{\ell,i,j} \geq 0$ if $Y_{\ell,i,j}=1$ and $Y_{q+\ell,i,j}>0$. 

For $\ell=q+1,\dots,2q$, 
$W_{\ell,i,j}$ is latent if the $\ell\th$ episodic component is reported to have not been consumed in the $j\th$ recall of the $i\th$ individual 
and is observed and equals the reported consumed positive amount $Y_{\ell,i,j}$ otherwise.  
That is, for $\ell=q+1,\dots,2q$, $W_{\ell,i,j}$ is latent if $Y_{\ell-q,i,j}=0$ and $Y_{\ell,i,j}=0$ and is observed with $W_{\ell,i,j}=Y_{\ell,i,j}$ if $Y_{\ell-q,i,j}=1$ and $Y_{\ell,i,j}>0$. 

For $\ell=2q+1,\dots,2q+p$, $W_{\ell,i,j}$ denotes the reported intake of the $\ell\th$ regular component and is always observable. 
That is, for $\ell=2q+1,\dots,2q+q$, $W_{\ell,i,j}=Y_{\ell,i,j}>0$. 

We let $\bX_{i}=(X_{1,i},\dots,X_{q+p,i})\trans$ denote the latent daily average long-term intakes of the $i\th$ individual, consumption and non-consumption days combined.  
We now let $\bX_{i}^{+}=(X_{1,i}^{+},\dots,X_{q+p,i}^{+})\trans$ denote the latent daily average long-term intakes of the $i\th$ individual on consumption days only.  
We then define $\wt\bX_{i}=(\wt{X}_{1,i},\dots,\wt{X}_{2q+p,i})\trans$ as  
\vspace{-4ex}\\
\bse
\wt{X}_{\ell,i} &=& h_{\ell}(X_{\ell,i}),~~~~~~~~~~~\hbox{for}~\ell=1,\dots,q, \label{eq: episodic X1}\\
\wt{X}_{\ell,i} &=& X_{\ell-q,i}^{+},~~~~~~~~~~~~~\hbox{for}~\ell=q+1,\dots,2q,   \label{eq: episodic X2}\\
\wt{X}_{\ell,i} &=& X_{\ell-q,i},~~~~~~~~~~~~~\hbox{for}~\ell=2q+1,\dots,2q+p.   \label{eq: regular X3}
\ese
\vspace{-4ex}\\
Here $h_{\ell}(\cdot)$ is an unknown function to be estimated from data. 
The reasons behind defining $\wt\bX_{i}$ in this manner will be clear shortly. 

For $i=1,\dots,n, j=1,\dots,m_{i}$, we let $\bU_{i,j}=(U_{1,i,j},\dots,U_{2q+p,i,j})\trans$ and consider the model
\vspace{-4ex}\\
\bse
\bW_{i,j} &=& \wt\bX_{i} + \bU_{i,j},~~~~~~~~\eE(\bU_{i,j} \mid \wt\bX_{i}) = \bzero.  \label{eq: episodic component latent and observed proxies}
\ese
\vspace{-4ex}\\
For $\ell=1,\dots,q$, $W_{\ell,i,j}$ is always latent and the associated $U_{\ell,i,j}$ represents a pseudo-error that account for their within person daily variations. 
For $\ell=q+1,\dots,2q$, $U_{\ell,i,j}$ denotes the measurement error contaminating $W_{\ell,i,j}$ when it is observed and pseudo-errors when they are latent.  
Finally, for $\ell=2q+1,\dots,2q+p$, $U_{\ell,i,j}$ denotes the measurement error contaminating $W_{\ell,i,j}$ which are always observed.  

According to our model, for $\ell=1,\dots,q$, the probability of reporting a positive consumption on the $\ell\th$ episodic component, denoted henceforth as $P_{\ell}({X}_{\ell,i})$, 
is obtained as $\Pr(Y_{\ell,i,j} = 1 \mid {X}_{\ell,i}) = \Pr(W_{\ell,i,j} >0 \mid {X}_{\ell,i}) =  \Pr\{U_{\ell,i,j} > -h_{\ell}({X}_{\ell,i}) \mid {X}_{\ell,i}\}$.
For $\ell=1,\dots,q$, we also have $\eE(Y_{\ell+q,i,j} \mid Y_{\ell,i,j}=1,\wt{X}_{\ell+q,i}) = \eE(W_{\ell+q,i,j} \mid Y_{\ell,i,j}=1,\wt{X}_{\ell+q,i}) = \eE(W_{\ell+q,i,j} \mid \wt{X}_{\ell+q,i}) = \wt{X}_{\ell+q,i}=X_{\ell,i}^{+}$. 
The positive recalls $Y_{\ell,i,j}$'s and the $W_{\ell+q,i,j}$'s, latent or observed, are thus unbiased for the latent average long-term intakes of the episodic components on consumption days only.   
For $\ell=1,\dots,q$, the expectation $\eE(Y_{\ell+q,i,j} \mid X_{\ell,i}, X_{\ell,i}^{+}) = \Pr(W_{\ell,i,j}>0 \mid {X}_{\ell,i}) \eE(W_{\ell+q,i,j} \mid X_{\ell,i}^{+}) = P_{\ell} (X_{\ell,i}) X_{\ell,i}^{+}$ 
then defines the overall long-term average intake, consumption and non-consumption days combined. 
By definition, this is also $X_{\ell,i}$, giving us the relationship $X_{\ell,i} = P_{\ell} (X_{\ell,i}) X_{\ell,i}^{+}$. 

For regularly consumed components $\ell=q+1,\dots,q+p$, of course, $\wt{X}_{\ell+q,i} = X_{\ell,i}^{+} = X_{\ell,i}$. 
The recalls in these cases are all observed and are unbiased for the latent long-term intakes as $\eE(Y_{\ell+q,i,j} \mid \wt{X}_{\ell+q,i}) = \eE(W_{\ell+q,i,j} \mid \wt{X}_{\ell+q,i}) = \wt{X}_{\ell+q,i}$.
 
Written in terms of the long-term average intakes $X_{\ell,i}$, the model thus becomes 
\vspace{-4ex}\\
\be
W_{\ell,i,j} &=& h_{\ell}(X_{\ell,i})+U_{\ell,i,j},~~~~~~~~~~~~~~~~~~~~\hbox{for}~\ell=1,\dots,q,		\nonumber\\
W_{\ell,i,j} &=& X_{\ell-q,i}/P_{\ell-q}(X_{\ell-q,i})+U_{\ell,i,j},~~~~~\hbox{for}~\ell=q+1,\dots,2q,		 	\label{eq: episodic W2}\\
W_{\ell,i,j} &=& X_{\ell-q,i}+U_{\ell,i,j},~~~~~~~~~~~~~~~~~~~~~~\hbox{for}~\ell=2q+1,\dots,2q+p.  	\nonumber 
\ee
\vspace{-4ex}\\
This formulation now allows the problem to be reduced to that of modeling the components $f_{\bX}$, $f_{\bU \vert \wt\bX}$ and $P_{\ell}(X_{\ell})$ in a Bayesian hierarchical framework. 
It also simplifies the estimation of the distribution energy-adjusted intakes. 
We address this latter problem in Section \ref{sec: mvt copula density of energy adjusted intakes}. 

The complex nature of our problem warranted the introduction of many different variables representing the many random variables of our model. 
For easy reference, these variables and a few others to be introduced shortly are listed in Table \ref{tab: var description}.

\vspace{0.5cm}
\begin{table}[!htb]
\begin{center}
\footnotesize
\begin{tabular}{|p{20mm}|p{125mm}|}
\hline
Notation 			&	Description \\ \cline{1-2}
$q$ 				&	Number of episodically consumed components. \\ \hline
$p$ 				&	Number of regularly consumed components. \\ \hline
$Y_{\ell,i,j}$		&	Observed recall of the $\ell\th$ dietary component for the $i\th$ individual on the $j\th$ sampling occasion 
					- binary for $\ell=1,\dots,q$, zero if the component was not consumed, one otherwise;  
					zero or positive continuous for $\ell=q+1,\dots,2q$, representing the reported intakes, zero when the component was not consumed, positive continuous otherwise; 
					positive continuous for $\ell=2q+1,\dots,2q+p$, representing the reported intakes. \\ \hline
$W_{\ell,i,j}$		&	Proxy recall of the $\ell\th$ dietary component for the $i\th$ individual on the $j\th$ sampling occasion - always continuous; 
					latent for $\ell=1,\dots,q$, negative if $Y_{\ell,i,j}=0$, positive if $Y_{\ell,i,j}=1$;  
					latent or observed for $\ell=q+1,\dots,2q$, latent when the component was not consumed, observed and equals $Y_{\ell,i,j}$ when a positive recall was recorded; 
					positive  for $\ell=2q+1,\dots,2q+p$, equaling $Y_{\ell,i,j}$, the reported positive intake. \\ \hline
$X_{\ell,i}$ 		&	Long-term daily average intake of the $\ell\th$ dietary component for the $i\th$ individual, consumption and non-consumption days combined. 
					Strictly positive and continuous.
					For $\ell=1,\dots,q+p$, the observed recalls $Y_{q+\ell,i,j}$ are unbiased for $X_{\ell,i}$. \\ \hline
$X_{\ell,i}^{+}$		&	Long-term daily average intake of the $\ell\th$ dietary component for the $i\th$ individual, on consumption days only. 
					Strictly positive and continuous.
					For $\ell=1,\dots,q+p$, the proxy recalls $W_{q+\ell,i,j}$'s are unbiased for the $X_{\ell,i}^{+}$'s. \\ \hline
$P_{\ell}(X_{\ell,i})$	&	Probability of reporting positive consumption on the $\ell\th$ dietary component by the $i\th$ individual on any sampling occasion. \\ \hline
$\wt{X}_{\ell,i}$		&	Functions of $X_{\ell,i}, X_{\ell,i}^{+}$ and $P_{\ell}(X_{\ell,i})$ such that $W_{\ell,i,j}$ is unbiased for $\wt{X}_{\ell,i}$. \\ \hline
$U_{\ell,i,j}$		&	Measurement errors or pseudo-errors contaminating $\wt{X}_{\ell,i}$ to generate $W_{\ell,i,j}$. 
					The $U_{\ell,i,j}$'s are all unbiased for zero.
					For $\ell=q+1,\dots,2q+p$, variability of $U_{\ell,i,j}$ depends on the associated $\wt{X}_{\ell,i}$. \\ \hline
$s_{\ell}^{2}(\wt{X}_{\ell,i})$	&	Variance function explaining how the conditional variability of $U_{\ell,i,j}$ depends on the associated $\wt{X}_{\ell,i}$ for $\ell=q+1,\dots,2q+p$. \\ \hline
$\epsilon_{\ell,i,j}$	&	Scaled measurement error or pseudo-error obtained by scaling $U_{\ell,i,j}$ by $s_{\ell}(\wt{X}_{\ell,i})$. 
					The $\epsilon_{\ell,i,j}$'s are unbiased for zero, homoscedastic and independent of $\wt{X}_{\ell,i}$. \\ \hline
$Z_{\ell,i}$		&	Long-term daily average normalized intake of the $\ell\th$ dietary component for the $i\th$ individual, normalized by energy. \\
\hline
\end{tabular}
\caption{\baselineskip=10pt 
Variables representing the data and other random variables in our model.
}
\label{tab: var description}
\end{center}
\end{table}

\begin{figure}[!ht]
\centering
\includegraphics[height=5.75cm, width=7.5cm, trim=2cm 1cm 1cm 1cm, clip=true]{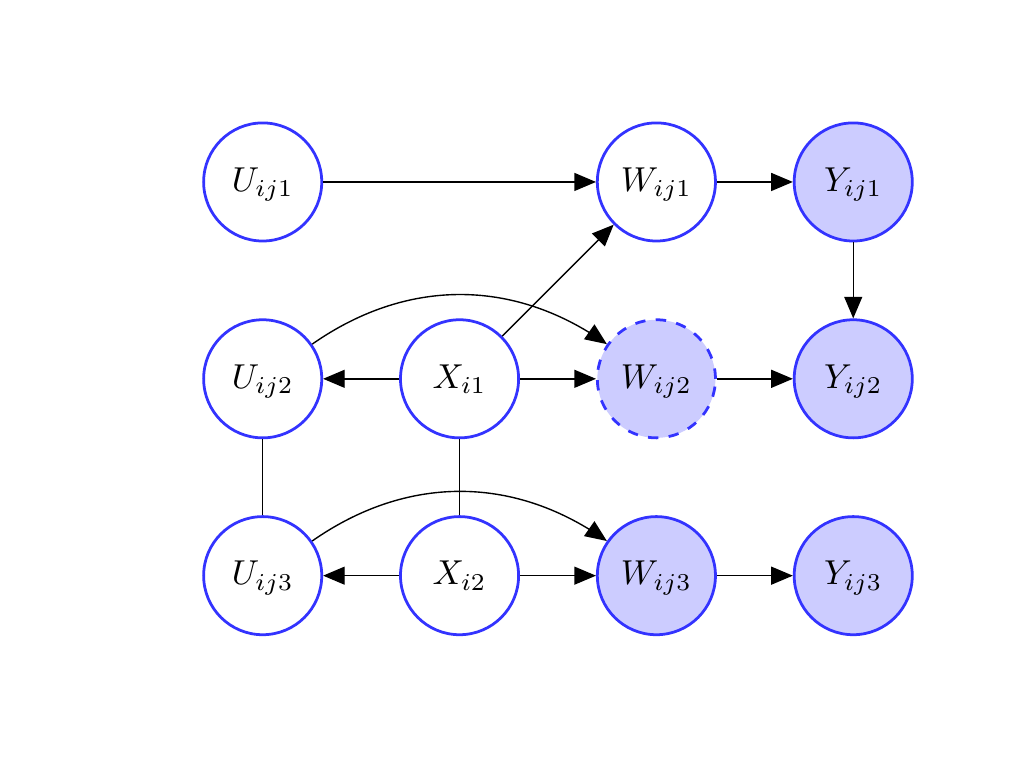}
\label{fig: Graphical model}
\caption{\baselineskip=10pt Graphical model depicting the dependency structure of the generative deconvolution model described in Section \ref{sec: models} 
for one episodically consumed component $X_{1}$ and one regularly consumed component $X_{2}$.
The unfilled and shaded nodes with solid boundaries signify latent and observable variables, respectively. The filled node with dashed boundary may be observed on some of occasions and latent on others.
}
\label{fig: graph our model}
\end{figure}

\subsection{Modeling the Density $f_{\bX}$} \label{sec: mvt copula density of interest}
In this article, $f_{\bX}$ is specified using a Gaussian copula density model 
\vspace{-4ex}\\
\bse
\textstyle f_{\bX}(\bX) = |\bR_{\bX}|^{-\frac{1}{2}} \exp\left\{-\frac{1}{2}\bY_{\bX}\trans(\bR_{\bX}^{-1}- \bI_{q+p})\bY_{\bX}\right\}  \prod_{\ell=1}^{q+p}f_{X,\ell}(X_{\ell}),  \label{eq: copula mixture for f_X} 
\ese
\vspace{-4ex}\\
with $F_{X,\ell}(X_{\ell}) = \Phi(Y_{X,\ell})$ for all $\ell$ and $\bR_{\bX}$ is the correlation matrix of $\bX$.

In initial attempts, we modeled the marginal densities $f_{X,\ell}$ as flexible mixtures of truncated normal kernels $\TN(\cdot \mid \mu,\sigma^{2},[A,B])$ with location $\mu$ and scale $\sigma$ and range restricted to the interval $[A,B]$.
In multivariate applications such as ours, where the components represent similar variables and have highly overlapping supports,  
we can greatly reduce dimension and borrow information across different dietary components, 
by allowing the component specific parameters of the mixture models to be shared among the variables. 
We thus modeled the marginal densities as 
\vspace{-4ex}\\
\bse
&& f_{X,\ell}(X_{\ell}) = \textstyle\sum_{k=1}^{K_{X}}\pi_{X,\ell,k}~\TN(X_{\ell} \mid \mu_{X,k},\sigma_{X,k}^{2},[A_{\ell},B_{\ell}]), \label{eq: tunc norm mixture for f_X}\\
&& \hspace{-1cm} \bpi_{X,\ell} \sim \Dir(\alpha_{X,\ell}/K_{X},\dots,\alpha_{X,\ell}/K_{X}), ~~\mu_{X,k}  \sim \Normal(\mu_{X,0},\sigma_{X,0}^2), ~~ \sigma_{X,k}^{2} \sim \IG(a_{\sigma_{X,0}^{2}},b_{\sigma_{X,0}^{2}}).
\ese
\vspace{-4ex}\\
The models for different components $\ell$ thus share the same atoms $(\mu_{X,k},\sigma_{X,k}^{2})$ but with varying probability weights $\pi_{X,\ell,k}$. 

Despite being specifically tailored to capture boundary discontinuities, in numerical experiments, we found the model to often 
produce steeply decaying and highly peaked estimates with underestimated (local) variance in these regions. 
After further investigations, we could attribute the issue to smoothness properties of such models, 
characterized by the variance components $\sigma_{X,k}^{2}$ 
which are estimated `locally' utilizing only the data points associated with the corresponding mixture components. 
For the episodic components, 
the scarcity of informative observations near the left boundaries often allows the sampled latent $X_{\ell,i}$'s to cluster away from these boundaries, 
resulting in the associated $\sigma_{X,k}^{2}$'s to be underestimated and hence the estimated densities to be peaked away from the boundaries. 
Setting informative lower bounds to the variance parameters solves the problem. 
Determining such bounds for the latent variables from their contaminated recalls, however, proved to be difficult.

\begin{figure}[!ht]
\begin{center}
\includegraphics[width=16cm, trim=1cm 1.5cm 0cm 0cm, clip=true]{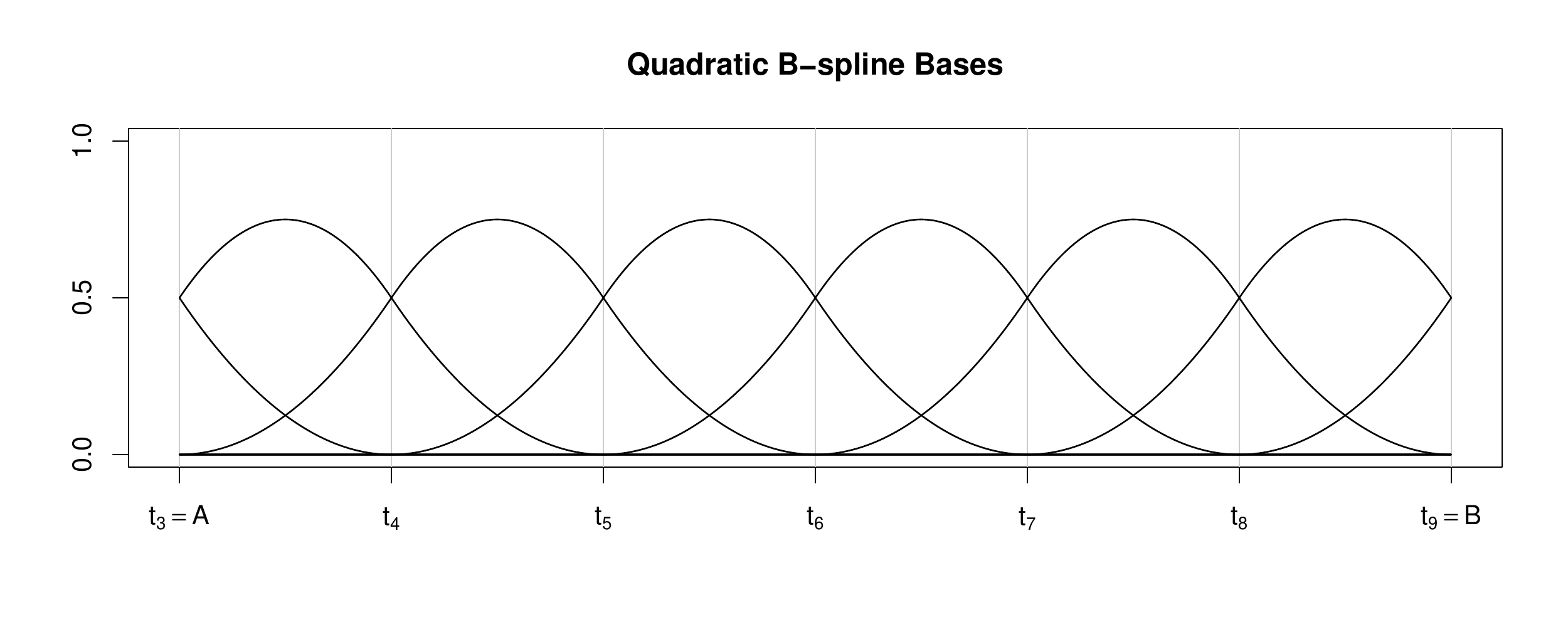}
\end{center}
\caption{\baselineskip=10pt Plot of 9 quadratic $(d=2)$ B-splines on $[A,B]$ defined using $11$ knot points that divide $[A,B]$ into $K=6$ equal subintervals.}
\label{fig: Quadratic B-splines}
\end{figure}

For episodic components, we thus needed models that can accommodate local variations in shape but would also allow the smoothness to be learned from regions where more informative data points are available. 
To achieve this, we employed flexible penalized normalized mixtures of B-splines with smoothness inducing priors on the coefficients 
to model the densities of the episodic components.
For the $\ell\th$ component, we partition the interval $[A_{\ell},B_{\ell}]$ of interest into $L_{\ell}$ subintervals using knot points
$A_{\ell} = t_{\ell,1}=\dots=t_{\ell,d+1} < t_{\ell,d+2} < t_{\ell,d+3} < \dots < t_{\ell,d+L_{k}} < t_{\ell,d+L_{\ell}+1} = \dots = t_{\ell,2d+L_{\ell}+1}=B_{\ell}$. 
Using these knot points, $J_{\ell}=(d+L_{\ell})$ B-spline bases of degree $d$, denoted by $\bB_{d,\ell,J_{\ell}}=\{b_{d,\ell,1},b_{d,\ell,2},\dots,b_{d,\ell,J_{\ell}}\}$, can be defined through a recursion relation \cite[page 90]{de_Boor:2000}. 
See Figure~\ref{fig: Quadratic B-splines} and Section \ref{sec: Quadratic B-splines} in the supplementary material. 
B-splines are nearly orthogonal and locally supported. 
For equidistant knot points with $\delta_{\ell} = (t_{\ell,J}-t_{\ell,J-1})$, the areas under these curves can be easily computed as 
\vspace{-4ex}\\
\bse
\delta_{\ell,j}=\int_{A_{\ell}}^{B_{\ell}} b_{2,\ell,j}(X)dX = \left\{\begin{array}{ll}
\delta_{\ell}/6 & \text{for}~j=1,J_{\ell},\\
5\delta_{\ell}/6 & \text{for}~j=2,J_{\ell}-1,\\
\delta_{\ell} & \text{for}~j=3,\dots,J_{\ell}-2.
\end{array}\right.
\ese
\vspace{-4ex}\\
Mixtures of B-splines can therefore be easily normalized. 
A flexible model for the density functions is then obtained as 
\vspace{-4ex}\\
\bse
&f_{X,\ell}(X_{\ell}) 
= \bB_{d,\ell,J_{\ell}}(X_{\ell}) \exp(\bxi_{\ell}) \left\{\sum_{m=1}^{J_{\ell}}\delta_{\ell,m}\exp(\xi_{\ell,m}) \right\}^{-1}, \label{eq: normalized mixture of Bsplines} \\
& (\bxi_{\ell}\mid J_{\ell}, \sigma_{\xi,\ell}^{2}) \propto (2\pi\sigma_{\xi,\ell}^{2})^{-J_{\ell}/2} \exp\{-\bxi_{\ell}\trans \bP_{\ell}\bxi_{\ell}/(2\sigma_{\xi,\ell}^{2})\},~~~~ \sigma_{\xi,\ell}^{2} \sim \IG(a_{\xi},b_{\xi}).
\ese
\vspace{-4ex}\\
Here $\bxi_{\ell} = \{\xi_{\ell,1},\xi_{\ell,2},\dots,\xi_{\ell,J_{\ell}}\}\trans$; 
$\exp(\bxi_{\ell}) = \{\exp(\xi_{\ell,1}), \exp(\xi_{\ell,2}),\dots,\exp(\xi_{\ell,J_{\ell}})\}\trans$;
and $\bP_{\ell}=\bD_{\ell}\trans \bD_{\ell}$, where $\bD_{\ell}$ is a $(J_{\ell}+2) \times J_{\ell}$ matrix such that $\bD_{\ell}\bxi_{\ell}$ computes the second order differences in $\bxi_{\ell}$.
The prior $p_{0}(\bxi_{\ell}\mid\sigma_{\xi,\ell}^{2})$ induces smoothness in the coefficients because it penalizes $\sum_{j=1}^{J_{\ell}}(\Delta^2 \xi_{\ell,j})^2 = \bxi_{\ell}\trans P_{\ell}\bxi_{\ell}$, the sum of
squares of the second order differences in $\bxi_{\ell}$ \citep{Eilers_Marx:1996}.
The parameters $\sigma_{\xi,\ell}^{2}$ play the role of smoothing parameters -
the smaller the value of $\sigma_{\xi,\ell}^{2}$, the stronger the penalty and the smoother the associated variance function. 
The inverse-Gamma hyper-priors on $\sigma_{\xi,\ell}^{2}$ allow the data to influence the posterior smoothness and make the approach data adaptive.
Importantly, the smoothness is now informed by data points across the entire range, resulting in vast improvements in the density estimates near the left boundaries.

For regularly consumed components with strictly positive recalls, we found mixtures of truncated normals to slightly outperform normalized mixtures of B-splines. 
This is also consistent with findings reported in \cite{Sarkar_etal:2014}. 
For regularly consumed components, we thus still use mixtures of truncated normals with shared atoms. 
With densities smoothed out to zeros at the boundaries, truncations are not strictly needed for regularly consumed dietary components. 
We still retain the truncations to make our approach broadly applicable to other potential applications where boundary discontinuities may be present even when the recalls are all continuous.  

Next, we consider the problem of modeling $\bR_{\bX}$. 
The problem of modeling correlation matrices has garnered some attention in the literature \citep{barnard2000modeling, liechty2004bayesian, pourahmadi2015distribution, tsay2017modelling}. 
Here, we adapt the model from \cite{Zhang2011b} based on spherical coordinate representation of Cholesky factorizations 
that allows the involved parameters to be treated separately of each other, simplifying posterior computation 
while guaranteeing the resulting matrix to always be a valid correction matrix. 
We prove in \ref{appendix: supplementary results} that the converse is also true. 
That is, any correlation matrix can be represented in this form which establishes its nonparametric nature. 
We drop the subscript $\bX$ for the rest of this subsection to keep the notation clean. 

Let $\bV^{(q+p)\times (q+p)}$ be a lower triangular matrix such that $\bR=\bV\bV\trans$. 
The form of $\bV$ is 
\vspace{-4ex}\\
\bse
\bV &=& 
\left(\begin{array}{c c c c}
v_{1,1} & 0 & \dots & 0\\
v_{2,1} & v_{2,2} & \dots & 0\\
\vdots & \vdots& \vdots& \vdots \\
v_{q+p,1} & v_{q+p,2} & \dots & v_{q+p,q+p}
\end{array} \right).
\ese
\vspace{-2ex}\\
We have $r_{\ell,\ell'} = \sum_{k=1}^{\ell}v_{\ell,k}v_{\ell',k}$ for all $\ell \leq \ell'$. 
The restriction that $\bR$ is a correlation matrix then implies $\sum_{k=1}^{\ell}v_{\ell,k}^{2} = 1$ for all $\ell=1,\dots,(q+p)$. 
The restrictions are satisfied by the following parameterization
\vspace{-4ex}\\
\bse
&& v_{1,1}=1, \\
&& v_{2,1}=b_{1}, ~ v_{2,2}=\sqrt{1-	b_{1}^{2}},\\
&& v_{3,1} =b_{2}\sin\theta_{1}, ~v_{3,2}=b_{2}\cos\theta_{1},~v_{3,3}=\sqrt{1-b_{2}^{2}},\\
&& v_{\ell,1}=b_{\ell-1}\sin\theta_{i_{1}(\ell)},\\
&& v_{\ell,k}=b_{\ell-1}\cos\theta_{i_{1}(\ell)}\cos\theta_{i_{1}(\ell)+1}\dots \cos\theta_{i_{1}(\ell)+k-2}\sin\theta_{i_{1}(\ell)+k-1}, \\
&&\hspace{9cm} \hbox{for}~k=2,\dots,(\ell-2),\\
&& v_{\ell,\ell-1}=b_{\ell-1}\cos\theta_{i_{1}(\ell)}\cos\theta_{i_{1}(\ell)+1}\dots \cos\theta_{i_{2}(\ell)-1}\cos\theta_{i_{2}(\ell)},~~~~v_{\ell,\ell}=\sqrt{1-b_{\ell-1}^{2}},
\ese
\vspace{-4ex}\\
where $\ell=4,\dots,(q+p)$,  
$i_{1}(\ell) = 1+\{1+\dots+(\ell-3)\} 
= (\ell^{2}-5\ell+8)/2$ and $i_{2}(\ell)=i_{1}(\ell)+(\ell-3) 
= (\ell^{2}-3\ell+2)/2$,  
$\abs{b_{t}}\leq 1$, $t=1,\dots,(q+p-1)$, $\abs{\theta_{s}}\leq \pi$, $s=1,\dots,i_{2}(q+p)$. 
The total number of parameters is $\{1+2+\dots+(q+p-1)\}=(q+p)(q+p-1)/2$.
We have $\abs{\bR}=\abs{\bV}^{2} = \prod_{\ell=2}^{q+p}v_{\ell,\ell}^{2} = \prod_{\ell=1}^{q+p-1}(1-b_{\ell}^{2})$. 
The model for $\bR$ is completed by assigning uniform priors on $b_{t}$'s and $\theta_{s}$'s
\vspace{-4ex}\\
\bse
b_{t} \sim \Unif(-1,1), ~~~~~ \theta_{s} \sim \Unif(-\pi,\pi). 
\ese
Here $\Unif(a,b)$ denotes a uniform distribution with support $(a,b)$.

\subsection{Modeling the Density $f_{\bU\mid \wt\bX}$}  \label{sec: mvt copula density of errors}
The reported intakes of the regularly consumed components exhibit strong conditional heteroscedasticity, so do the reported intakes of the episodic components, when consumed.
To accommodate conditional heteroscedasticity, we let 
\vspace{-4ex}\\
\bse
&\bU_{i,j} = \bS(\wt\bX_{i})\bepsilon_{i,j},~~~\hbox{with}~~~\eE(\bepsilon_{i,j}) = \bzero,\\
&\text{and}~~~\bS(\wt\bX_{i}) = \diag\{1,\dots,1,s_{q+1}(\wt{X}_{q+1,i}),\dots,s_{2q+p}(\wt{X}_{2q+p,i})\}.  	\label{eq: multiplicative structure}
\ese
\vspace{-4ex}\\
The above model implies that $\cov(\bU_{i,j}\mid \wt\bX_{i}) = \bS(\wt\bX_{i})~ \cov(\bepsilon_{i,j}) ~ \bS(\wt\bX_{i})$ and marginally $\var(U_{\ell,i,j}\mid \wt{X}_{\ell,i}) = s_{\ell}^{2}(\wt{X}_{\ell,i})\var(\epsilon_{\ell,i,j})$. 
Other features of the distribution of $\bU$ including its shape and correlation structure are derived from $f_{\bepsilon}$. 
The multiplicative structural assumption arises naturally for conditionally heteroscedastic multivariate measurement errors \citep{sarkar2018bayesian}. 
The model also automatically accommodates multiplicative measurement errors 
via a simple reformulation. 

As in Section \ref{sec: mvt copula density of interest}, we use a Gaussian copula density model to specify the density $f_{\bepsilon}$ 
but the model now has to satisfy mean zero constraints.
Specifically, we let   
\vspace{-4ex}\\
\bse
\textstyle f_{\bepsilon}(\bepsilon) =\prod_{\ell=1}^{q}f_{\epsilon,\ell}(\epsilon_{\ell}) \times  |\bR_{\bepsilon}|^{-\frac{1}{2}} \exp\left\{-\frac{1}{2}\bY_{\bepsilon}\trans(\bR_{\bepsilon}^{-1}- \bI_{p})\bY_{\bepsilon}\right\}  \prod_{\ell=q+1}^{2q+p}f_{\epsilon,\ell}(\epsilon_{\ell}),  \label{eq: mixture model for f_X}  \\
\text{subject to}~ \textstyle\eE_{f_{\epsilon,\ell}}(\epsilon_{\ell}) = 0, ~~~\hbox{for}~\ell=1,\dots,2q+p. 
\ese
\vspace{-4ex}\\
Here, $F_{\epsilon,\ell}(\epsilon_{\ell}) = \Phi(Y_{\epsilon,\ell})$ for all $\ell$. 
The first $q$ components of $\bepsilon$ are independent of each other and also independent of the rest of the $q+p$ components. 
The latter $q+p$ components may be correlated with correlation matrix $\bR_{\bepsilon}$.


The copula approach again allows us to use different models for the distributions of the pseudo-errors $f_{\epsilon,\ell}(\epsilon), \ell=1,\dots,q$, 
and the distributions of the actual scaled measurement errors $f_{\epsilon,\ell}(\epsilon), \ell=q+1,\dots,2q+p$.

Here, we only model the correlation between different scaled error components $\epsilon_{\ell,i,j},\epsilon_{\ell^{\prime},i,j}$ for $\ell \neq \ell^{\prime}$ 
but ignore the correlation between different sampling occasions $\epsilon_{\ell,i,j},\epsilon_{\ell,i,j^{\prime}}$ for $j \neq j^{\prime}$. 
The correlation between $W_{\ell,i,j}, W_{\ell,i,j^{\prime}}$ for $j \neq j^{\prime}$ is thus explained entirely by their shared component $\wt{X}_{\ell,i}$.   
In post model fit correlation analysis with estimated scaled `residuals', presented in Figure \ref{fig: EATS6 residual Tables} in the Supplementary Material, 
we found no real evidence 
that the errors $\epsilon_{\ell,i,j}, \epsilon_{\ell,i,j^{\prime}}$ are significantly correlated for $j \neq j^{\prime}$.

For $\ell=1,\dots,q$, we model the marginal densities $f_{\epsilon,\ell}$ as $f_{\epsilon,\ell}(\epsilon_{\ell}) = \Normal(\epsilon_{\ell} \mid 0,1)$. 
This implies a probit model for the probabilities of consumptions $P_{\ell}({X}_{\ell}) = \Pr\{U_{\ell} > -h_{\ell}(X_{\ell})\} = \Phi\{h (X_{\ell})\}$. 
Flexibility of this probability model thus depends on the choice of $h_{\ell}(X_{\ell})$. 
We discuss this issue in Section \ref{sec: mvt copula prob of consumption}. 

For $\ell=q+1,\dots,2q+p$, we model the marginal densities $f_{\epsilon,\ell}(\epsilon)$ 
using an adapation of the moment restricted model in \cite{Sarkar_etal:2014} but with shared atoms as 
\vspace{-4ex}\\
\bse
&& f_{\epsilon,\ell}(\epsilon_{\ell}) = \textstyle\sum_{k=1}^{K_{\epsilon}}\pi_{\epsilon,\ell,k}~f_{c\epsilon}(\epsilon_{\ell} \mid p_{\epsilon,k},\tmu_{\epsilon,k},\sigma_{\epsilon,k,1}^{2},\sigma_{\epsilon,k,2}^{2}), \label{eq: Model_distribution_of_errors}  ~~~~~  \bpi_{\epsilon,\ell} \sim \Dir(\alpha_{\epsilon,\ell}/K_{\epsilon},\dots,\alpha_{\epsilon,\ell}/K_{\epsilon}),\\
&& (p_{\epsilon,k},\tmu_{\epsilon,k},\sigma_{\epsilon,k,1}^{2},\sigma_{\epsilon,k,2}^{2}) \sim \hbox{Unif}(0,1)~\hbox{Normal}(0,\sigma_{\epsilon,\tmu}^{2})~\hbox{IG}(a_{\epsilon},b_{\epsilon})~\hbox{IG}(a_{\epsilon},b_{\epsilon}), ~~
\ese
\vspace{-4ex}\\
where $f_{c\epsilon}(\epsilon\mid p,\wt\mu,\sigma_{1}^{2},\sigma_{2}^{2}) = \{p~\Normal(\epsilon \mid \mu_{1},\sigma_{1}^{2})+(1-p)~\Normal(\epsilon \mid \mu_{2},\sigma_{2}^{2})\}$, with
$\mu_{1}=c_{1}\tmu,\mu_{2}=c_{2}\tmu$,
 $c_{1}=(1-p)/\{p^2+(1-p)^{2}\}^{1/2}$ and $c_{2} = -p/\{p^2+(1-p)^{2}\}^{1/2}$.
The zero mean constraint on the errors is satisfied, since $~p\mu_{1}+(1-p)\mu_{2} = \{pc_{1}+(1-p)c_{2}\}\tmu = 0$.
Normal densities are included as special cases with $(p,\tmu) = (0.5,0)$ or $(0,0)$ or $(1,0)$.
Symmetric component densities are included as special cases when $p=0.5$ or $\tmu=0$.
Specification of the prior for $f_{\epsilon}$ is completed assuming non-informative priors for $(p,\tmu,\sigma_{1}^{2},\sigma_{2}^{2})$.
Here $\hbox{Unif}(\ell,u)$ denotes a uniform distribution on the interval $[\ell,u]$.

As in the case of $\bR_{\bX}$, we assume $\bR_{\bepsilon}^{(q+p)\times(q+p)}=((r_{\bepsilon,\ell,\ell'}))=\bV_{\bepsilon}\bV_{\bepsilon}\trans$ and parameterize the elements of $\bV_{\bepsilon}$ using spherical coordinates. 
We assign uniform priors on $b_{\bepsilon,t}, t=1,\dots,(q+p-1)$ and $\theta_{\bepsilon,s}, s=1,\dots,i_{2}(q+p)$ 
\vspace{-4ex}\\
\bse
b_{\bepsilon,t} \sim \Unif(-1,1), ~~~~~ \theta_{\bepsilon,s} \sim \Unif(-\pi,\pi). 
\ese

Finally, for $\ell=q+1,\dots,2q+p$, we model the variance functions $v_{\ell}(\wt{X}_{\ell})=s_{\ell}^{2}(\wt{X}_{\ell})$ by flexible penalized mixtures of B-splines with smoothness inducing priors on the coefficients as in \cite{Staudenmayer_etal:2008} as 
\vspace{-4ex}\\
\bse
&& v_{\ell}(\wt{X}_{\ell}) = s_{\ell}^{2}(\wt{X}_{\ell}) = \textstyle\sum_{j=1}^{J_{\ell}} b_{d,\ell,j}(\wt{X}_{\ell}) \exp(\vartheta_{\ell,j}) = \bB_{d,\ell,J_{\ell}}(\wt{X}_{\ell}) \exp(\bvartheta_{\ell}), \label{eq: models for variance functions} \\
&& (\bvartheta_{\ell}\mid J_{\ell}, \sigma_{\vartheta,\ell}^{2}) \propto (2\pi\sigma_{\vartheta,\ell}^{2})^{-J_{\ell}/2} \exp\{-\bvartheta_{\ell}\trans \bP_{\ell}\bvartheta_{\ell}/(2\sigma_{\vartheta,\ell}^{2})\},~~~~ \sigma_{\vartheta,\ell}^{2} \sim \IG(a_{\vartheta},b_{\vartheta}).
\ese
\vspace{-4ex}\\
As before, the parameters $\sigma_{\vartheta,\ell}^{2}$ play the role of smoothing parameter, and the inverse-Gamma hyper-priors allow them to be learned from the data themselves.


\subsection{Modeling the Consumption Probabilities $P_{\ell}(X_{\ell})$} \label{sec: mvt copula prob of consumption}
We recall that, according to our model, the probability of reporting positive consumptions by an individual with long-term average intake $X_{\ell}$ is given by 
\vspace{-4ex}\\
\bse
P_{\ell}({X}_{\ell}) = \Pr\{U_{\ell} > -h_{\ell}(X_{\ell}) \mid X_{\ell}\} = \Phi\{h (X_{\ell})\}. 
\ese
\vspace{-4ex}\\
We model $h_{\ell}(X_{\ell})$ using flexible mixtures of B-splines again as 
\vspace{-4ex}\\
\bse
&& h_{\ell}({X}_{\ell}) = \textstyle\sum_{j=1}^{J_{\ell}} b_{d,\ell,j}({X}_{\ell}) \beta_{\ell,j} = \bB_{d,\ell,J_{\ell}}({X}_{\ell}) \bbeta_{\ell}, \\
&& (\bbeta_{\ell}\mid J_{\ell}, \sigma_{\beta,\ell}^{2},\bmu_{\beta,\ell},\bSigma_{\beta,\ell}) \propto (2\pi\sigma_{\beta,\ell}^{2})^{-J_{\ell}/2} \exp\{-\bbeta_{\ell}\trans \bP_{\ell}\bbeta_{\ell}/(2\sigma_{\beta,\ell}^{2})\}  ~ \MVN_{J_{\ell}}(\bbeta_{\ell} \mid \bmu_{\beta,\ell,0},\bSigma_{\beta,\ell,0}), \\  
&& \sigma_{\beta,\ell}^{2} \sim \IG(a_{\beta},b_{\beta}).
\ese
\vspace{-4ex}\\
The flexibility of $h_{\ell}(X_{\ell})$ compensates for the parametric nature of the probit link, making the model $P_{\ell}(X_{\ell})$ robust.

The right panels of Figure \ref{fig: EATS Milk & Whole Grains} suggest that  
as $X_{\ell}$ increases, 
the probability of reporting a positive consumption also increases on average. 
We model this flexibly as $\Phi\{h_{\ell}(X_{\ell})\}$. 
It is certainly possible that two individuals have (nearly) the same long-term average intakes, 
even though one of them consumes less often than the other but consumes larger amounts.
One could hope that additional subject-specific random effects terms would help capture this heterogeneity. 
It is, however, not clear that such models would be identifiable in the first place.
To see this, consider adding random effects $R_{\ell,i}$ to model (\ref{eq: episodic W2}). 
Letting $h_{\ell}(X_{\ell,i})=X_{\ell,i}$ for simplicity, 
we then obtain $W_{\ell,i,j} = X_{\ell,i}+R_{\ell,i}+U_{\ell,i,j}, \ell=1,\dots,q$.
With only the standard zero mean assumption on the distribution of the random effects, 
it is impossible to separately nonparametrically identify the distributions of $X_{\ell,i}$ and $R_{\ell,i}$ in this model.

\subsection{Modeling Energy-Adjusted Intakes} \label{sec: mvt copula density of energy adjusted intakes}
We now consider the problem of modeling the distribution of energy-adjusted long-term intakes. 
We now denote $\bX = (X_{1},\dots,X_{q+p})\trans = (X_{1},\dots,X_{J})\trans$ 
with $J=q+p$ and $X_{J}=X_{q+p}$ representing the energy intake. 
We are interested in the distribution of the intakes normalized by energy, that is, the distribution of $\bZ = (X_{1}/X_{J},\dots,X_{J-1}/X_{J})$.
The joint distribution of $\bZ$ is then straightforwardly obtained as 
\vspace{-4ex}\\
\bse
f_{\bZ}(\bZ) = \int X_{J}^{J} f_{\bX}(Z_{1} X_{J},\dots,Z_{J-1} X_{J}, X_{J}) dX_{J}.
\ese
\vspace{-4ex}\\
The marginal distribution of any $Z_{\ell}$ is likewise obtained as 
\vspace{-4ex}\\
\bse
f_{Z,\ell}(Z_{\ell}) = \int X_{J} f_{X_{\ell},X_{J}}(Z_{\ell} X_{J}, X_{J}) dX_{J}.
\ese
\vspace{-4ex}\\
These are integrals of single variables and can thus be easily numerically evaluated.

\subsection{Model Flexibility}
For most practical purposes, including our motivating applications, 
our models for the densities of interest $f_{X,\ell}$, the densities of the scaled errors $f_{\ell,\epsilon}$, the variance functions $s_{\ell}^{2}$, 
and the probabilities of consumptions $P_{\ell}(X_{\ell})$ are all highly flexible whenever sufficiently large numbers of B-spline bases and mixture components are allowed. 
Adapting similar results from \cite{sarkar2018bayesian}, formal statements and proofs establishing theoretical flexibility of these model components 
can be easily formulated using known results for B-splines and mixture models. 
Our model for the correlation matrices $\bR$ is also nonparametric. A formal proof is provided in the Appendix. 
The only real parametric component of our model is thus the Gaussian copula. 
Extending the model to other elliptical classes, like the multivariate t, would be conceptually straightforward. 
It is, however, often difficult to distinguish between such classes even in much simpler low dimensional measurement error free scenarios \citep{dos2008copula}. 
The problem only gets an order of magnitude more difficult when the variables whose densities are being modeled using copulas are all latent. 
Since the number of parameters in elliptical copulas increases only quadratically with dimension, they also scale well to higher dimensions. 
It is thus also not clear if other stylized copula classes could be any useful in nutritional epidemiology data sets like ours. 
Exploration of these issues will be pursued elsewhere.

\subsection{Model identifiability}
In the following, we investigate identifiability of our model. 
For notational simplicity, we drop the subscript $i$ and consider for  $j=1,\dots,m$, $\bY_{j}=(Y_{1,j},\dots,Y_{2q+p,j})\trans$, 
and similarly $\bW_{j}, \bU_{j}, \wt\bX$ and $\bX$. 
Then our proposed hierarchical model can be written as 
\vspace{-4ex}\\
\bse
\bY_{j} &=& \psi(\bW_{j}),~~~~~~\bW_{j} = \wt\bX  + \bU_{j},~~~~~~\eE(\bU_{j} \mid \wt\bX) = \bzero,~~~~~~\wt\bX = \phi(\bX),
\label{eq:overall_model}
\ese
\vspace{-4ex}\\
where the functions $\psi(\cdot): \rR^{2q+ p} \to \rR^{2q+ p}$ and $\phi(\cdot): \rR^{2q+ p} \to \rR^{2q+ p}$ 
are easily identified from models (\ref{eq: episodic Y2}) and (\ref{eq: episodic W2}). 
Specifically, $\phi(\cdot)$ is given by 
\vspace{-4ex}\\
\be
\wt{X}_{\ell} &=& h_{\ell}(X_{\ell}),~~~~~~~~~~~~~~~~~~~~~~~~~~~\hbox{for}~\ell=1,\dots,q, \nonumber\\
\wt{X}_{\ell} &=& X_{\ell-q}/P_{\ell-q}(X_{\ell-q}),~~~~~~~~~~~~~\hbox{for}~\ell=q+1,\dots,2q,   \label{eq: phi defn}\\
\wt{X}_{\ell} &=& X_{\ell-q},~~~~~~~~~~~~~~~~~~~~~~~~~~~~~\hbox{for}~\ell=2q+1,\dots,2q+p, \nonumber
\ee
\vspace{-4ex}\\
where, for $\ell = 1, \ldots, q$, $P_{\ell}(X_{\ell}) = P(W_{\ell,j} > 0 \vert X_{\ell}) = \Phi\{h_{\ell}({X}_{\ell})\}$ 
for some arbitrary functions $h_{\ell}(\cdot): \rR \to \rR$. 

We state the basic assumptions needed for identifiability and our main result on identifiability below. 
The proof is deferred to \ref{appendix: proof of identifiability}. 

\begin{Assmp}
{\bf (A1)} The number of replicates $m \geq 3$. 
{\bf (A2)}  $\bU_{j} \mid \wt\bX \stackrel{d}{=} \bS(\wt\bX)\bepsilon_{j},~\bepsilon_{j} \sim f_{\bepsilon},j=1,2,3$, where $f_{\bepsilon}$ has a Fourier transform that is non-vanishing everywhere. 
\end{Assmp}
Observe that {\bf (A2)} includes the homoscedastic case, that is, when $s_\ell(X_\ell)$ is a constant function of $X_\ell$.  
\begin{Thm} \label{thm: identifiability}
Under {(\bf A1)}-{(\bf A2)}, given the observed density $f_{\bY_{1}, \bY_{2}, \bY_{3}}$, the equation 
\vspace{-4ex}\\
\bse
f_{\bY_{1}, \bY_{2}, \bY_{3}}(\bY_{1}, \bY_{2}, \bY_{3}) = \int f_{\bY_{1}\mid \wt\bX}(\bY_{1}\mid \wt\bX) f_{\bY_{2}\mid \wt\bX}(\bY_{2}\mid \wt\bX) f_{\bY_{3}\mid \wt\bX}(\bY_{3}\mid \wt\bX) f_{\wt\bX}(\wt\bX) d\wt\bX
\ese
\vspace{-4ex}\\
admits a unique solution for $f_{\bY_{j}\mid \wt\bX}(\bY_{j}\mid \wt\bX)$  for $j=1,\ldots,3$ and $f_{\wt\bX}(\wt\bX)$. 
Furthermore, if $\bX$ and $\wt\bX$ are related by \eqref{eq: phi defn}, 
then $f_{\bX}(\bX)$ is uniquely identified from $f_{\bY_{j}\mid \wt\bX}(\bY_{j}\mid \wt\bX)$  for $j=1,\ldots,3$ and $f_{\wt\bX}(\wt\bX)$. 
\end{Thm}

In practice, for identifiability, we require $m_{i} \geq 3$ recalls for at least some values of $i$. 
As long as this condition is satisfied, missing values in recall data can be simply ignored.  
For our motivating EATS data set, we have $m_{i} = 4$ for all $i$ with no missing recalls. 
So the conditions are easily satisfied.

\section{Simulation Studies} \label{sec: simulation studies}

Our final model components described in Section \ref{sec: models} were decided 
after extensive numerical experiments with many different choices for these components and their many combinations 
to obtain the best empirical performances in a wide variety of scenarios. 
Such experiments included 
taking reflections of $X_{\ell,i}$'s to fix boundary issues; 
adaptations of the method of \cite{sarkar2018bayesian} to model the joint densities; 
mixtures of truncated normals as well as mixtures of half-normal distributions and their few variations for modeling the densities of episodic and regular components; 
mixtures of normalized B-splines, as originally proposed in \cite{Staudenmayer_etal:2008}, for modeling the densities of episodic components; 
mixtures of splines vs mixtures of truncated normals to model the densities of regular components; 
simple parametric as well as more flexible polynomial models for the functions $h_{\ell}(X_{\ell})$; 
these choices of $h_{\ell}(X_{\ell})$ with and without mixtures of mean restricted normals for modeling the distributions of the associated pseudo-errors etc. 
To keep things concise, we focus here on comparisons with our main competitor, the method of \cite{Zhang2011b}, only. 
Simulation scenarios to perform these comparisons were designed as follows. 

We chose $(q+p) = 3$ dimensional $\bX$ with 
(a) all regular ($q=0, p =3$),
(b) all episodic ($q=3, p =0$), and 
(c) mixed ($q=2, p =1$) components. 
Our proposed method scales very well to much higher dimensional problems, 
but with $3$ total components, 
the results can be conveniently graphically summarized. 
 
To generate the true $X_{\ell,i}$'s for $\ell=1,\dots,q+p$, we (a) first sampled $\bX_{i}^{\triangle} \sim \MVN_{p+q}(\bzero,\bR_{\bX})$, 
(b) then, set $\bX_{i}^{\triangle\triangle}=\Phi(\bX_{i}^{\triangle})$, (c) finally, set $X_{\ell,i} = F_{TN,mix}^{-1}(X_{\ell,i}^{\triangle\triangle} \mid \bpi_{X,\ell},\bmu_{X,\ell},\bsigma_{X,\ell}^{2}, X_{\ell,L},X_{\ell,U})$, 
where $F_{TN,mix}(X \mid \bpi,\bmu,\bsigma^{2},X_{L},X_{U}) = \sum_{k=1}^{K}\pi_{k}F_{TN}(X \mid \mu_{k},\sigma_{k}^{2},X_{L},X_{U})$. 
This way, the marginal distributions are mixtures of truncated normal distributions and hence can take widely varying shapes while the correlation between different components is $\bR_{\bX}$. 
See Figure \ref{fig: SimStudy1}. 
We set 
\vspace{-4ex}\\
\bse
& 
\hskip -0.5cm ~\bR_{\bX} = \left(\begin{array}{c c c}
1 & 0.7 & 0.7^2 \\
 & 1 & 0.7 \\
 &  & 1 
\end{array} \right),
~\bpi_{X,\ell} =  \left(\begin{array}{c}
0.25 \\
0.50 \\
0.25 
\end{array}\right) ~\text{for all}~\ell,
~\bmu_{\bX} =  \left(\begin{array}{c}
\bmu_{X,1}^{\trans} \\
\bmu_{X,2}^{\trans} \\
\bmu_{X,3}^{\trans}
\end{array} \right) 
= 
\left(\begin{array}{c c c}
-0.5 & 0.75 & 2 \\
0 & 3 & 0 \\
2 & 2 & 2 
\end{array} \right),\\
& X_{\ell,L}=0,~X_{\ell,U}=6 ~\text{for all}~\ell, 
~\text{and}~
~\sigma_{X,\ell,k}^{2}=0.75^2~\text{for all}~\ell,k. 
\ese
\vspace{-4ex}

We used a similar procedure to simulate the true scaled errors $\epsilon_{\ell,i,j}$'s, $\ell=q+1,\dots,2q+p$. 
We (a) first sampled $\bepsilon_{i,j}^{\triangle} \sim \MVN_{p+q}(\bzero,\bR_{\bepsilon})$, 
(b) then, set $\bepsilon_{i,j}^{\triangle\triangle}=\Phi(\bepsilon_{i}^{\triangle})$, (c) finally, set $\epsilon_{\ell,i,j} = F_{\epsilon,\ell,mix,scaled}^{-1}(\epsilon_{\ell,i,j}^{\triangle\triangle} \mid \bpi_{\epsilon,\ell}, \btheta_{\epsilon,\ell})$. 
Here, for $\ell=q+1,\dots,2q+p-1$, $F_{\epsilon,\ell,mix,scaled}$ is a scaled version of 
$F_{\epsilon,\ell,mix}(\epsilon \mid \bpi_{\epsilon,\ell}, \btheta_{\epsilon,\ell}) = \sum_{k=1}^{K_{\epsilon,\ell}}\pi_{\epsilon,\ell,k}F_{c\epsilon}(\epsilon \mid p_{\epsilon,\ell,k},\wt\mu_{\epsilon,\ell,k},\sigma_{\epsilon,\ell,k,1}^{2},\sigma_{\epsilon,\ell,k,2}^{2})$, scaled to have variance $1$, 
with $\btheta_{\epsilon,\ell}=\{(p_{\epsilon,\ell,k},\wt\mu_{\epsilon,\ell,k},\sigma_{\epsilon,\ell,k,1}^{2},\sigma_{\epsilon,\ell,k,2}^{2})\}_{k=1}^{K_{\epsilon,\ell}}$. 
And, for $\ell=2q+p$, $F_{\epsilon,\ell,mix,scaled}$ is a scaled version of $F_{\epsilon,\ell,mix}(\epsilon \mid \bpi_{\epsilon,\ell}, \btheta_{\epsilon,\ell}) = \sum_{k=1}^{K_{\epsilon,\ell}}\pi_{\epsilon,\ell,k}F_{\tiny{\Laplace}}(\epsilon \mid m_{\epsilon,\ell,k}, b_{\epsilon,\ell,k})$ 
with $\btheta_{\epsilon,\ell}=\{(m_{\epsilon,\ell,k},b_{\epsilon,\ell,k})\}_{k=1}^{K_{\epsilon,\ell}}$, 
adjusted to have mean zero and variance $1$. 
See Figure \ref{fig: SimStudy1}. 
In this case, we set
\vspace{-4ex}\\
\bse
& 
~\bR_{\bepsilon} = \left(\begin{array}{c c c}
1 & 0.5 & 0.5^2 \\
 & 1 & 0.5 \\
 &  & 1 
\end{array} \right),
~\bpi_{\epsilon,\ell} =  \left(\begin{array}{c}
0.25 \\
0.50 \\
0.25 
\end{array}\right) ~\text{for all}~\ell,\\ 
&~\btheta_{\bepsilon} =  \left(\begin{array}{c}
\btheta_{\epsilon,1}^{\trans} \\
\btheta_{\epsilon,2}^{\trans} \\
\btheta_{\epsilon,3}^{\trans}
\end{array} \right) 
= 
\left(\begin{array}{c c c}
(0.4,2,2,1) & (0.4,2,2,1) & (0.4,2,2,1) \\
(0.5,0,0.25,0.25) & (0.5,0,0.25,0.25) & (0.5,0,5,5) \\
(0,2) & (0,2) & (0,2) 
\end{array} \right).
\ese
The representations with $K_{\epsilon,\ell}=3$ components above are more than what are really needed to describe the particular assumed truths - 
we are effectively using a single component mixture of two-component scaled normals for $f_{\epsilon,q+1}$ producing a bimodal error distribution, 
a two component $(0.75,0.25)$ mixture of two-component scaled normals for $f_{\epsilon,q+2}$ producing a unimodal but heavier tailed error distribution, and finally 
a single component scaled Laplace for $f_{\epsilon,q+3}$ producing a unimodal ordinary smooth error distribution. 
See Figure \ref{fig: SimStudy1}.
As is, however, clear from the figure, such 3-component models are capable of generating a very wide variety of shapes, including multimodality heavy-tails etc., for the error distributions. 
We used such representations to perform small scale simulations to check our model's flexibility and efficiency. 
The results, not presented here for brevity, were comparable to the ones reported for the aforementioned choices. 


Combining the values of $\wt{X}_{\ell,i}=X_{\ell-q,i}$ and $\epsilon_{\ell,i,j}$ for $\ell=q+1,\dots,2q+p$ as generated above, 
we then simulated $W_{\ell,i,j}$ as $W_{\ell,i,j}=X_{\ell,i}+U_{\ell,i,j}$ where $U_{\ell,i,j}=s_{\ell}(\wt{X}_{\ell,i})\epsilon_{\ell,i,j}$ with $s_{\ell}(\wt{X}_{\ell}) = \wt{X}_{\ell}/3$ for each $\ell$.

To obtain zero consumption reportings for the episodic variables, we next simulated $U_{\ell,i,j} = \epsilon_{\ell,i,j}$'s for $\ell=1,\dots,q$ from $\MVN_{q}(\bzero,\bI_{q})$ and then 
set $W_{\ell,i,j}=\gamma_{\ell,0}+\gamma_{\ell,1}\newlog(X_{\ell,i})+U_{\ell,i,j}$ with $\gamma_{1,0}=1.5$ and $\gamma_{\ell,0}=1$ for $\ell=2,\dots,q$ and $\gamma_{\ell,1}=1$ for all $\ell=1,\dots,q$, where the function $\newlog$ is obtained by a Taylor series expansion of the natural $\log$ function up to the fourth order. 

Finally, we generated the `observed' data $\bY_{i,j}$ as 
$Y_{\ell,i,j} = \Ind(W_{\ell,i,j}>0)$ for $\ell=1,\dots,q$;
$Y_{\ell,i,j} = Y_{\ell-q,i,j} W_{\ell,i,j}$ for $\ell=q+1,\dots,2q$; and 
$Y_{\ell,i,j} = W_{\ell,i,j}$ for $\ell=2q+1,\dots,2q+p$.
This resulted in approximately {$20\%$, $35\%$ and $17\%$} zero recalls, respectively, when these components are designed to be episodic.

\vspace{0.25cm}
\begin{table}[!ht]
\begin{center}
\footnotesize
\begin{tabular}{|c|c|c|c|c|c|}
\hline
\multirow{2}{24pt}{No of ECs} &  \multirow{2}{24pt}{No of RCs} 	& \multirow{2}{28pt}{Sample Size}	& \multicolumn{3}{|c|}{Median ISE $\times 1000$} 	\\ \cline{4-6}
					&						& 		& {Sarkar, et al. (2018)}  		&{Zhang, et al. (2011)} 		& Our Method	\\ \hline \hline
\multirow{2}{*}{0}		& \multirow{2}{*}{3}			& 500	& 48.52 (35.90, 46.43, 1.62)	& 97.00 (16.13, 40.19, 2.47)	& 4.79 (1.52, 3.31, 0.47)  	\\
					&						& 1000	& 38.50 (29.31, 32.50, 1.35)	& 96.04 (14.82, 35.67, 2.87)	& 2.75 (0.73, 1.76, 0.40)	\\\cline{1-6}
\multirow{2}{*}{2}		& \multirow{2}{*}{1}			& 500	& $\times$				& 96.65 (17.67, 50.09, 3.68)	& 5.83 (1.83, 7.06, 0.46)	\\
					&						& 1000	& $\times$				& 96.02 (17.06, 48.39, 3.14)	& 2.79 (1.01, 3.06, 0.43)	\\\cline{1-6}
\multirow{2}{*}{3}		& \multirow{2}{*}{0}			& 500	& $\times$ 				& 98.11 (17.63, 49.41, 4.49)	& 13.04 (1.78, 7.05, 13.03)	\\
					&						& 1000	& $\times$				& 97.13 (17.19, 44.23, 3.22)	& 8.12 (1.63, 4.05, 6.28)	\\\cline{1-6}
\hline
\end{tabular}
\caption{\baselineskip=10pt 
Median integrated squared error (MISE) performance 
of density deconvolution models described in Section \ref{sec: models} of this article 
compared with the methods of \cite{sarkar2018bayesian} and \cite{Zhang2011b}. 
See Section \ref{sec: simulation studies} for additional details.
Here, EC and RC are abbreviations for episodic and regular components, respectively.
We have reported here the MISEs for estimating the three-dimensional joint densities as well as the three univariate marginals (in parenthesis). 
}
\label{tab: MISEs 1}
\end{center}
\end{table}
\vspace{-10pt}

The integrated squared error (ISE) of estimation of $f_{\bX}$ by $\wh{f}_{\bX}$ is defined as $ISE = \int \{f_{\bX}(\bX)-\widehat{f}_{\bX}(\bX)\}^{2}d\bX$.
A Monte Carlo estimate of ISE is given by
$ISE_{est} = \sum_{m=1}^{M}\{f_{\bX}(\bX_{m})-\widehat{f}_{\bX}(\bX_{m})\}^{2}/p_{0}(\bX_{m})$,
where $\{\bX_{m}\}_{m=1}^{M}$ are random samples from the density $p_{0}$.
%
We used the true densities $f_{\bX}$ for $p_{0}$ and the true values of the $\bX_{i}$'s for the $\bX_{m}$'s. 
Table \ref{tab: MISEs 1} reports the median ISEs (MISEs) for estimating the trivariate joint densities and the univariate marginals  
obtained by our method, compared with the method of \cite{sarkar2018bayesian} and the method of \cite{Zhang2011b}. 
The MISEs reported here are all based on $B=100$ simulated data sets.
As Table \ref{tab: MISEs 1} shows, our method vastly outperforms both methods in all cases. 

The multivariate density deconvolution method of \cite{sarkar2018bayesian} can only handle strictly continuous proxies for the latent $\bX$. 
The method is thus not applicable to episodic components with exact zero recalls. 
Simulations for this method are thus also restricted to cases where the components are all regularly consumed. 
The method of \cite{sarkar2018bayesian} derives the marginals from mixture models for joint distributions. 
Even with all continuous recalls, such a strategy is insufficiently flexible when the marginals have widely varying shapes as in our simulation scenarios. 
  
The method of \cite{Zhang2011b} accommodates exact zero recalls but, 
as discussed in the introduction and detailed in Section \ref{sec: mvt copula comparison with NCI} in the Supplementary Material, 
makes many restrictive and unrealistic model assumptions, resulting in highly inefficient density estimates. 
Interestingly, the MISEs for the method of \cite{Zhang2011b} remained practically unchanged even when the sample sizes were doubled. 
The MISEs for the method of \cite{Zhang2011b} mainly comprise the bias resulting from their highly restrictive model assumptions. 
As was also noted in \cite{Sarkar_etal:2014}, the bias in the estimates produced by restrictive deconvolution methods 
often actually increase, sometimes quite significantly, with an increase in the sample size
as more data points not conforming to the model assumptions are included in the analysis.  

Figures \ref{fig: SimStudy1}, \ref{fig: SimStudy2}, \ref{fig: SimStudy3} in the main paper and Figures \ref{fig: SimStudy4}, \ref{fig: SimStudy5} in the supplementary materials 
show various estimates obtained by our method and the method of \cite{Zhang2011b} 
for the $2$ episodic and $1$ regular component case for the data sets that produced the $25$ percentile ISEs for these models. 
Figures \ref{fig: SimStudy1} shows the true and estimated univariate marginal densities. 
The estimates produced by the method of \cite{Zhang2011b} capture the general overall shapes of the true densities but are clearly far from the truths. 
In particular, they decay and dip near zero, especially markedly in case of the first episodic component. 
Although they do not smooth out to zero but show discontinuities near zero, 
the model did not actually capture these discontinuities - they are just artifacts of our final adjustments to restrict their supports to $\rR^{+}$. 
The estimates obtained by our method, on the other hand, provide excellent fits to the truths. 
Figure \ref{fig: SimStudy2} shows the true and the estimated probabilities of consumptions. 
The close agreement between the true and the estimated probability curves is remarkable especially in light of the fact that 
the surrogates $W_{\ell,i,j}$, introduced to model these probabilities as well as the associated predictor values $X_{\ell,i}$ were all latent for $\ell=1,\dots,q$. 
Figure \ref{fig: SimStudy3} shows the true and the estimated univariate marginals of normalized intakes of the first two episodic components normalized by the regular component. 
Finally, Figures \ref{fig: SimStudy4} and \ref{fig: SimStudy5} presented in the supplementary material 
show the true and estimated bivariate marginals produced by the two methods. 
Our method has again produced excellent estimates of the true bivariate marginals, whereas the estimates produced by the method of \cite{Zhang2011b} are much poorer in comparison.

\vspace{12pt} 
\begin{figure}[!ht]
\centering
\includegraphics[height=15cm, width=16cm, trim=0cm 0cm 0cm 0cm, clip=true]{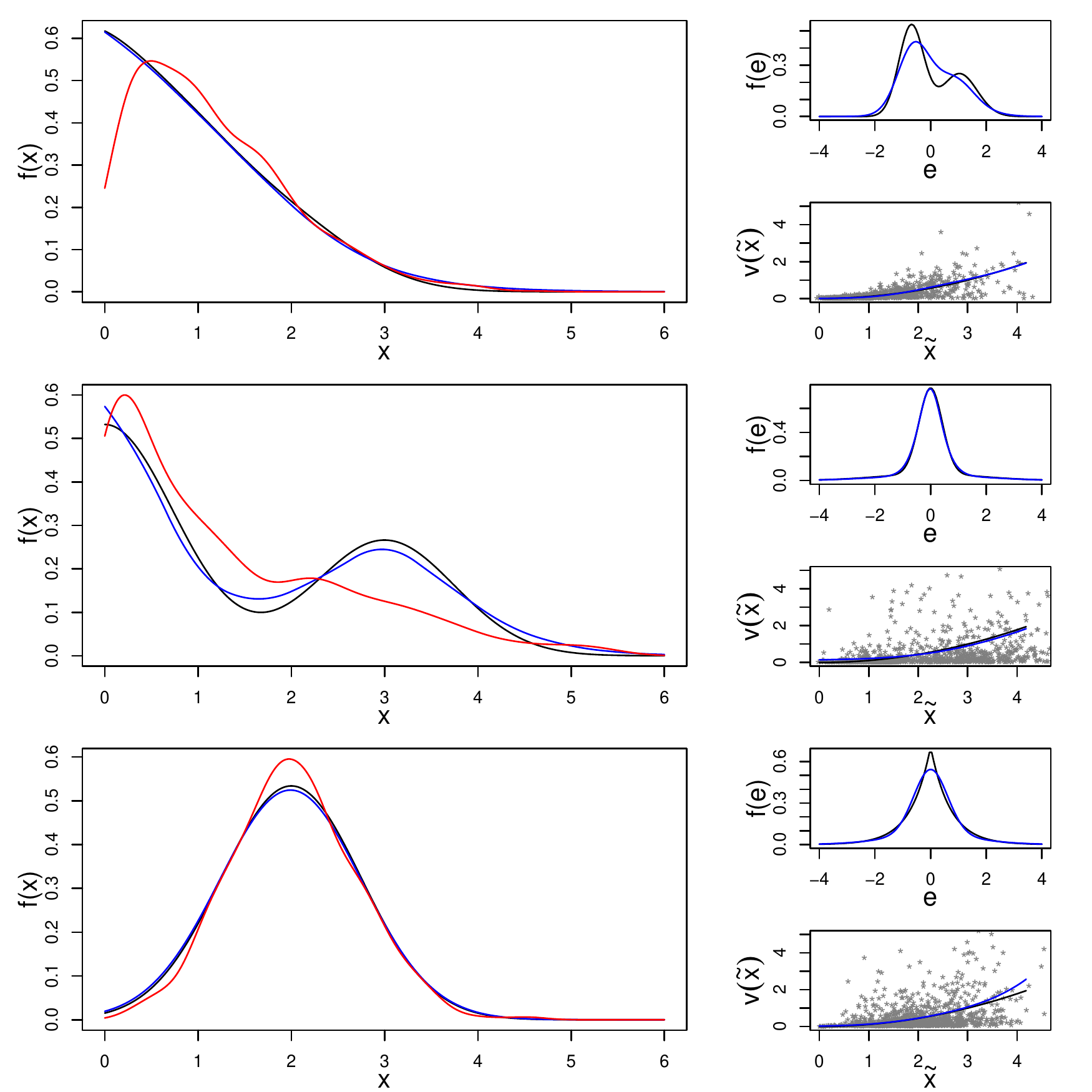}
\caption{\baselineskip=10pt 
Results for simulated data sets with sample size $n = 1000$, $q=2$ episodic components and $p=1$ regular components, each subject having $m_{i}=3$ replicates, for the data sets corresponding to the 25th percentile $3$-dimensional ISEs. 
From top to bottom, the left panels show the estimated densities $f_{X,\ell}(X_{\ell})$ of the two episodic components and the one regular component, respectively, 
obtained by our method (in blue) and the method of \cite{Zhang2011b} (in red).  
The right panels show the estimated distributions of the scaled errors $f_{\epsilon,q+\ell}(\epsilon_{q+\ell})$ and the estimated variance functions $v_{\ell}(\wt{X}_{\ell}) = s_{\ell}^{2}(\wt{X}_{\ell})$, estimated by our method. 
In all panels, the black lines represent the truth.
}
\label{fig: SimStudy1}
\end{figure}

\begin{figure}[!ht]
\centering
\includegraphics[height=6cm, width=16cm, trim=0cm 0cm 0cm 0cm, clip=true]{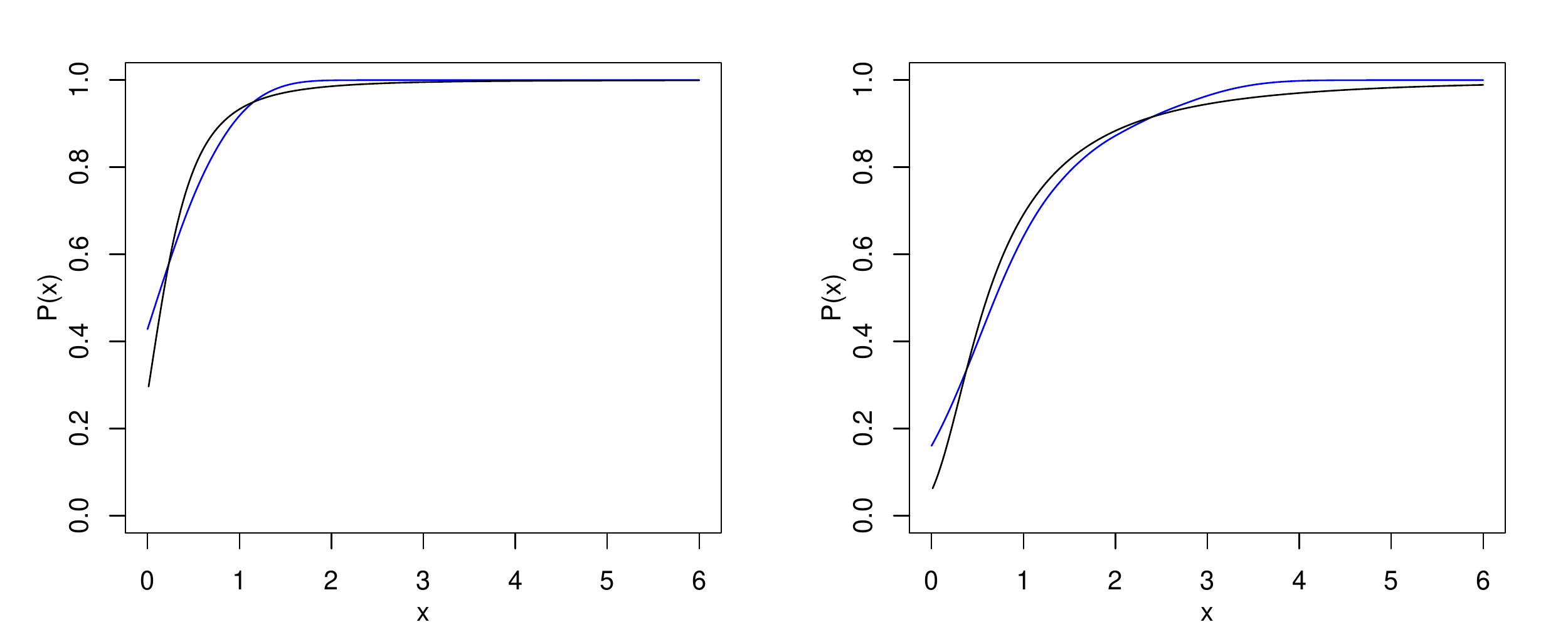}
\caption{\baselineskip=10pt 
Results for simulated data sets with sample size $n = 1000$, $q=2$ episodic components and $p=1$ regular components, each subject having $m_{i}=3$ replicates, for the data set corresponding to the 25th percentile $3$-dimensional ISE. 
The estimated (in blue) probabilities of reporting positive consumption $P_{\ell}(X_{\ell})$ for the two episodic components, estimated by our method.
In all panels, the black lines represent the truth.
}
\label{fig: SimStudy2}
\end{figure}

\begin{figure}[!ht]
\centering
\includegraphics[height=5cm, width=16cm, trim=0cm 0cm 0cm 0cm, clip=true]{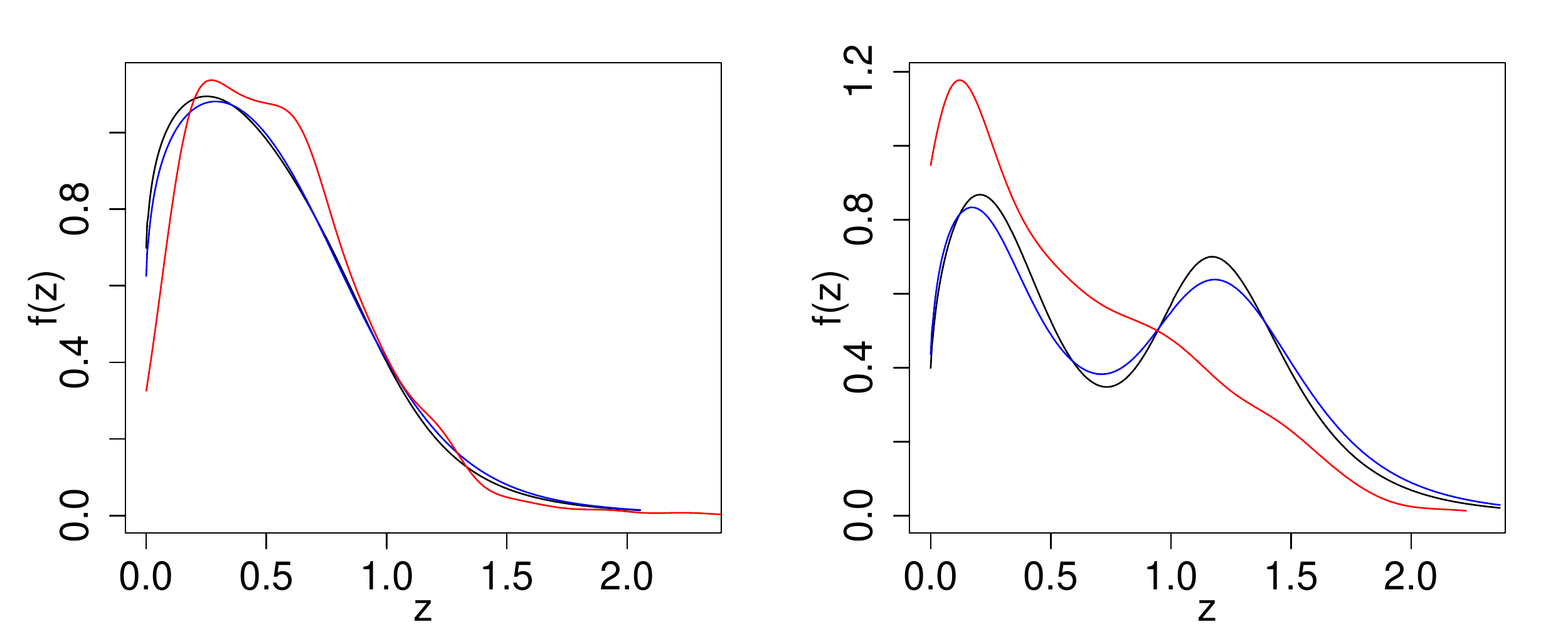}
\caption{\baselineskip=10pt 
Results for simulated data sets with sample size $n = 1000$, $q=2$ episodic components and $p=1$ regular components, each subject having $m_{i}=3$ replicates, for the data sets corresponding to the 25th percentile $3$-dimensional ISEs. 
The estimated distributions of the two episodic components, normalized by the regular component, 
estimated by our method (in blue) and by the method of \cite{Zhang2011b} (in red).
In all panels, the black lines represent the truth.
}
\label{fig: SimStudy3}
\end{figure}

As we have seen from Figure \ref{fig: EATS Milk & Whole Grains}, in real data sets the distributions of episodically consumed components are typically extremely right-skewed with discontinuities at zero. 
The case when all the components are designed to be episodic, including the third variable whose distribution is symmetric unimodal, 
was thus rather artificial 
but is still helpful in providing some insight.   
The marginal ISEs for the third component in this case were consistently significantly larger than the ISEs for the first two components for our method 
even though the shape of its distribution was much simpler compared to the extreme right skewed distributions of the first two components. 
This can be attributed to the fact that, since the third component also had high probabilities of reporting non-consumptions near the left boundary but the center of the distribution was away from the left boundary, 
almost all its recalls for small true intakes near the left boundary were zero recalls. 
This made estimating the relatively simple third distribution and the associated variance function more difficult than estimating these functions for the first two components 
which, even with similar probabilities of reporting non-consumptions, still had a good number of no-zero recalls available near the left boundaries. 
This is the only case when \cite{Zhang2011b} outperformed us, taking advantage of the simple unimodal bell shape of the true distribution which conforms closely to the method's parametric assumptions. 
However, as also discussed in the beginning of this paragraph, this is a highly unrealistic case. 
In practice, the true distributions of episodic components are never unimodal bell-shaped but are always reflected J-shaped, 
in which case our method vastly dominates. 
%
The univariate and the multivariate density estimates obtained by method of \cite{Zhang2011b} are based on the estimated values of $X_{\ell,i}$'s and $\bX_{i}$'s, respectively, 
but are otherwise not related but independently derived. 
The univariate details thus get masked in the three-dimensional estimates which, being based on single component multivariate normal models, remain relatively stable in all cases 
even though they are consistently heavily biased. 
This example illustrates the importance of assessing the estimation of the univariate marginals in multivariate deconvolution problems, 
reiterating the suitability of copula based approaches in applications like ours where the univariate marginals could be widely different.

Additional small scale simulations, where we closely mimicked the parametric assumptions of \cite{Zhang2011b}, 
are presented in Section \ref{sec: additional sims} in the Supplementary Material.

\section{Applications in Nutritional Epidemiology} \label{sec: applications}

In this section, we discuss the results of our method applied to the EATS data set. 
Specifically, we consider the problem of estimating the distributions of long-term average daily intakes of two episodic components - milk and whole grains, 
and two regular components -  sodium and energy. 
The surrogates for milk and whole grains, we recall, had approximately $21\%$ and $37\%$ exact zeros.

\begin{figure}[!ht]
\centering
\includegraphics[height=18.5cm, width=16cm, trim=0cm 0cm 0cm 0cm, clip=true]{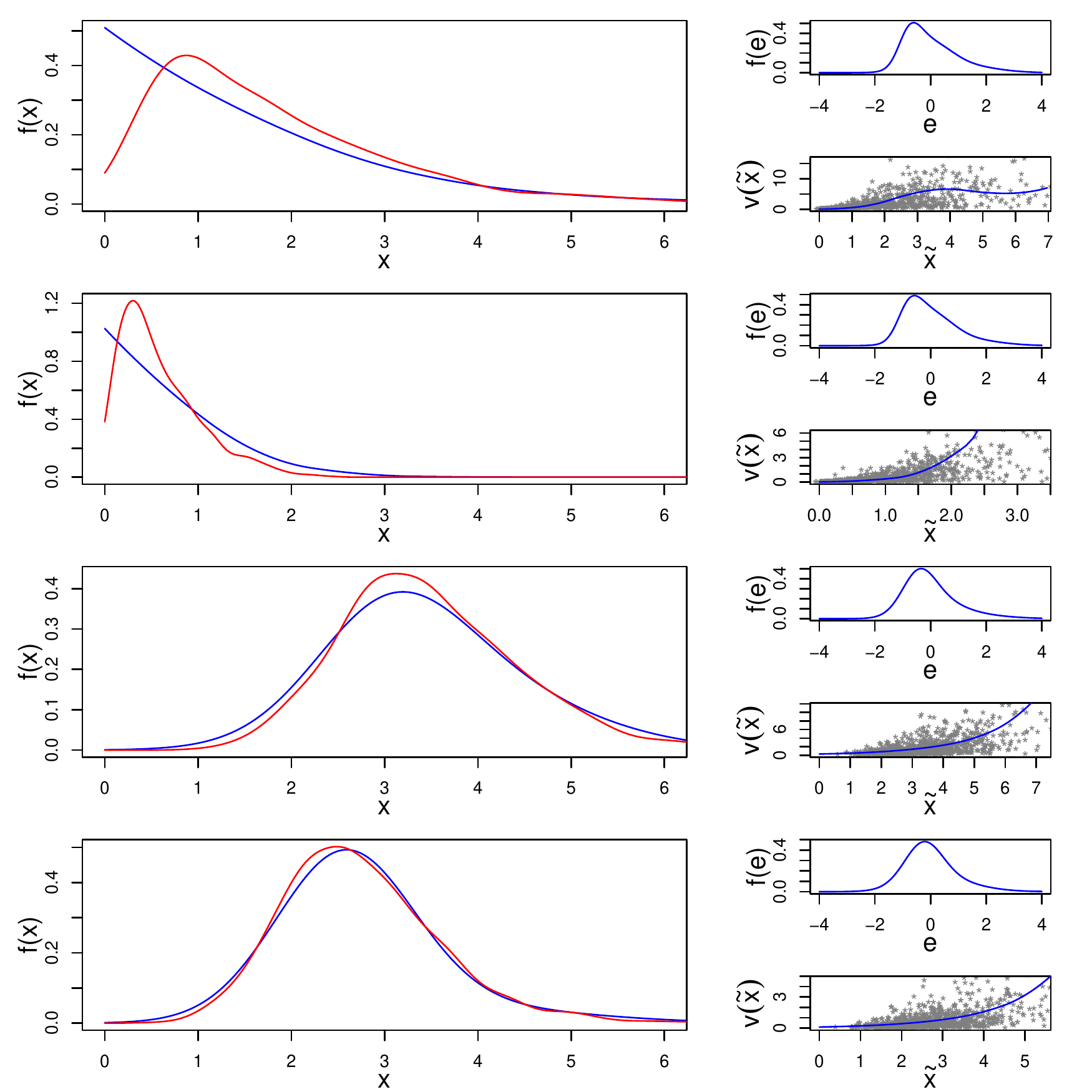}
\caption{\baselineskip=10pt 
Results for the EATS data sets with sample size $n = 965$, $q=2$ episodic components, milk and whole grains, and $p=2$ regular components, sodium and energy, each subject having $m_{i}=4$ replicates. 
From top to bottom, the left panels show the estimated densities $f_{X,\ell}(X_{\ell})$ of milk and whole grains, sodium, and energy, respectively, 
obtained by our method (in blue) and the method of \cite{Zhang2011b} (in red).  
The right panels show the associated distributions of the scaled errors $f_{\epsilon,q+\ell}(\epsilon_{q+\ell})$ and the associated variance functions $v_{\ell}(\wt{X}_{\ell}) = s_{\ell}^{2}(\wt{X}_{\ell})$, estimated by our method. 
}
\label{fig: EATS1}
\end{figure}

\begin{figure}[!ht]
\centering
\includegraphics[height=6cm, width=16cm, trim=0cm 0cm 0cm 0cm, clip=true]{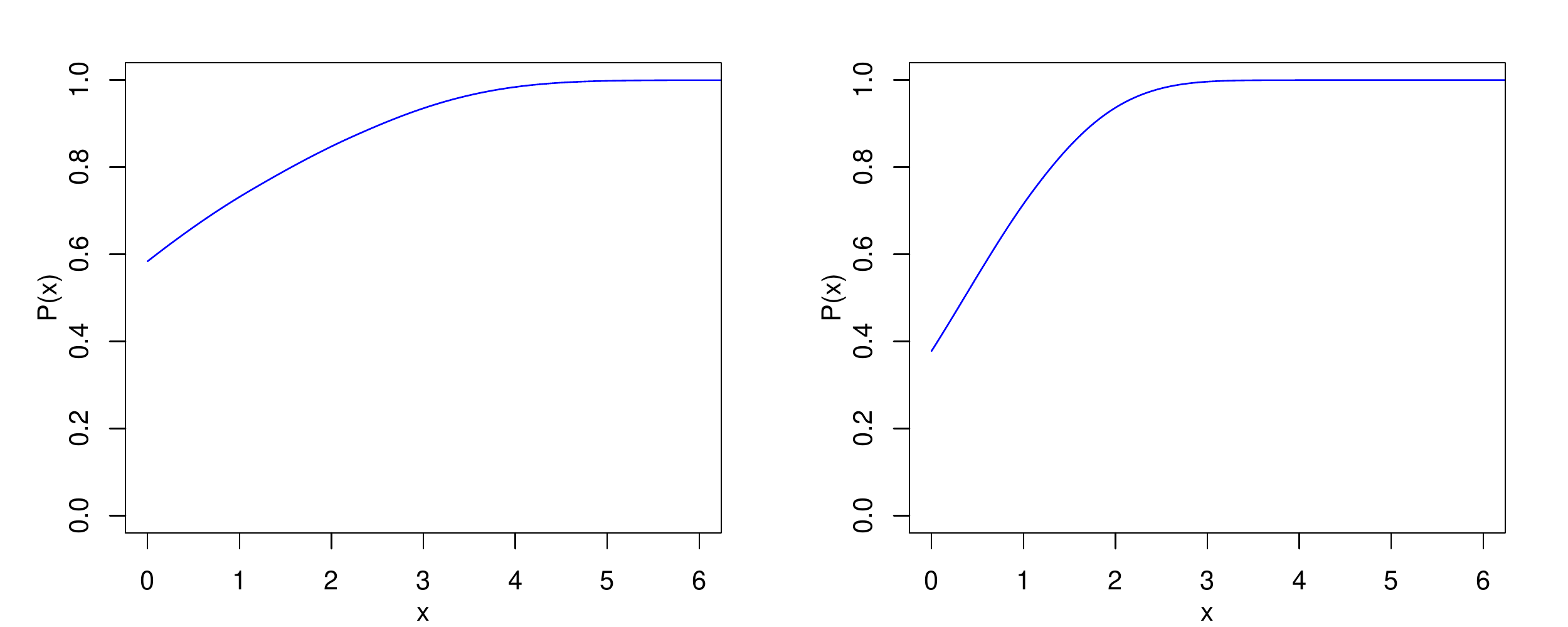}
\caption{\baselineskip=10pt 
Results for the EATS data sets with sample size $n = 965$, $q=2$ episodic components, milk and whole grains, and $p=2$ regular components, sodium and energy, each subject having $m_{i}=4$ replicates. 
The estimated probabilities of reporting positive consumption $P_{\ell}(X_{\ell})$ for the episodic components milk (left panel) and whole grains (right panel), estimated by our method.
}
\label{fig: EATS2}
\end{figure}

\begin{figure}[!ht]
\centering
\includegraphics[height=5cm, width=16cm, trim=0cm 0cm 0cm 0cm, clip=true]{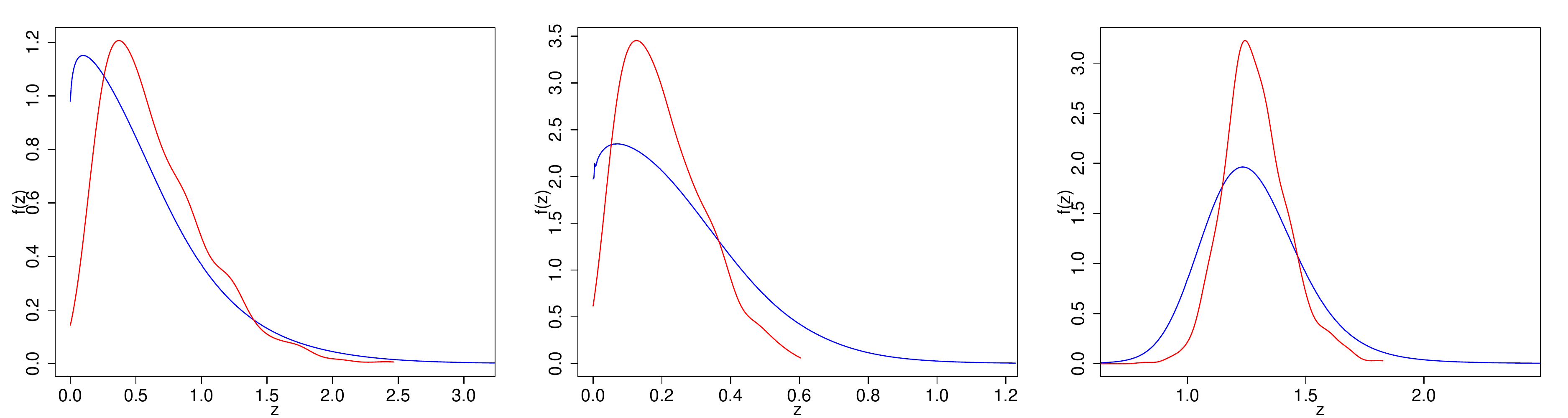}
\caption{\baselineskip=10pt 
Results for the EATS data sets with sample size $n = 965$, $q=2$ episodic components, milk and whole grains, and $p=2$ regular components, sodium and energy, each subject having $m_{i}=4$ replicates. 
From left to right, the estimated distributions of normalized intakes of milk, whole grains and sodium, normalized by total energy, 
estimated by our method (in blue) and by the method of \cite{Zhang2011b} (in red).
}
\label{fig: EATS3}
\end{figure}

Figure \ref{fig: EATS1} shows the estimated marginal densities $f_{X,\ell}$ obtained by our method and the method of \cite{Zhang2011b}. 
For sodium and energy, there is general agreement between the estimates obtained by our method and the method of \cite{Zhang2011b}. 
For the episodic components milk and whole grains, on the other hand, the estimated densities look very different, especially near the left boundary. 
Our method shows these densities to continually increase as we approach zero from right, as is expected from Figure \ref{fig: EATS Milk & Whole Grains}. 
Consistent with Figure \ref{fig: EATS Milk & Whole Grains}, compared to milk, the distribution of whole grains is also more concentrated near zero.
The estimates produced by \cite{Zhang2011b}, on the other hand, dip near zero, as was also observed in simulation scenarios. 

The right panels in Figure \ref{fig: EATS1} show the estimates of the densities of scaled measurement errors $f_{\epsilon,q+\ell}(\epsilon_{q+\ell})$ and the estimates of the variance functions $s_{\ell}^{2}(\wt{X}_{\ell})$. 
The estimated $f_{\epsilon,q+\ell}$'s are positively skewed for all components. 
And, as expected from Figures \ref{fig: EATS Sodium & Energy} and \ref{fig: EATS Milk & Whole Grains}, 
the estimated $s_{\ell}^{2}$'s show strong patterns of conditional heteroscedasticity for all components. 
For the episodic components, our method also provides estimates of the probabilities of reporting positive consumptions which are shown in Figure \ref{fig: EATS2}. 
The recalls for whole grains have more zeros than the recalls for milk. 
Its distribution is also more concentrated near zero. 
The probability of reporting positive consumptions for whole grains thus increases more rapidly as its true daily average intake increases. 

Figure \ref{fig: EATS3} shows the distributions of normalized intakes obtained by our method and the method of \cite{Zhang2011b}. 
The estimates look very different, including the one for the regular component sodium. 
Our method provides more realistic estimates of the distribution of normalized intakes that are more concentrated near zero but are more widely spread. 

Figures \ref{fig: EATS4} in the supplementary material 
shows the estimated bivariate marginals for produced by our method and the method of \cite{Zhang2011b}. 
Figure \ref{fig: EATS5 Z Tables} in the Supplementary Material additionally illustrates how the redundant mixture components become empty after reaching steady states in our MCMC based implementation. 
Figure \ref{fig: EATS5 Z Tables}  also shows how in practice the mixture component specific parameters get shared across different dimensions in our models with shared parameters for the marginal densities.

\section{Discussion} \label{sec: discussion}
\hspace*{7mm}{\bf Summary:}
In this article, we considered the problem of multivariate density deconvolution 
when replicated proxies are available but, complicating the challenges, the proxies also include exact zeros for some of the components. 
The problem is important in nutritional epidemiology for estimating long-term intakes of episodically consumed dietary components. 
We developed a novel copula based deconvolution approach that focuses on the marginals first and then models the dependence among the components to build the joint densities, 
allowing us to adopt different modeling strategies for different marginal distributions which proved crucial in accommodating important features of our motivating data sets. 
In contrast to previous approaches of modeling episodically consumed dietary components, 
our novel Bayesian hierarchical modeling framework allows us to model the distributions of interest more directly, 
resulting in vast improvements in empirical performances 
while also providing estimates of quantities of secondary interest, 
including probabilities of reporting non-consumptions, measurement errors' conditional variability etc.  

{\bf Other potential applications:}
Applications of the multivariate deconvolution approach developed here are not limited to zero-inflated data only but also naturally include data with strictly continuous recalls, as was shown in the simulations.
Advanced multivariate deconvolution methods are also needed to correct for measurement errors in regression settings when multiple error contaminated predictors are needed to be included in the model. 

{\bf Methodological extensions:}
Other methodological extensions and subjects of ongoing research include 
inclusion of associated exactly measured covariates like age, sex etc. that can potentially influence the consumption patterns, 
establishing theoretical convergence guarantees for the posterior,  
accommodation of dietary components which, unlike regular or episodic components, are never consumed by a percentage of the population, 
accommodation of subject specific survey weights,  
exploration of non-Gaussian copula classes, 
inclusion of additional information provided by food frequency questionnaires 
etc.

{\bf HEI index:}
Aside being of independent interest, episodic dietary components also contribute to 
the Healthy Eating Index (HEI, https://www.cnpp.usda.gov/healthyeatingindex), a performance measure developed by the US Department of Agriculture (USDA) to assess and promote healthy diets 
\citep{Guenther2008a,krebs2018update}. 		
The index is based on $13$ energy adjusted dietary components, as many as $8$ of which are episodic, 
and is currently calculated using the NCI method discussed in Section \ref{sec: mvt copula comparison with NCI}. 
The methodology developed in this article provides a much more sophisticated framework for modeling the HEI index 
and makes up an important component of our ongoing research.

\baselineskip=17pt
\section*{Supplementary Material}
The supplementary material presents a brief review of copula and explicit formula of quadratic B-splines for easy reference. 
The supplementary material also provides a detailed comparison of our method with previous approaches to zero-inflated data. 
The supplementary material additionally details the choice of hyper-parameters and the MCMC algorithm used to sample from the posterior, 
presents some additional figures, and the results of some additional numerical experiments. 
R programs implementing the deconvolution methods developed in this article are included in the supplementary material. 
The EATS data analyzed in Section \ref{sec: applications} can be accessed from National Cancer Institute by arranging a Material Transfer Agreement.
A simulated data set, simulated according to one of the designs described in Section \ref{sec: simulation studies}, 
and a `readme' file providing additional details are also included in the supplementary material.  

\section*{Acknowledgments}
Pati's research was supported in part by NSF grant DMS1613156.
Mallick's research was supported by grant R01CA194391 from the National Cancer Institute and grant CCF-1934904 from the NSF.
Carroll's research was supported in part by grant U01-CA057030 from the National Cancer Institute.

We thank the University of Texas Advanced Computing Center (TACC) for providing computing resources that contributed to the research reported here.

\newcommand{\Appendix}{\appendix\def\thesection{Appendix~\Alph{section}}\def\thesubsection{\Alph{section}.\arabic{subsection}}}
\section*{Appendix}
\begin{appendix}
\Appendix
\renewcommand{\theequation}{A.\arabic{equation}}
\setcounter{equation}{0}
\baselineskip=17pt

\section{Supplementary Results}\label{appendix: supplementary results}
The following result establishes that our model for the correlation matrices from Section \ref{sec: mvt copula density of interest} is sufficiently flexible. 
\begin{Lem}
Any correlation matrix admits a parametrization proposed in Section \ref{sec: mvt copula density of interest}. 
\end{Lem}
\begin{proof}
For any correlation matrix $\bR^{p\times p}$, consider, without loss of generality, an associated covariance matrix $\bSigma=((\sigma_{ij}))=\bD\bR\bD$ where $\bD=\diag(\sigma_{1},\dots,\sigma_{p})$ with $\sigma_{ii}=\sigma_{i}^{2}$. 
Let $\bSigma=\bL\bL\trans$ be the Cholesky decomposition of $\bSigma$ 
where $\bL=((L_{ij}))$ is a lower triangular matrix with $L_{ij}=0$ for all $i<j$ and positive diagonal elements $L_{ii}>0$. 
The Cholesky decomposition of $\bR$ is then 
$\bR = \bV\bV\trans$ where $\bV=\bD^{-1}\bL$. 
The elements of each row are otherwise unrestricted and can be represented using spherical coordinates as 
\vspace{-4ex}\\
\bse
&& L_{1,1}=\sigma_{1}, \\
&& L_{2,1}=\sigma_{2} \sin\phi_{2,1}, ~ L_{2,2}=\sigma_{2} \cos\phi_{2,1},\\
&& L_{3,1} =\sigma_{3}\sin\phi_{3,1}, ~L_{3,2}=\sigma_{3}\cos\phi_{3,1}\sin\phi_{3,2},~L_{3,3}=\sigma_{3}\cos\phi_{3,1}\cos\phi_{3,2},\\
&& \vdots \\
&& L_{p,1}=\sigma_{p}\sin\phi_{p,1}, ~ L_{p,2}=\sigma_{p}\cos\phi_{p,1}\sin\phi_{p,2}, \dots, ~ L_{p,p}=\sigma_{p}\cos\phi_{p,1}\dots\cos\phi_{p,p-1},
\ese
\vspace{-4ex}\\
with $\{\phi_{i,j}\}_{i=1,j=1}^{p,i-1}\in(-\pi/2,\pi/2)$. 
The elements of $\bV=\bD^{-1}\bL=((v_{ij}))$ are thus given by 
\vspace{-4ex}\\
\bse
&& v_{1,1}=1, \\
&& v_{2,1}=\sin\phi_{2,1}, ~ v_{2,2}=\cos\phi_{2,1},\\
&& v_{3,1} =\sin\phi_{3,1}, ~v_{3,2}=\cos\phi_{3,1}\sin\phi_{3,2},~v_{3,3}=\cos\phi_{3,1}\cos\phi_{3,2},\\
&& \vdots \\
&& v_{p,1}=\sin\phi_{p,1}, ~v_{p,2}=\cos\phi_{p,1}\sin\phi_{p,2}, \dots, ~ v_{p,p}=\cos\phi_{p,1}\dots\cos\phi_{p,p-1}.
\ese
\vspace{-4ex}\\
The above representation is clearly over-parameterized. 
Setting 
\vspace{-4ex}\\
\bse
&& v_{1,1}=1, \\
&& v_{2,1}=\sin\phi_{2,1}=b_{1}, ~ v_{2,2}=\cos\phi_{2,1}=\sqrt{1-b_{1}^{2}},\\
&& v_{3,1}=\sin\phi_{3,1}=b_{2}\sin\theta_{1}, ~v_{3,2}=\cos\phi_{3,1}\sin\phi_{3,2}=b_{2}\cos\theta_{1},~v_{3,3}=\cos\phi_{3,1}\cos\phi_{3,2}=\sqrt{1-b_{2}^{2}},\\
&& \vdots \\
&& v_{p,1}=\sin\phi_{p,1}=b_{p-1}\sin\theta_{(p^{2}-5p+8)/2}, ~v_{p,2}=\cos\phi_{p,1}\sin\phi_{p,2}=b_{p-1}\cos\theta_{(p^{2}-5p+8)/2}\sin\theta_{(p^{2}-5p+8)/2+1}, \\
&& \hspace{9cm}\dots, ~ v_{p,p}=\sqrt{1-b_{p-1}^{2}}
\ese
removes the redundancies and results in the parametrization of Section \ref{sec: mvt copula density of interest}
\end{proof}

\section{Proof of Theorem \ref{thm: identifiability}}\label{appendix: proof of identifiability}
\begin{proof}
The first part of the proof proceeds along the lines of \cite{hu2008instrumental} with some important differences.  We first show that 
$f_{\bY_{1}\mid \wt\bX}(\bY_{1}\mid \wt\bX), f_{\bY_{2}\mid \wt\bX}(\bY_{2}\mid \wt\bX), f_{\wt\bX}(\wt\bX \mid \bY_{3})$ are recoverable from the conditional density $f_{\bY_{1}, \bY_{2} \mid \bY_{3}}(\bY_{1}, \bY_{2} \mid \bY_{3})$, given by 
\vspace{-4ex}\\
\bse
f_{\bY_{1}, \bY_{2} \mid \bY_{3}}(\bY_{1}, \bY_{2} \mid \bY_{3}) = \int f_{\bY_{1}\mid \wt\bX}(\bY_{1}\mid \wt\bX) f_{\bY_{2}\mid \wt\bX}(\bY_{2}\mid \wt\bX) f_{\wt\bX}(\wt\bX \mid \bY_{3}) d\wt\bX. 
\ese
\vspace{-4ex}\\
For any $\bZ_{1}, \bZ_{2}, \bZ_{3}, \bX$, define a collection of operators $T_{\bZ_{1}; \bZ_{2}\mid \bZ_{3}}: L_{1} \to L_{1}$  indexed by $\bZ_{1}$ an operator $T_{\bZ_{1}\mid \bZ_{3}}: L_{1} \to L_{1}$ and a collection of diagonal operators indexed by $\bZ_{2}$ as  $D_{\bZ_{1}; \bX}: L_{1} \to L_{1}$.  
\vspace{-4ex}\\
\bse
T_{\bZ_{1}; \bZ_{2}\mid \bZ_{3}} g(\bZ_{2}) &=& \int f_{\bZ_{1} \mid \bX}(\bZ_{1} \mid \bX) f_{\bZ_{2} \mid \bX}(\bZ_{2} \mid \bX) f_{\bX \mid \bZ_{3}}(\bX \mid \bZ_{3}) g(\bZ_{3}) d\bX d\bZ_{3}, \\
T_{\bZ_{2}\mid \bZ_{3}} g(\bZ_{2}) &=& \int  f_{\bZ_{2} \mid \bX}(\bZ_{2} \mid \bX) f_{\bX \mid \bZ_{3}}(\bX \mid \bZ_{3}) g(\bZ_{3}) d\bX d\bZ_{3}, \quad 
D_{\bZ_{1}; \bX} g(\bX) = f_{\bZ_{1} \mid \bX }(\bZ_{1} \mid \bX)  g(\bX)
\ese
\vspace{-4ex}\\
Note that 
\vspace{-6ex}\\
\be
T_{\bY_{1}; \bY_{2}\mid \bY_{3}} &=& T_{\bY_{2} \mid \wt\bX} D_{\bY_{1}; \wt\bX} T_{\wt\bX \mid \bY_{3}} \label{op1}, \\
T_{\bY_{2}\mid \bY_{3}} &=& T_{\bY_{2}\mid \wt\bX} T_{\wt\bX \mid \bY_{3}}  \label{op2}.  
\ee  
\vspace{-4ex}\\
In the following, we prove that the operators $T_{\bY_{2}\mid \wt\bX}$ and $T_{\bY_{2}\mid \bY_{3}}$ are invertible. 
Then, \eqref{op1}-\eqref{op2} will imply that 
\vspace{-6ex}\\
\be
T_{\bY_{1}; \bY_{2}\mid \bY_{3}} T^{-1}_{\bY_{2}\mid \bY_{3}}  &=&  T_{\bY_{2} \mid \wt\bX} D_{\bY_{1}; \wt\bX} T^{-1}_{\bY_{2} \mid \wt\bX}. \label{opdecomp}
\ee  
\vspace{-4ex}\\
It is easy to see that $T_{\bY_{2}\mid \bY_{3}}$ is invertible if and only if $T_{\bY_{2}\mid \wt\bX}$ and $T_{\bY_{3}\mid \wt\bX}$ both are invertible. 
Now since $\bY_{j}, j=1,\dots,3$, are identically distributed conditioned on  $\wt\bX$, then it is enough to show that $T_{\bY_{1}\mid \wt\bX}$ is invertible. 
Note that
\vspace{-4ex}\\
\bse
\textstyle T_{\bY_{1}\mid \wt\bX}(g) (\bY_{1}) = \int f_{\bU \mid \wt\bX} ( \bY_{1} - \wt\bX \mid \wt\bX) g(\wt\bX) d \wt\bX 
= \int \frac{1}{\prod_{\ell} s_\ell(\widetilde{X}_\ell)}f_{\bepsilon}\big[ \mbox{diag}\{\bS(\wt\bX)\}^{-1}(\bY_{1} - \wt\bX)\big] g(\wt\bX) d \wt\bX.
\ese
\vspace{-4ex}\\
If  $s_\ell(\widetilde{X}_\ell) = 0$ for some $\ell$ and for $\widetilde{X}_\ell \in \mathcal{Z}$ for some set $\mathcal{Z}$, then 
$f_{\bY_1 \mid \wt\bX} (\bY_1 \mid \wt\bX) = \delta_{\wt\bX}$ for $\widetilde{X}_\ell \in \mathcal{Z}$, where $\delta_{\wt\bX}$ is a degenerate probability measure at the point $\wt\bX$. In this case, for $\widetilde{X}_\ell \in \mathcal{Z}$, the $f_{\bY_1 \mid \wt\bX}$ is known and hence recoverable. Hence, we can assume $\bS(\wt\bX) > 0$. Then by {(A2)} the Fourier transform of  $f_{\bepsilon}$ is non-vanishing everywhere, by Wiener's Theorem \citep{goldberg1962}, the closed linear span of $f_{\bU \mid \wt\bX}$ is $L_{1}$. By Hahn-Banach Theorem, the dual space of $L_{1}$ is $L_{\infty}$ and there is an isometric isomorphism from $L_\infty$
to $L_{1}$ by $\Phi: g  \mapsto T_{\bY_{1}\mid \wt\bX}(g)$. 
Since the closed 
linear span of $f_{\bU \mid \wt\bX}$ is $L_{1}$,  $T_{\bY_{1}\mid \wt\bX}(g) = 0$ for all $\wt\bX$ implies  that the mapping $\Phi$ is identically equal to zero.  This proves that $T_{\bY_{1}\mid \wt\bX}$ is invertible. 

By {(A2)},  for all $\bX_{1}, \bX_{2}$, the set $\{\bY:f_{\bY_{1} \mid \bX_{1}}(\bY \mid \bX_{1})  \neq f_{\bY_{2} \mid \bX_{2}}(\bY \mid \bX_{2})  \}$ has a positive probability whenever $\bX_{1} \neq \bX_{2}$. Note that this includes the case when $s_{\ell}(X_{\ell})$ is a constant function for all $\ell$. In that case, the variation in the conditional density is caused by the location term $\wt\bX$ and the above-mentioned set has thus a positive probability. 
Hence, by the proof of Theorem 1 in \cite{hu2008instrumental}, \eqref{opdecomp} is a unique decomposition. 
Therefore, $f_{\bY_{1}\mid \wt\bX}$ can be recovered from $f_{\bY_{1}, \bY_{2} \mid \bY_{3}}$ and hence from $f_{\bY_{1}, \bY_{2}, \bY_{3}}$.   Now since $T_{\bY_{2}\mid \wt\bX} ^{-1}T_{\bY_{2}\mid \bY_{3}} = T_{\wt\bX \mid \bY_{3}}$,  $f_{\wt\bX \mid \bY_{3}}$ is identifiable. This implies  $f_{\wt\bX}(\wt\bX) = \int f_{\wt\bX \mid \bY_{3}}(\wt\bX \mid \bY_{3}) f_{\bY_{3}}(\bY_{3}) d\bY_{3}$ is identifiable, completing the proof. 

To prove the second part, observe that it suffices to show that $f_{\bX}(\bX)$ 
is uniquely identified from $f_{\bY_{j}\mid \wt\bX}(\bY_{j}\mid \wt\bX)$ for $j=1,\ldots,3$ and $f_{\wt\bX}(\wt\bX)$. 
Since $h_\ell(X_\ell) = \wt X_{\ell}$, its distribution is uniquely identified from $\wt X_{\ell}$ for $\ell=1, \ldots, q$. 
It thus remains to show that the distribution of $P_{\ell}(X_{\ell})$ is uniquely identified for $\ell=1, \ldots, q$. 
Since $P_{\ell}(X_{\ell}) = \Phi\{h_{\ell}(X_{\ell})\}$, this follows immediately. 
\end{proof}

\end{appendix}

\baselineskip=14pt

\bibliographystyle{natbib}
\bibliography{ME,Copula}

\clearpage\pagebreak\newpage
\pagestyle{fancy}
\fancyhf{}
\rhead{\bfseries\thepage}
\lhead{\bfseries SUPPLEMENTARY MATERIAL}

\begin{center}
\baselineskip=27pt
{\LARGE{\bf Supplementary Material for\\ Bayesian Copula Density Deconvolution for Zero Inflated Data in Nutritional Epidemiology}}
\end{center}

\vskip 2mm
\begin{center}
Abhra Sarkar\\
abhra.sarkar@utexas.edu \\
Department of Statistics and Data Sciences,
The University of Texas at Austin\\
2317 Speedway D9800, Austin, TX 78712-1823, USA\\
\hskip 5mm \\
Debdeep Pati and Bani K. Mallick\\
debdeep@stat.tamu.edu and 
bmallick@stat.tamu.edu\\
Department of Statistics, Texas A\&M University\\ 
3143 TAMU, College Station, TX 77843-3143, USA\\
\hskip 5mm \\
Raymond J. Carroll\\
carroll@stat.tamu.edu\\
Department of Statistics, Texas A\&M University\\ 
3143 TAMU, College Station, TX 77843-3143, USA\\
School of Mathematical and Physical Sciences, University of Technology Sydney\\ 
Broadway NSW 2007, Australia\\
\end{center}

\setcounter{equation}{0}
\setcounter{page}{1}
\setcounter{table}{1}
\setcounter{figure}{0}
\setcounter{section}{0}

\sectionfont{\fontsize{15}{15}\selectfont}

\numberwithin{table}{section}
\renewcommand{\theequation}{S.\arabic{equation}}
\renewcommand{\thesubsection}{S.\arabic{section}.\arabic{subsection}}
\renewcommand{\thesection}{S.\arabic{section}}
\renewcommand{\thepage}{S.\arabic{page}}
\renewcommand{\thetable}{S.\arabic{table}}
\renewcommand{\thefigure}{S.\arabic{figure}}
\baselineskip=17pt

\section{Review of Copula Basics} \label{sec: copula basics}
The literature on copula models is enormous. See, for example, \citelatex{nelsen2007introduction, joe2015dependence, shemyakin2017introduction} and the references therein.
For easy reference, we provide a brief review of the basics here. 

A function $\cC(\bu) = \cC (u_{1},\dots,u_{p}) : [0,1]^p \rightarrow [0,1]$ is called a copula 
if $\cC(\bu)$ is a continuous cumulative distribution function (cdf) on $[0,1]^p$ 
such that each marginal is a uniform cdf on $[0,1]$. 
That is, for any $\bu \in [0,1]^p$,
$\cC(\bu) = \cC (u_{1},\dots,u_{p}) = \Pr(U_{1}\leq u_{1}, \dots,U_{p}\leq u_{p})$ with 
$\cC(1,\dots,1,u_{i},1,\dots,1) = \Pr(U_{i} \leq u_{i}) = u_{i},  i=1,\dots,p$.
If $\{X_i\}_{i=1}^{p}$ are absolutely continuous random variables having marginal cdf $\{H_{i}(x_{i})\}_{i=1}^{p}$ 
and marginal probability density functions (pdf) $\{h_{i}(x_{i})\}_{i=1}^{p}$,  joint cdf $H(x_{1},\dots,x_{p})$ and joint pdf $h(x_{1},\dots,x_{p})$, 
then a copula $\cC$ can be defined in terms of $H$ as 
$\cC(u_{1},\dots,u_{p}) = H \left(x_{1}, \dots, x_{p}\right)$ where $u_{i} = H_{i}(x_{i}), i=1,\dots,p$.
It follows that
$h(x_{1},\dots,x_{p}) = c(u_{1},\dots,u_{p})\prod_{i=1}^{p}h_{i}(x_{i})$, 
where $c(u_{1},\dots,u_{p}) = {\partial^p \cC(u_{1},\dots,u_{p})}  /  {(\partial u_{1} \dots \partial u_{p})}$.
This defines a copula density $c(\bu)$ in terms of the joint and marginal pdfs of $\{X_{i}\}_{i=1}^{p}$ as
\vspace{-4ex}\\
\be
\textstyle c(u_{1},\dots,u_{p}) = h(x_{1},\dots,x_{p}) / \prod_{i=1}^{p}h_{i}(x_{i}). \label{eq:Copula A2}
\ee
\vspace{-4ex}\\
Conversely, if $\{V_i\}_{i=1}^{p}$ are continuous random variables having fixed marginal cdfs $\{F_{i}(v_i)\}_{i=1}^{p}$, 
then their joint cdf $F(v_{1},\dots,v_{p})$, with a dependence structure introduced through a copula $\cC$, can be defined as
\vspace{-4ex}\\
\be
F(v_{1},\dots,v_{p}) = \cC\{F_{1}(v_{1}),\dots,F_{p}(v_{p})\} =\cC(u_{1},\dots,u_{p}),         \label{eq:Copula A3} 
\ee
\vspace{-4ex}\\
where $u_{i} = F_{i}(v_{i}), i=1,\dots,p$.
If $\{V_i\}_{i=1}^{p}$ have marginal densities $\{f_{i}(v_i)\}_{i=1}^{p}$, 
then from (\ref{eq:Copula A3}) it follows that the joint density $f(v_{1},v_{2},\dots,v_{p}) $ is given by
\vspace{-4ex}\\
\be
f(v_{1},\dots,v_{p}) &= c(u_{1},\dots,u_{p})\prod_{i=1}^{p}f_{i}(v_i).  \label{eq:Copula A4}
\ee
\vspace{-4ex}\\
With $F_i(v_i) = u_{i} = H_i(x_{i}), i=1,\dots,p$, substitution of the copula density (\ref{eq:Copula A2}) into (\ref{eq:Copula A4}) gives 
\vspace{-6ex}\\
\be
&& f(v_{1},\dots,v_{p}) = c(u_{1},\dots,u_{p})\prod_{i=1}^{p}f_{i}(v_i) = \bigg\{\frac{h(x_{1},\dots,x_{p})}{\prod_{i=1}^{p}h_{i}(x_{i})}\bigg\}  \prod_{i=1}^{p}f_{i}(v_i). \label{eq:Copula A5}
\ee
\vspace{-4ex}\\

Equation (\ref{eq:Copula A3}) can be used to define flexible multivariate dependence structure using standard known multivariate densities \citeplatex{sklar1959}. 
Let $\MVN_{p}(\bmu,\bSigma)$ denote a $p$-variate normal distribution with mean vector $\bmu$ and positive semidefinite covariance matrix $\bSigma$. 
An important case is $\bX=(X_{1},\dots,X_{p})\trans \sim \MVN_{p}(\bzero,\bR)$, where $\bR$ is a correlation matrix. 
In this case, $\cC (u_{1},\dots,u_{p}|\bR) = \Phi_{p} \{\Phi^{-1}(u_{1}),\dots,\Phi^{-1}(u_{p})\mid \bR\}$, 
where $\Phi(x) = \Pr\{X \leq x \vert X\sim \Normal(0,1)\}$ and 
$\Phi_{p}(x_{1},\dots,x_{p}|\bR) = \Pr\{X_{1} \leq x_{1}, \dots,X_{p} \leq x_{p} \vert \bX \sim \MVN_{p}(\bzero,\bR)\}$. 
If $\bX \sim N_{p}(\bzero,\bSigma)$, where $\bSigma = ((\sigma_{i,j}))$ is a covariance matrix with $\sigma_{ii}=\sigma_i^2$, then defining $\bLambda=\diag(\sigma_{1}^2,\dots,\sigma_{p}^2)$ and $\bY = \bLambda^{-\frac{1}{2}}\bX$ and noting that $\bSigma=\bLambda^{1/2} \bR \bLambda^{1/2}$, we have  
\vspace{-4ex}\\
\bse
c(u_{1},\dots,u_{p}) = {\MVN_{p}(\bx\mid\bzero,\bSigma)}   /   {\MVN_{p}(\bx\mid\bzero,\bLambda)}
= |\bLambda|^{1/2} |\bSigma|^{-1/2} \exp\left\{-\bx\trans(\bSigma^{-1}-\bLambda^{-1})\bx/2\right\} \\
= |\bR|^{-1/2} \exp\{-\by\trans(\bR^{-1}-\bI_{p})\by/2\} 
=  {\MVN_{p}(\by\mid \bzero,\bR)}  /  {\MVN_{p}(\by\mid \bzero, \bI_{p})}.
\ese 
\vspace{-4ex}\\
Sticking to the standard normal case, 
a flexible dependence structure between random variables $\{V_i\}_{i=1}^{p}$ with given marginals $\{F_i(v_i)\}_{i=1}^{p}$ may thus be obtained  
assuming a Gaussian distribution on the latent random variables $\{Y_{i}\}_{i=1}^{p}$ 
obtained through the transformations $F_{i}(v_{i}) = u_{i} = \Phi(y_{i}), i=1,\dots,p$. 
The joint density of $\bV=(V_{1},\dots,V_{p})\trans$ is then given by
\vspace{-4ex}\\
\bse
&& \hspace{-1cm} f(v_{1},\dots,v_{p}) = c(u_{1},\dots,u_{p})\prod_{i=1}^{p}f_{i}(v_i) 
= \frac{\MVN_{p}(\by\mid \bzero,\bR)}   {\MVN_{p}(\by\mid \bzero, \bI_{p})}  \prod_{i=1}^{p}f_{i}(v_i). 
\ese
\vspace{-4ex}\\
We have 
\vspace{-6ex}\\
\bse
\Pr(V_{1} \leq v_{1},\dots,V_{p} \leq v_{p}) = \Pr[Y_{1} \leq \Phi^{-1}\{F_{1}(v_{1})\},\dots,Y_{p} \leq \Phi^{-1}\{F_{p}(v_{p})\}\mid \bY \sim \MVN_{p}(\bzero,\bR)].
\ese
\vspace{-4ex}\\
For $q\leq p$, with $(Y_{1},\dots,Y_{q})\trans \sim \MVN_{q}(\bzero,\bR_{q})$, we then have 
\vspace{-4ex}\\
\bse
\Pr(V_{1} \leq v_{1},\dots,V_{q} \leq v_{q}) = \Pr[Y_{1} \leq \Phi^{-1}\{F_{1}(v_{1})\},\dots,Y_{q} \leq \Phi^{-1}\{F_{q}(v_{q})\}\mid \bY \sim \MVN_{q}(\bzero,\bR_{q})],
\ese
\vspace{-4ex}\\
implying that the density of $(V_{1},\dots,V_{q})$ will be 
\vspace{-4ex}\\
\bse
&& \hspace{-1cm} f(v_{1},\dots,v_q) = c(u_{1},\dots,u_q)\prod_{i=1}^{q}f_{i}(v_i) 
= \frac{\MVN_{q}(\by\mid \bzero,\bR_{q})}   {\MVN_{q}(\by\mid \bzero, \bI_{q})}  \prod_{i=1}^{q}f_{i}(v_i). 
\ese
\vspace{-4ex}

\section{Quadratic B-splines}\label{sec: Quadratic B-splines}

Consider knot-points $t_{1} = t_{2} = t_{3} = A < t_{4} < \dots < B = t_{K+3} = t_{K+4} = t_{K+5}$,
where $t_{3:(K+3)}$ are equidistant with $\delta = (t_{4} - t_{3})$.
For $j=3,4,\dots,(K+2)$, quadratic B-splines $b_{2,j}$ are defined as 
\vspace{-4ex}\\
\bse
 b_{2,j}(X) &= \left\{\begin{array}{ll}
        \{(X-t_{J-1})/\delta\}^{2}/2  				& ~~~~\text{if } t_{J-1} \leq X < t_{J},  \\
        -\{(X-t_{J})/\delta\}^{2} + (X-t_{J})/\delta + 1/2	     	& ~~~~\text{if } t_{J} \leq X < t_{j+2},  \\
        \{1-(X-t_{j+2})/\delta\}^{2}		       	    	& ~~~~\text{if } t_{j+2} \leq X < t_{j+3},  \\
        0  							& ~~~~ \text{otherwise}.
        \end{array}\right.
\ese
The components at the ends are likewise defined as 
\bse
b_{2,1}(X) &=&  \left\{\begin{array}{ll}
        \{1-(X-t_{1})/\delta\}^{2} /2           	& ~~~~~~~~~~~~~~~~~~~~~~~\text{if } t_{3} \leq X < t_{4},  \\
        0  						& ~~~~~~~~~~~~~~~~~~~~~~~ \text{otherwise}.
        \end{array}\right.\\
b_{2,2}(X) &=&  \left\{\begin{array}{ll}
        -\{(X-t_{3})/\delta\}^2 + (X-t_{4})/\delta + 1/2     	& ~~~~\text{if } t_{3} \leq X < t_{4},  \\
        \{1-(X-t_{4})/\delta\}^{2} /2            			& ~~~~\text{if } t_{4} \leq X < t_{5},  \\
        0  							& ~~~~ \text{otherwise}.
        \end{array}\right.\\
b_{2,K+1}(X) &=&  \left\{\begin{array}{ll}
        \{(X-t_{K+1})/\delta\}^{2} /2  					& ~~~~\text{if } t_{K+1} \leq X < t_{K+2},  \\
        -\{(X-t_{K+2})/\delta\}^{2} + (X-t_{K+2})/\delta + 1/2     	& ~~~~\text{if } t_{K+2} \leq X < t_{K+3},  \\
        0  								& ~~~~ \text{otherwise}.
        \end{array}\right.\\
b_{2,K+2}(X) &=&  \left\{\begin{array}{ll}
        \{(X-t_{K+2})/\delta\}^{2} /2  		& ~~~~~~~~~~~~~~~~~~~~~~~~~\text{if } t_{K+2} \leq X < t_{K+3},  \\
        0  					& ~~~~~~~~~~~~~~~~~~~~~~~~~ \text{otherwise}.
        \end{array}\right.
\ese

\newpage
\section{Comparison with Previous Works} \label{sec: mvt copula comparison with NCI}
As discussed briefly in the introduction of the main paper, 
the problem of estimating nutritional intakes from zero-inflated data has previously been considered by a few, including \citelatex{Tooze2002, Tooze2006, Kipnis_etal:2009, Zhang2011a, Zhang2011b}. 
We review here, in greater details, the methodology of \citelatex{Zhang2011a,Zhang2011b}. 
To the best of our knowledge, the main working principles outlined below, often referred to as the {\it NCI method}, are common to all previous approaches on zero-inflated data. 

Let $\bW_{tr,i,j} = (W_{tr,1,i,j},\dots,W_{tr,2q+p,i,j})\trans$ be made of all continuous components, as our $\bW_{i,j}$ before, but are now related to the observed data $\bY_{i,j}$ as 
\vspace{-4ex}\\
\bse
& Y_{\ell,i,j} = \Ind(W_{\ell,i,j}>0),~~~~~W_{tr,\ell,i,j} = \wt{X}_{tr,\ell,i} + U_{tr,\ell,i,j}, ~~~~~\hbox{for}~\ell=1,\dots,q,\\
& \{g_{tr}(Y_{\ell,i,j},\lambda_{\ell-q}) \mid Y_{\ell-q,i,j}=1\} = W_{tr,\ell,i,j} = \wt{X}_{tr,\ell,i} + U_{tr,\ell,i,j}, ~~~~~\hbox{for}~\ell=q+1,\dots,2q, \\
& g_{tr}(Y_{\ell,i,j},\lambda_{\ell-q}) = W_{tr,\ell,i,j} = \wt{X}_{tr,\ell,i} + U_{tr,\ell,i,j}, ~~~~~\hbox{for}~\ell=2q+1,\dots,2q+p.
\ese 
\vspace{-4ex}\\
Here, $g_{tr}(Y,\lambda) = \sqrt{2}\{g(Y,\lambda)-\mu(\lambda)\}/\sigma(\lambda)$, where $g(Y,\lambda)$ is the usual Box-Cox transformation
\vspace{-6ex}\\
\bse
& g(Y, 0) = \log~Y,~~~\lambda=0,\\
& g(Y, \lambda) = \frac{Y^{\lambda}-1}{\lambda},~~~\lambda \neq 0.
\ese
\vspace{-4ex}\\
The transformation parameter $\lambda$ as well as $\mu(\lambda),\sigma(\lambda)$, the mean and standard deviation of $g(Y,\lambda)$, are all calculated using positive recall data only 
and then kept fixed for the rest of the analysis. 
Here, $\wt{\bX}_{tr,i}=(\wt{X}_{tr,1,i},\dots,\wt{X}_{tr,2q+p,i})\trans$ are random effects for the $i\th$ subject 
and $\bU_{tr,i,j}=(U_{tr,1,i,j},\dots,U_{tr,2q+p,i,j})\trans$ are errors and pseudo-errors for the $j\th$ recall of the $i\th$ subject. 
The components $\wt{\bX}_{tr,i}$ and $\bU_{tr,i,j}$ are assumed to be independently distributed as $\MVN_{2q+p}(\wt{\bX}_{tr,i} \mid \bmu_{\bX,tr},\bSigma_{\bX,tr})$ 
and $\MVN_{2q+p}(\bU_{tr,i,j} \mid \bzero,\bSigma_{\bU,tr})$, respectively. 
For identifiability etc., $\bSigma_{\bU,tr}$ is restricted to have the special structure 
\vspace{-4ex}\\
\bse
\bSigma_{\bU,tr} = ((\sigma_{u,r,s})),~~~\sigma_{u,r,r} =1,~\sigma_{u,r,r+q}=0,~\text{for}~r=1,\dots,q.
\ese
\vspace{-4ex}\\
A parametrization similar to one we used for modeling the correlation matrices $\bR$ was developed to enforce these restrictions. 
Appropriate priors were assigned on the parameters and an MCMC algorithm was used to draw samples from the posterior. 
Based on estimates of these parameters, the true intakes $T_{\ell,i}$ were {defined} as 
\vspace{-4ex}\\
\bse
&& X_{\ell,i} = \Phi(\wt{X}_{tr,\ell,i}) ~ g_{tr}^{\star}(\wt{X}_{tr,q+\ell,i}, \lambda_{\ell},\sigma_{u,q+\ell,q+\ell}), ~~~~~\ell=1,\dots,q,\\
&& X_{\ell,i} = g_{tr}^{\star}(\wt{X}_{tr,q+\ell,i}, \lambda_{\ell},\sigma_{u,q+\ell,q+\ell}), ~~~~~~\ell=q+1,\dots,q+p,
\ese
\vspace{-6ex}\\
where 
\vspace{-6ex}\\
\bse
& g_{tr}^{\star}(X, \lambda,\sigma) = g_{tr}^{-1}(X, \lambda) + \frac{1}{2}\sigma \frac{\partial^{2} g_{tr}^{-1}(X,\lambda)}{\partial X^{2}}, \\
&g_{tr}^{-1}(X, 0) = \exp\{\mu(0) +  \frac{1}{\sqrt{2}} \sigma(0) X\}, ~~~\frac{\partial^{2} g_{tr}^{-1}(X,0)}{\partial X^{2}} = \frac{\sigma^{2}(0)}{2}g_{tr}^{-1}(X,0), ~\text{when}~\lambda=0, \\
&g_{tr}^{-1}(X, \lambda) = [1+\lambda\{\mu(\lambda) +  \frac{1}{\sqrt{2}} \sigma(\lambda) X\}]^{\frac{1}{\lambda}}, ~\text{and}~\\
& \frac{\partial^{2} g_{tr}^{-1}(X,\lambda)}{\partial X^{2}} = \frac{\sigma^{2}(\lambda)}{2} (1-\lambda) [1+\lambda\{\mu(\lambda) + \frac{1}{\sqrt{2}}\sigma(\lambda) X\}]^{-2+\frac{1}{\lambda}}, ~\text{when}~\lambda \neq 0.
\ese
Finally, the distributions $f_{X,\ell}$ and $f_{\bX}$ are obtained by applying conventional (measurement error free) density estimation techniques on the estimates $X_{\ell,i}$'s and $\bX_{i}$'s, respectively, and then adjusting their supports to be restricted to the positive real line only.

\begin{figure}[!ht]
\centering
\includegraphics[height=5.95cm, width=9.5cm, trim=2cm 1cm 1cm 1cm, clip=true]{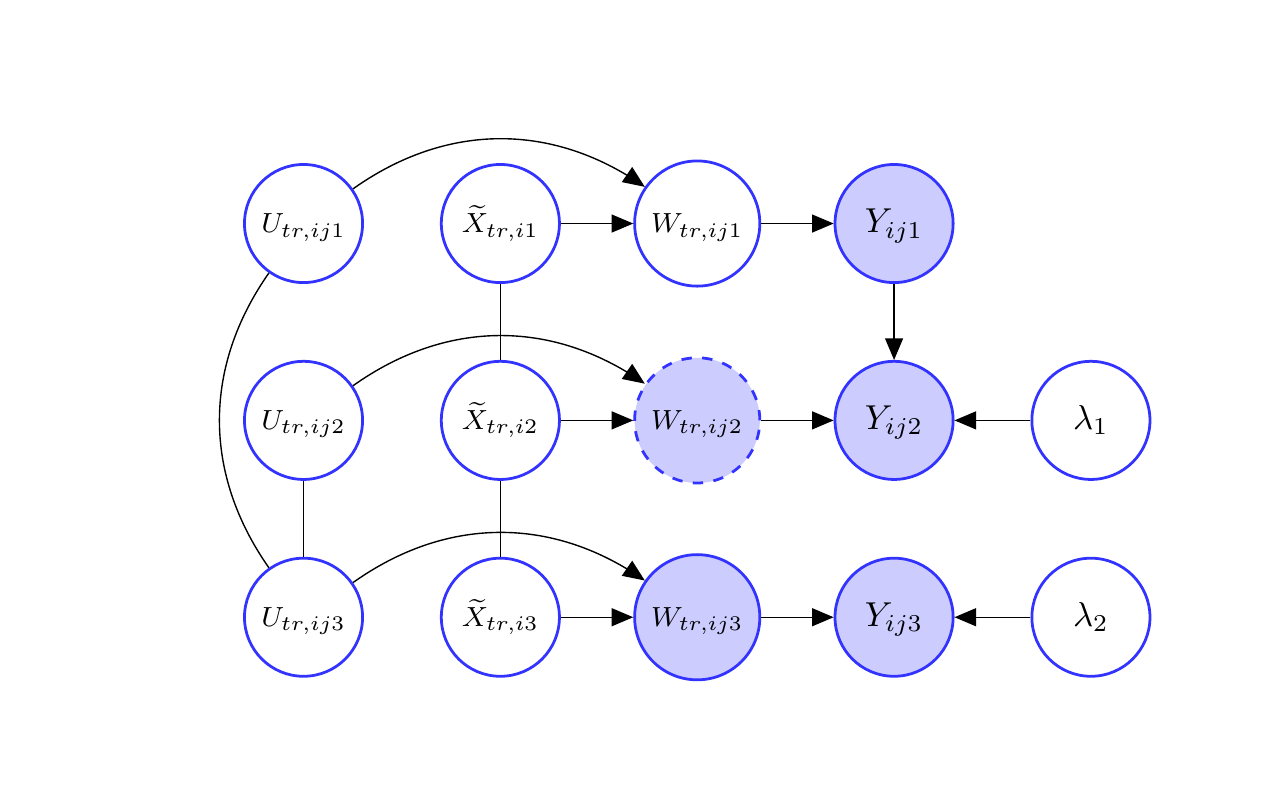}
\label{fig: Graphical model}
\caption{\baselineskip=10pt Graph depicting the dependency structure of the model developed in \cite{Zhang2011b} 
for one episodically consumed component and one regularly consumed component. 
While other parameters are suppressed, the Box-Cox parameters $\lambda_{1}$ and $\lambda_{2}$ are shown to highlight nonlinear transformations of the surrogates. 
Compare with our model depicted in Figure \ref{fig: graph our model} in the main paper. 
}
\label{fig: graph Zhang 2011b}
\end{figure}

There are many fundamental aspects where our proposed approach, 
including our latent variable framework described in Section \ref{sec: mvt copula latent variable framework} in the main paper, 
differs from these previous works which we highlight below.  

First, consider the probabilities of reporting positive consumptions. 
Previous approaches allow these probabilities to only indirectly depend on the associated $X_{\ell}$'s which, as can be seen in the right panels of Figure \ref{fig: EATS Milk & Whole Grains} in the main paper, 
are certainly informative about these probabilities.  
We make use of this information by modeling $P_{\ell}(X_{\ell})$ to be a function of $X_{\ell}$.  

Previous methods also require transformation of the surrogates to a different scale where assumptions such as normality, homoscedasticity, independence etc. are expected to hold, 
and then transformation the results back to the original scale. 
Box-Cox transformations that make the transformed surrogates conform to all the desired parametric assumptions, however, almost never exist \citep{Sarkar_etal:2014}. 
Such transformations and retransformations, when applied to surrogates, thus result in loss of information, introducing bias. 
%
Our method, in contrast, does not rely on restrictive parametric assumptions even though it addresses the modeling challenges more directly 
where the modeling assumptions of zero mean errors etc. are more meaningful. 


The assumption of unbiasedness of the recalls for the true latent consumptions also makes the most sense in the original observed scale as our proposed method assumes, 
and not in any arbitrarily defined nonlinearly transformed scale as all previously existing methods for zero-inflated data, including \cite{Zhang2011a,Zhang2011b}, assume.

Previous literature on estimating $f_{\bX}$ from zero-inflated data also do not model $f_{\bX}$ directly but rely on 
first estimating $\bX^{+}$, the long-term daily intakes on consumption days, 
and the probabilities of reporting positive consumptions, and then combining them to arrive at estimates of $\bX$, 
and finally using these estimates to construct an estimate of $f_{\bX}$. 
The novel design of our hierarchical latent variable framework, on the other hand, allows us to model $f_{\bX}$ directly. 
We rather leave the distribution of $\bX^{+}$ unspecified, which is usually not of much interest. 
If needed, it can be obtained as 
\vspace{-4ex}\\
\bse
\textstyle f_{\bX^{+}}(\bX^{+}) = f_{\bX}(\bX) \prod_{\ell=1}^{q+p}  \abs{J(X_{\ell} \to X_{\ell}^{+})}.  \label{eq: mixture model for f_X+} 
\ese
\vspace{-4ex}\\
Here, the right hand side, including the Jacobians of transformations $J(X_{\ell} \to X_{\ell}^{+})$, is implicitly understood to be evaluated at $\bX^{+}$. 
Since $X_{\ell}$ and $X_{\ell}^{+}$ are not guaranteed to be strictly one-one, the Jacobians need to be carefully calculated. 
An easy-to-implement 
alternative would be to apply (measurement error free) density estimation algorithms to the estimates of the $X_{\ell}^{+}$'s. 


As discussed in Section \ref{sec: mvt copula density of energy adjusted intakes} in the main paper, 
our approach also allows us to estimate the distribution of energy adjusted long-term intakes, namely $f_{\bZ}$, straightforwardly from $f_{\bX}$ via a simple one-dimensional integration. 
This is again in contrast with previous approaches where such estimates are constructed applying (measurement error free) density estimation methods to estimated values of the $\bZ_{i}$'s.

Aside the probabilities of consumptions $P_{\ell}(X_{\ell})$, our method also produces estimates of the densities of the scaled errors $f_{\epsilon,\ell}$ 
as well as estimates of the measurement errors' conditional variability $s_{\ell}^{2}(X_{\ell})$ 
which may be of some interest to nutritionists but are not available from previously existing methods. 

Finally, previous approaches for zero-inflated data could not handle multiplicative measurement errors. 
Building on the general recipe outlined in \cite{sarkar2018bayesian}, our model, on the other hand, automatically accommodates both conditionally heteroscedastic additive measurement errors as well as multiplicative measurement errors. 

\section{Hyper-parameter Choices and Posterior Computation}  \label{sec: posterior computation}

Samples from the posterior can be drawn using the MCMC algorithm described below. 
In what follows, $\bzeta$ denotes a generic variable that collects the data as well as all parameters of the model, 
including the imputed values of $\bX_{1:n}$ and $\bepsilon_{1:N}$, that are not explicitly mentioned.
Also, the generic notation $p_{0}$ is sometimes used for specifying priors and hyper-priors. 

We now discuss our choices for the prior hyper-parameters and the initial values of the MCMC sampler. 
The starting values of some of the parameters for the multivariate problem are determined by first running samplers for the univariate marginals.
We thus describe the hyper-parameter choices and the initial values for the sampler for the marginal univariate models first. 
Unless otherwise mentioned, the prior hyper-parameter choices for similar model components for the multivariate model remain the same as that used for the univariate marginal models. 
We only detail the sampling steps for the multivariate method. 
The steps for the univariate method were straightforwardly adapted from the multivariate sampler. 

To make the recalls for all the components to be unit free and have a shared support, we transformed the recalls as 
$Y_{\ell,i,j} = 20 \times \frac{Y_{\ell,i,j}}{\max\{Y_{\ell,i,j}\}}$.
The latent $X_{\ell,i}$'s can then be safely assumed to lie in $[0,10]$, greatly simplifying model specification and hyper-parameter selection. 
As opposed to the non-linear Box-Cox transformations used in the previous literature, including  \cite{Zhang2011b}, which often result in loss of information and introduce bias, 
we only make linear scale transformations here that preserve all features of the original data points. 

For the univariate samplers for the marginal components, we then set $W_{\ell,i,j} = Y_{\ell,i,j}$. 
We used the subject-specific sample means $\overline{W}_{\ell,1:n}$ as the starting values for $X_{\ell,1:n}$.
The appropriate number of mixture components in a mixture model depends on the flexibility of the component mixture kernels as well as on specific demands of the particular application at hand. 
With appropriately chosen mixture kernels, univariate mixture models with 5-10 components have often been found to be sufficiently flexible. 
Similar claims can also be made for normalized mixtures of B-splines. 
Detailed guidelines on selecting the number of mixture components for the specific context of deconvolution problems can be found in Section S.1 and S.6 in the Supplementary Materials of \cite{sarkar2018bayesian}. 
Based on such guidelines and extensive numerical experiments, 
we used $J_{\ell}=12$ equidistant knot points for the B-splines supported on $[A_{\ell},B_{\ell}] = [0,10]$ for modeling the densities and the probabilities of reporting non-consumptions of the episodic components, 
as well as the variance functions of both episodic and regular components. 
For regular components, we allowed $K_{X,\ell}=10$ mixture components for the truncated normal mixtures modeling their densities. 
We also allowed $K_{\epsilon,\ell}=10$ mixture components for the mixtures modeling the densities of the scaled errors. 
For the Dirichlet prior hyper-parameters, we set $\alpha_{X} = 1/K_{X}$, $\alpha_{\epsilon} = 1/K_{\epsilon}$. 
The hyper-parameters for the smoothness inducing parameters are set to be mildly informative as 
$a_{\xi}=a_{\beta}=a_{\vartheta}=10, b_{\vartheta}=b_{\beta}=b_{\xi}=1$. 
Introducing latent mixture component allocation variables $\bC_{X,1:(q+p),1:n}$, $\bC_{\epsilon,(q+1):(2q+p),1:N}$ and $\bC2_{\epsilon,(q+1):(2q+p),1:N}$, we can write  
\vspace{-4ex}\\
\bse
& (X_{\ell,i} \mid C_{X,\ell,i}=k,\mu_{X,\ell,k},\sigma_{X,\ell,k}^{2}) \sim \TN(X_{\ell,i} \mid \mu_{X,\ell,k},\sigma_{X,\ell,k}^{2},[A_{\ell},B_{\ell}]), ~~~\ell=1,\dots,q+p,~~\hbox{and}\\
& (\epsilon_{\ell,i,j} \mid C_{\epsilon,\ell,i,j}=k,C2_{\epsilon,\ell,i,j}=t,\mu_{\epsilon,\ell,k,t},\sigma_{\epsilon,\ell,k,t}^{2}) \sim \Normal(\epsilon_{\ell,i,j} \mid \mu_{\epsilon,\ell,k,t},\sigma_{\epsilon,\ell,k,t}^{2}),\\
&\hspace{11cm}~\ell=q+1,\dots,2q+p.
\ese
\vspace{-4ex}\\
The mixture labels $C_{X,\ell,i}$'s, and the component specific parameters $\mu_{X,\ell,k}$'s and $\sigma_{X,\ell,k}$'s are initialized by fitting a $k$-means algorithm with $k=K_{X}$.
The parameters of the distribution of scaled errors are initialized at values that correspond to the special standard normal case.
The initial values of the smoothness inducing parameters are set at $\sigma_{\vartheta,\ell}^{2} = \sigma_{\xi,\ell}^{2} = \sigma_{\xi,\ell}^{2} = 0.1$. 
The associated mixture labels $C_{\epsilon,\ell,i,j}$'s are thus all initialized at $C_{\epsilon,\ell,i,j}=1$.
The initial values of $\bvartheta_{\ell}$'s are obtained by maximizing 
\vspace{-4ex}\\
\bse
\ell(\bvartheta_{\ell}\mid \sigma_{\vartheta,\ell}^{2},\overline{\bW}_{\ell,1:n}) = -\frac{\bvartheta_{\ell}\trans \bP_{\ell} \bvartheta_{\ell}}{2\sigma_{\vartheta,\ell}^{2}} - \sum_{i=1}^{n} \frac{1}{2 s_{\ell}^{2}(\overline{W}_{\ell,i},\bvartheta_{\ell})} \sum_{j=1}^{m_{i}}(W_{\ell,i,j}-\overline{W}_{\ell,i})^{2}
\ese
\vspace{-4ex}\\
with respect to $\bvartheta_{\ell}$. 
Likewise, the parameters $\bxi_{\ell}$'s specifying the densities of the episodic components are initialized by maximizing 
\vspace{-4ex}\\
\bse
\ell(\bbeta_{\ell}\mid \sigma_{\xi,\ell}^{2},\overline{\bW}_{\ell,1:n}) = - \frac{\bxi_{\ell}\trans \bP_{\ell} \bxi_{\ell}}{2\sigma_{\xi,\ell}^{2}} - \sum_{i=1}^{n}\left\{\wh{f}_{X,\ell,i}^{Kern}-\wh{f}_{X,\ell,i}(\overline{W}_{\ell,i},\bxi_{\ell})\right\}^{2}, 
\ese
\vspace{-4ex}\\
where $\wh{f}_{X,\ell,i}^{Kern}$ is an off-the-shelf kernel density estimator based on $\overline{\bW}_{\ell,1:n}$ as the data points 
and $\wh{f}_{X,\ell,i}(\cdot, \bxi_{\ell})$ is the normalized mixtures of B-splines based estimator proposed in Section \ref{sec: mvt copula density of interest} of the main article. 
Finally, the parameters $\bbeta_{\ell}$'s specifying the probabilities of reporting non-consumptions $P_{\ell}(X_{\ell})$ for the episodic components are initialized by maximizing  
\vspace{-4ex}\\
\bse
\ell(\bbeta_{\ell}\mid \sigma_{\beta,\ell}^{2},\overline{\bW}_{\ell,1:n}, \wh{\bP}_{\ell,1:n}) = -\frac{\bbeta_{\ell}\trans \bP_{\ell} \bbeta_{\ell}}{2\sigma_{\beta,\ell}^{2}} - \sum_{i=1}^{n}\left[\wh{P}_{\ell,i}-\Phi\{\bB_{d,\ell,J_{\ell}}(\overline{W}_{\ell,i}) \bbeta_{\ell}\}\right]^{2}
\ese
\vspace{-4ex}\\
with respect to $\bbeta_{\ell}$, where $\wh{P}_{\ell,i}$ is the proportion of zero recalls for the $\ell\th$ episodic component for $i\th$ individual.

We now discuss how we set the initial values of the sampler for the multivariate method. 
The starting values of the $W_{\ell,i,j}$'s, $X_{\ell,i}$'s, $U_{\ell,i,j}$'s, $\bxi_{\ell}$'s, $\bvartheta_{\ell}$'s, $\bbeta_{\ell}$'s were all set at the corresponding estimates returned by the univariate samplers. 
We set the number of shared atoms of the mixture models for the densities $f_{X,\ell}$ and $f_{\epsilon,\ell}$ at $K_{X} = K_{\epsilon} = \max\{(q+p)\times 5,20\}$.  
We set $\alpha_{X,\ell}=\alpha_{\epsilon,\ell}=1$. 
The atoms of the mixtures of truncated normals for the marginal densities $f_{X,\ell}$ of the regular components are shared, 
so are the atoms of the mixture models for the univariate marginals $f_{\epsilon,\ell}$ of the scaled errors, 
and hence these parameters could not be initialized directly using the univariate model output. 
We initialized these parameters by iteratively sampling them from their posterior full conditionals $100$ times, keeping the estimated $X_{\ell,i}$'s fixed.   
Adopting a similar strategy, we initialized the parameters specifying the densities $f_{\epsilon,\ell}$ of the scaled errors 
by iteratively sampling them from their posterior full conditionals $100$ times, keeping the estimated errors $U_{\ell,i,j}$'s fixed.   
Finally, the parameters specifying $\bR_{\bX}$ and $\bR_{\bepsilon}$ were set at values that correspond to the special case $\bR_{\bX}=\bR_{\bepsilon}=\bI_{q+p}$.

In our sampler for the multivariate problem, 
we first update the parameters specifying the different marginal densities using a pseudo-likelihood that ignores the contribution of the copula. 
The parameters characterizing the copula and the latent $\bX_{i}$'s are then updated using the exact likelihood function conditionally on the parameters obtained in the first step. 
We then update the parameters of the marginal densities again and so forth. 
A more appealing approach would have been to perform joint estimation of the marginal distributions and the copula functions. 
Joint estimation algorithms, most involving carefully designed Metroplis-Hastings (M-H) moves, have been proposed in much simpler settings in \citelatex{pitt2006efficient, wu2014bayesian, wu2015bayesian} etc. 
Designing such moves for our complex deconvolution problem is a daunting task. 
Importantly, the results of \citelatex{dos2008copula} suggest that two-stage approaches often perform just as good as joint estimation procedures, validating their use for practical reasons.

We are now ready to detail our sampler for the multivariate model which iterates between the following steps. 

\begin{enumerate}[topsep=0ex,itemsep=2ex,partopsep=2ex,parsep=0ex, leftmargin=0cm, rightmargin=0cm, wide=3ex]
\item {\bf Updating the parameters specifying $f_{X,\ell}, \ell=1,\dots,q+p$:}
	We modeled the marginals densities of the episodic components $\ell=1,\dots,q$ using normalized mixtures of B-splines 
	and the marginals densities of the regular components $\ell=q+1,\dots,q+p$ using mixtures of truncated normals with shared atoms. 
	\begin{enumerate}[topsep=0ex,itemsep=2ex,partopsep=2ex,parsep=0ex, leftmargin=0cm, rightmargin=0cm, wide=3ex]
	\item {\bf Updating the parameters specifying $f_{X,\ell}, \ell=1,\dots,q$:} 
	The full conditional for each $\bxi_{\ell}$ is given by
	$p(\bxi_{\ell} \mid \bX_{\ell,1:n},\bzeta) \propto p_{0}(\bxi_{\ell}) \times \prod_{i=1}^{n}\prod_{j=1}^{m_{i}}f_{X_{\ell}}(X_{\ell,i}\mid \bxi_{\ell})$.
	We use M-H sampler to update $\bxi_{\ell}$ with random walk proposal
	$q(\bxi_{\ell} \rightarrow \bxi_{\ell,new}) = \MVN(\bxi_{\ell,new}\mid \bxi_{\ell},\bSigma_{\xi,\ell})$.
	We then update the hyper-parameter $\sigma_{\xi,\ell}^{2}$ using its closed-form full conditional
	$(\sigma_{\xi,\ell}^{2}\mid \bxi_{\ell},\bzeta) = \hbox{IG}\{a_{\xi}+(J_{\ell}+2)/2,b_{\xi}+\bxi\trans_{\ell}\bP\bxi_{\ell}/2\}$.
	
	\item {\bf Updating the parameters specifying $f_{X,\ell}, \ell=q+1,\dots,q+p$:}
	The full conditional of $\pi_{X,\ell,k}$ is given by
	\vspace{-4ex}\\
	\bse
	 p(\bpi_{X,\ell} \mid \bzeta) &=& \textstyle \Dir\{\alpha_{X,\ell}+n_{X,\ell}(1),\dots,\alpha_{X,\ell}+n_{X,\ell}(K_{X,\ell})\}.
	\ese
	\vspace{-4ex}\\
	where $n_{X,\ell}(k) = \sum_{i=1}^{n}1(C_{X,\ell,i}=k)$ as before.
	The full conditional of $C_{X,\ell,i}$ is given by  
	\vspace{-4ex}\\
	\bse
	p(C_{X,\ell,i}=k \mid \bzeta) &\propto& \pi_{X,\ell,k} \times \TN(X_{\ell,i}\mid \mu_{X,k},\sigma_{X,k}^{2},[A_{\ell},B_{\ell}]), 
	\ese
	\vspace{-4ex}\\
	a standard multinomial. 
	The full conditional of  $\mu_{X,k}$ is given by
	\vspace{-4ex}\\
	\bse
	\textstyle p(\mu_{X,k}\mid \bzeta) \propto p_{0}(\mu_{X,k}) \times \prod_{\ell=q+1}^{q+p}\prod_{\{i: C_{X,\ell,i}=k\}}\TN(X_{\ell,i}\mid \mu_{X,k},\sigma_{X,k}^{2},[A_{\ell},B_{\ell}]),
	\ese
	\vspace{-4ex}\\
	and the full conditional of  $\sigma_{X,k}^{2}$ is given by
	\vspace{-4ex}\\
	\bse
	\textstyle p(\sigma_{X,k}^{2} \mid \bzeta) \propto p_{0}(\sigma_{X,k}^{2}) \times \prod_{\ell=q+1}^{q+p} \prod_{\{i: C_{X,\ell,i}=k\}}\TN(X_{\ell,i}\mid \mu_{X,k},\sigma_{X,k}^{2},[A_{\ell},B_{\ell}]).
\ese
	\vspace{-4ex}\\
	These parameters are updated by Metropolis-Hastings (MH) steps with the proposals  
	$q(\mu_{X,k}  \to  \mu_{X,k,new})=\Normal(\mu_{X,k,new} \mid \mu_{X,k},\sigma_{X,\ell,\mu}^{2})$ 
	and $q(\sigma_{X,k}^{2}  \to  \sigma_{X,k,new}^{2})=\TN(\sigma_{X,k,new}^{2}\mid \sigma_{X,k}^{2},\sigma_{X,\ell,\sigma}^{2},[\max\{0,\sigma_{X,k}^{2}-1\},\sigma_{X,k}^{2}+1])$, respectively.
	\end{enumerate}

\item {\bf Updating the parameters specifying $f_{\epsilon,\ell}, \ell=1,\dots,2q+p$:} 
For $\ell=1,\dots,q$, $f_{\epsilon,\ell} = \Normal(0,1)$. 
So we only need to update the parameters specifying $f_{\epsilon,\ell}$ for $\ell=q+1,\dots,2q+p$.
With $n_{\epsilon,\ell}(k)=\sum_{i=1}^{n}\sum_{j=1}^{m_{i}}1(C_{\epsilon,\ell,i,j}=k)$, we have
\vspace{-4ex}\\
\bse
(\bpi_{\epsilon,\ell} \mid \bzeta) &\sim& \textstyle \Dir\{1+n_{\epsilon,\ell}(1),\dots,\alpha_{\epsilon,\ell}+n_{\epsilon,\ell}(K_{\epsilon,\ell})\},\\
p(C_{\epsilon,\ell,i,j}=k \mid \bzeta) &\propto& \pi_{\epsilon,\ell,k} \cdot f_{W_{\ell}\mid \wt{X}_{\ell}}(W_{\ell,i,j}\mid p_{\epsilon,k},\tmu_{\epsilon,k},\sigma_{\epsilon,k,1}^{2},\sigma_{\epsilon,k,2}^{2},\bzeta). 
\ese
\vspace{-4ex}\\
For $\ell=q+1,\dots,2q+p$, 
the component specific parameters $(p_{\epsilon,k},\tmu_{\epsilon,k},\sigma_{\epsilon,k,1}^{2},\sigma_{\epsilon,k,2}^{2})$ are updated using M-H steps. 
We propose a new $(p_{\epsilon,k},\tmu_{\epsilon,k},\sigma_{\epsilon,k,1}^{2},\sigma_{\epsilon,k,2}^{2})$ with the proposal
$q\{ \btheta_{\epsilon,k} = (p_{\epsilon,k},\tmu_{\epsilon,k},\sigma_{\epsilon,k,1}^{2},\sigma_{\epsilon,k,2}^{2})  \rightarrow (p_{\epsilon,k,new},\tmu_{\epsilon,k,new},\sigma_{\epsilon,k,1,new}^{2},\sigma_{\epsilon,k,2,new}^{2}) = \btheta_{\epsilon,k,new}\} =
\hbox{TN}(p_{\epsilon,k,new}\mid p_{\epsilon,k},\sigma_{p,\epsilon,\ell}^{2},[0,1])~\times~
\Normal(\tmu_{\epsilon,k,new}\mid \tmu_{\epsilon,k}, \sigma_{\epsilon,\ell,\tmu}^{2})~\times~
\hbox{TN}(\sigma_{\epsilon,k,1,new}^{2}\mid \sigma_{\epsilon,k,1}^{2},\sigma_{\epsilon,\ell,\sigma}^{2},[\max\{0,\sigma_{\epsilon,k,1}^{2}-1\},\sigma_{\epsilon,k,1}^{2}+1])~\times~
\hbox{TN}(\sigma_{\epsilon,k,2,new}^{2}\mid \sigma_{\epsilon,k,2}^{2},\sigma_{\epsilon,\ell,\sigma}^{2},[\max\{0,\sigma_{\epsilon,k,2}^{2}-1\},\sigma_{\epsilon,k,2}^{2}+1])$.
We update $\btheta_{k}$ to the proposed value $\btheta_{k,new}$ with probability 
\vspace{-4ex}\\
\bse
\min\bigg\{1, \frac{q(\btheta_{\epsilon,k,new} \rightarrow \btheta_{\epsilon,k})}
		      {q(\btheta_{\epsilon,k} \rightarrow \btheta_{\epsilon,k,new})}
	    \frac{\prod_{\ell=q+1}^{2q+p}\prod_{\{i,j: C_{\epsilon,\ell,i,j}=k\}} f_{W_{\ell}\mid \wt{X}_{\ell}}(W_{\ell,i,j}\mid \btheta_{\epsilon,k,new},\bzeta) ~ p_{0}(\btheta_{\epsilon,k,new})}
	      {\prod_{\ell=q+1}^{2q+p} \prod_{\{i,j: C_{\epsilon,\ell,i,j}=k\}}f_{W_{\ell}\mid \wt{X}_{\ell}}(W_{\ell,i,j}\mid \btheta_{\epsilon,k},\bzeta) ~ p_{0}(\btheta_{\epsilon,k})}\bigg\}.
\ese
\vspace{-5ex}

\item {\bf Updating the parameters specifying $s_{\ell}$ for $\ell=q+1,\dots,2q+p$: } 
The full conditional for each $\bvartheta_{\ell}$ is given by
$p(\bvartheta_{\ell} \mid \bW_{\ell,1:N},\bzeta) \propto p_{0}(\bvartheta_{\ell}) \times \prod_{i=1}^{n}\prod_{j=1}^{m_{i}}f_{W_{\ell}\mid \wt{X}_{\ell}}(W_{\ell,i,j}\mid \bvartheta_{\ell},\bzeta)$.
We use M-H sampler to update $\bvartheta_{\ell}$ with random walk proposal
$q(\bvartheta_{\ell} \rightarrow \bvartheta_{\ell,new}) = \MVN(\bvartheta_{\ell,new}\mid \bvartheta_{\ell},\bSigma_{\vartheta,\ell})$.
We then update the hyper-parameter $\sigma_{\vartheta,\ell}^{2}$ using its closed-form full conditional
$(\sigma_{\vartheta,\ell}^{2}\mid \bvartheta_{\ell},\bzeta) = \hbox{IG}\{a_{\vartheta}+(J_{\ell}+2)/2,b_{\vartheta}+\bvartheta\trans_{\ell}\bP\bvartheta_{\ell}/2\}$.

\item {\bf Updating latent $W_{\ell,i,j}$'s for $\ell=1,\dots,2q$:}
For $\ell=1,\dots,q$, $W_{\ell,i,j}$ are all latent, 
whereas for $\ell=q+1,\dots,2q$, the variables $W_{\ell,i,j}$ are not observed when $Y_{\ell,i,j} = 0$. 
\begin{enumerate}[topsep=2ex,itemsep=0ex,partopsep=2ex,parsep=0ex, leftmargin=0cm, rightmargin=0cm, wide=3ex]
\item {\bf Updating $W_{\ell,i,j}$ for $\ell=1,\dots,q$:}
The log full conditional of $W_{\ell,i,j}$ is given by
\vspace{-4ex}\\
\bse
&& \hspace{-1cm} \log (W_{\ell,i,j}\mid \bzeta) = \log \{Y_{\ell,i,j} ~ 1(W_{\ell,i,j}>0) + (1-Y_{\ell,i,j}) ~ 1(W_{\ell,i,j}<0)\} - \{W_{\ell,i,j}- h({X}_{\ell,i})\}^{2}/2.
\ese
\vspace{-4ex}\\
It thus follows that 
\vspace{-4ex}\\
\bse
(W_{\ell,i,j}\mid \bzeta) \sim Y_{\ell,i,j} ~ \TN\{W_{\ell,i,j}\vert h({X}_{\ell,i}), 1, [0,\infty)\} + (1-Y_{\ell,i,j}) ~ \TN\{W_{\ell,i,j} \vert h({X}_{\ell,i}), 1, (-\infty,0]\}.
\ese

\item {\bf Updating latent $W_{\ell,i,j}$ for $\ell=(q+1),\dots,2q$:}
Given $C_{\epsilon,\ell,i,j}=k$ and $p_{\epsilon,k}$, we sample $C2_{\epsilon,\ell,i,j}$ as 
\vspace{-6ex}\\
\bse
C2_{\epsilon,\ell,i,j} \sim \Bernoulli(p_{\epsilon,k})+1.
\ese
\vspace{-4ex}\\
Given $C_{\epsilon,\ell,i,j}=k$ and $C2_{\epsilon,\ell,i,j}=t$, keeping the dependence on $k$ and $t$ implicit, 
define 
$X_{\ell,i,j}^{\nabla} = \wt{X}_{\ell,i}+s_{\ell}(\wt{X}_{\ell,i},\bvartheta_{\ell})\mu_{\epsilon,k,t}$ and 
$\sigma_{u,\ell,i,j}^{2\nabla} = s_{\ell}^{2}(\wt{X}_{\ell,i},\bvartheta_{\ell})\sigma_{\epsilon,k,t}^{2}$. 
The full conditional of $W_{\ell,i,j}$ is given by
\vspace{-6ex}\\
\bse
&& \hspace{-1cm} (W_{\ell,i,j} \mid \bzeta, C_{\epsilon,\ell,i,j}=k, C2_{\epsilon,\ell,i,j}=t) \sim \Normal (W_{\ell,i,j} \mid X_{\ell,i,j}^{\nabla},\sigma_{u,\ell,i,j}^{2\nabla}).
\ese
\end{enumerate}

\item {\bf Updating $\bbeta_{\ell}$ for $\ell=1,\dots,q$:}
The full conditionals of $\bbeta_{\ell}$ are available in closed form as 
\vspace{-7ex}\\
\bse
&& p(\bbeta_{\ell} \mid \bzeta) = \MVN_{J_{\ell}}(\bmu_{\beta,\ell},\bSigma_{\beta,\ell}), ~~~\text{where}\\
&& \bmu_{\beta,\ell} = \bSigma_{\beta,\ell} \{ \bSigma_{\beta,\ell,0}^{-1} \bmu_{\beta,\ell,0} + \textstyle\sum_{i=1}^{n} \sum_{j=1}^{m_{i}} W_{\ell,i,j} \bB_{d,\ell,J_{\ell}}(\wt{X}_{\ell,i})\trans \},  \\
&& \bSigma_{\beta,\ell}= \{ \sigma_{\beta,\ell}^{-2}\bP_{\ell} + \bSigma_{\beta,\ell,0}^{-1} + \textstyle\sum_{i=1}^{n} m_{i} \bB_{d,\ell,J_{\ell}}(\wt{X}_{\ell,i})\trans \bB_{d,\ell,J_{\ell}}(\wt{X}_{\ell,i}) \}^{-1}.
\ese 
\vspace{-4ex}

\item {\bf Updating the values of $\bX$:}
The full conditionals for $\bX_{i}$ are given by
\vspace{-4ex}\\
\bse
&&\hspace{-1cm}(\bX_{i}\mid \bzeta) 
\propto f_{\bX}(\bX_{i}\mid \bzeta) \times \textstyle\prod_{j=1}^{m_{i}} f_{\bW\mid \wt\bX}(\bW_{i,j} \mid \wt\bX_{i}, \bzeta)  \\
&& \textstyle = |\bR_{\bX}|^{-1/2} \exp\left\{-\frac{1}{2}\bY_{\bX,i}\trans(\bR_{\bX}^{-1}-\bI_{q+p})\bY_{\bX,i}\right\}\prod_{\ell=1}^{q+p}f_{X,\ell}(X_{\ell,i}\mid \bzeta)  \\
&& ~~~~~~ \textstyle\times ~ \prod_{j=1}^{m_{i}} \left[ |\bR_{\bepsilon}|^{-1/2} \exp\left\{-\frac{1}{2}\bY_{\epsilon,i,j}\trans(\bR_{\bepsilon}^{-1}-\bI_{2q+p})\bY_{\epsilon,i,j}\right\}\prod_{\ell=1}^{2q+p}f_{W_{\ell}\mid \wt{X}_{\ell}}(W_{\ell,i,j}\mid \wt{X}_{\ell,i},\bzeta) \right],
\ese
\vspace{-4ex}\\
where $F_{X,\ell}(X_{\ell,i}\mid\bzeta)=\Phi(Y_{X,\ell,i})$ and $F_{\epsilon,\ell}\{(W_{\ell,i,j}-\wt{X}_{\ell,i})/s_{\ell}(\wt{X}_{\ell,i})\mid\bzeta\}=\Phi(Y_{\epsilon,\ell,i,j})$. 
The full conditionals do not have closed forms. 
MH steps with independent truncated normal proposals for each component are used within the Gibbs sampler.

\item {\bf Updating the parameters specifying the copula: } 
We have $F_{X,\ell}(X_{\ell,i}\mid\bzeta)=\Phi(Y_{X,\ell,i})$ for all $i=1,\dots,n$ and $\ell=1,\dots,(q+p)$.
Conditionally on the parameters  specifying the marginals, $\bY_{\bX,1:(q+p),1:n}$ are thus known quantities. 
We plug-in these values and use that $(\bY_{\bX,i}\mid\bR_{\bX}) \sim \MVN_{q+p}(\bzero,\bR_{\bX})$ to update $\bR_{\bX}$.  
The full conditionals of the parameters specifying $\bR_{\bX}$ do not have closed forms. 
We use M-H steps to update these parameters. 
\begin{enumerate}[topsep=2ex,itemsep=0ex,partopsep=2ex,parsep=0ex, leftmargin=0cm, rightmargin=0cm, wide=3ex]
\item
For $t=1,\dots,(q+p-1)$, we discretized the values of $b_{X,t}$ to the set $\{-0.99+2\times 0.99(m-1)/(M-1)\}$, where $m=1,\dots,M$ and we chose $M=41$. 
A new value $b_{X,t,new}$ is proposed at random from the set comprising the current value of $b_{X,t}$ and its two neighbors. 
The proposed value is accepted with probability $\min\{1,a(b_{X,t,new})/a(b_{X,t})\}$, where 
\vspace{-4ex}\\
\bse
\textstyle a(b_{X,t}) = (1-b_{X,t}^{2})^{-n/2} \times \exp\left\{- (1/2)\sum_{i=1}^{n}\sum_{j=1}^{m_{i}}\bY_{\bX,i,j}\trans\{\bSigma_{\bX}(b_{X,t},\bzeta)\}^{-1}\bY_{\bX,i,j}\right\}. 
\ese
\vspace{-8ex}\\
\item
For $s=1,\dots,(q+p-1)(q+p-2)/2$, we discretized the values of $\theta_{X,s}$ to the set $\{-3.14+2\times 3.14 (m-1)/(M-1)\}$, where $m=1,\dots,M$ and $M=41$. 
A new value $\theta_{s,new}$ is proposed at random from the set comprising the current value and its two neighbors. 
The proposed value is accepted with probability $\min\{1,a(\theta_{X,s,new})/a(\theta_{X,s})\}$, where 
\vspace{-4ex}\\
\bse
\textstyle a(\theta_{X,s}) = \exp\left\{- (1/2)\sum_{i=1}^{n}\sum_{j=1}^{m_{i}}\bY_{\bX,i,j}\trans\{\bSigma_{\bX}(\theta_{X,s},\bzeta)\}^{-1}\bY_{\bX,i,j}\right\}. 
\ese
\vspace{-7ex}
\end{enumerate}
The parameters specifying $\bR_{\bepsilon}$ are updated in a similar fashion.

\end{enumerate}

With carefully chosen initial values and proposal densities for the MH steps, we were able to achieve quick convergence for the MCMC samplers. 
For our proposed method, $5,000$ MCMC iterations were run in each case with the initial $3,000$ iterations discarded as burn-in. 
The remaining samples were further thinned by a thinning interval of $5$. 
We programmed in {R}.
With $n=1000$ subjects and $m_{i}=3$ proxies for each subject, on an ordinary desktop, $5,000$ MCMC iterations required approximately $3$ hours to run.

\clearpage
\section{Additional Figures}\label{sec: additional figs}

\begin{figure}[!ht]
\centering
\includegraphics[height=17cm, width=17cm, trim=0cm 0cm 0cm 0cm, clip=true]{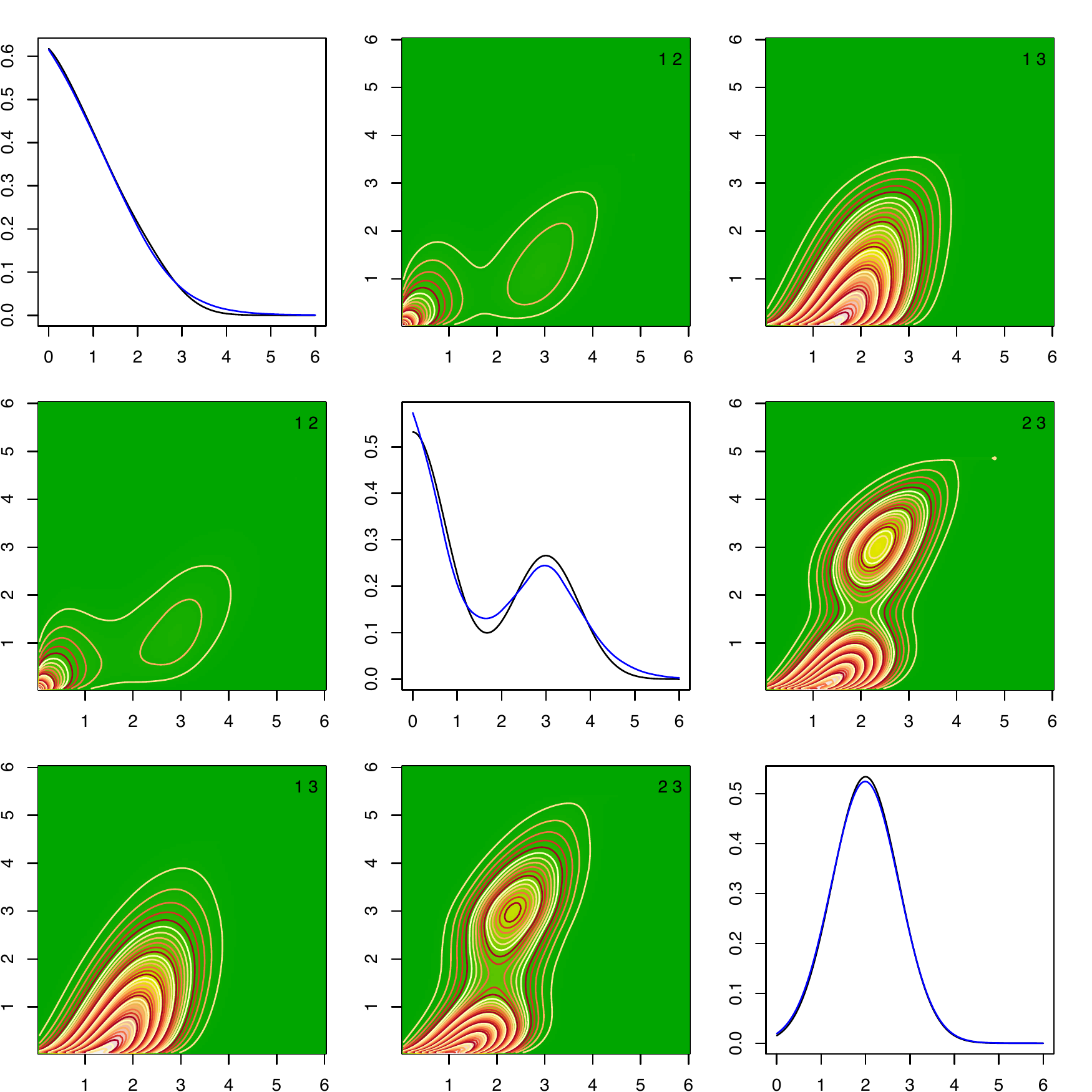}
\caption{\baselineskip=10pt 
Results for simulated data sets with sample size $n = 1000$, $q=2$ episodic components and $p=1$ regular components, each subject having $m_{i}=3$ replicates, for the data set corresponding to the 25th percentile $3$-dimensional ISE. 
The off-diagonal panels show the contour plots of the true two-dimensional marginals (upper triangular panels) and the corresponding estimates obtained by our method (lower triangular panels).
The numbers $i,j$ at the top right corners indicate which marginal densities $f_{X_{i},X_{J-1}}$ are plotted in those panels. 
The diagonal panels show the true (in black) one dimensional marginal densities and the corresponding estimates (in blue) produced by our method. 
Compare with Figure \ref{fig: SimStudy5}.
}
\label{fig: SimStudy4}
\end{figure}

\begin{figure}[!ht]
\centering
\includegraphics[height=17cm, width=17cm, trim=0cm 0cm 0cm 0cm, clip=true]{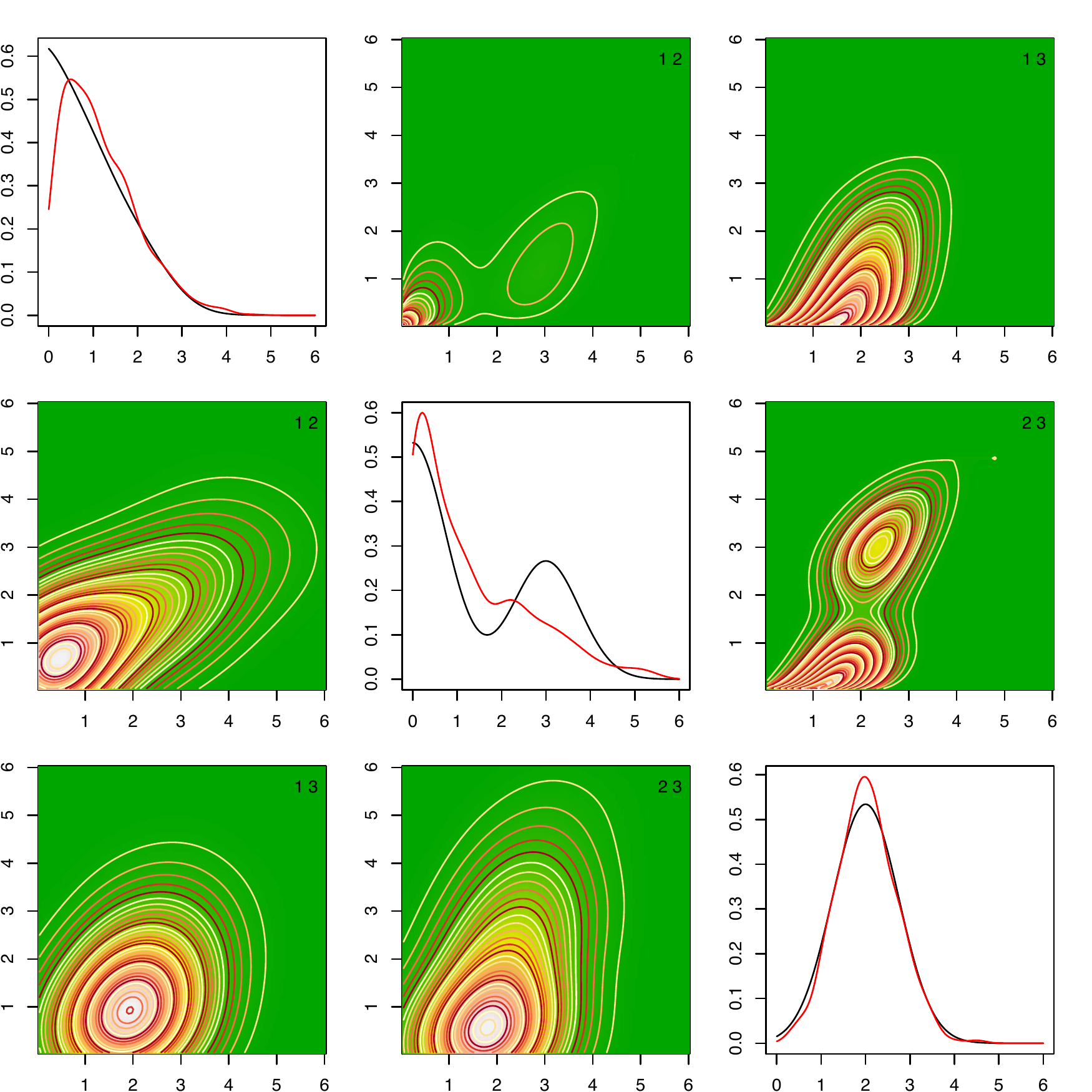}
\caption{\baselineskip=10pt 
Results for simulated data sets with sample size $n = 1000$, $q=2$ episodic components and $p=1$ regular components, each subject having $m_{i}=3$ replicates, for the data set corresponding to the 25th percentile $3$-dimensional ISE. 
The off-diagonal panels show the contour plots of the true two-dimensional marginals (upper triangular panels) and the corresponding estimates obtained by the method of \cite{Zhang2011b} (lower triangular panels).
The numbers $i,j$ at the top right corners indicate which marginal densities $f_{X_{i},X_{J-1}}$ are plotted in those panels. 
The diagonal panels show the true (in black) one dimensional marginal densities and the corresponding estimates (in red) produced by the method of \cite{Zhang2011b}.
Compare with Figure \ref{fig: SimStudy4}.
}
\label{fig: SimStudy5}
\end{figure}

\begin{figure}[!ht]
\centering
\includegraphics[height=17cm, width=17cm, trim=0cm 0cm 0cm 0cm, clip=true]{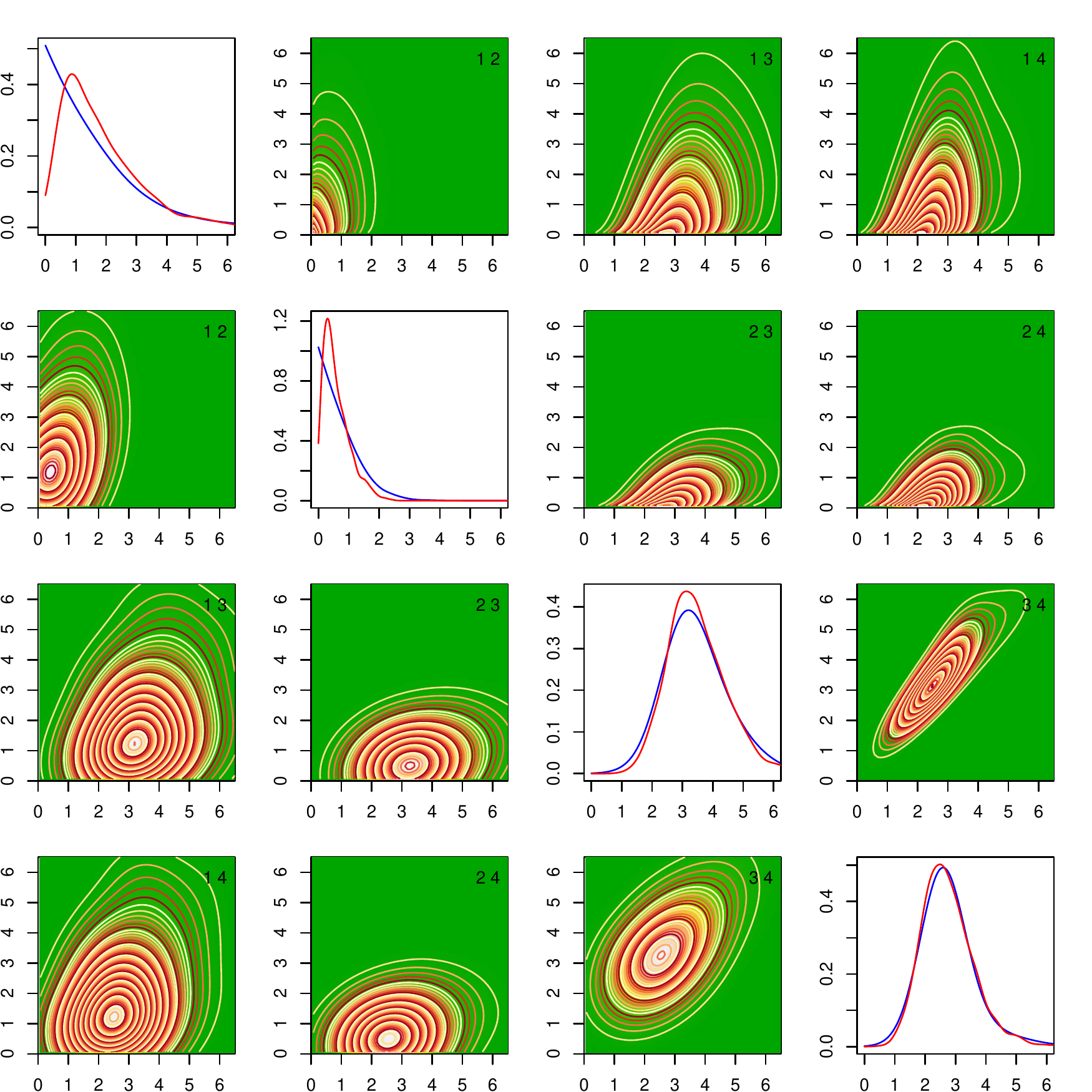}
\caption{\baselineskip=10pt 
Results for the EATS data sets with sample size $n = 965$, $q=2$ episodic components, milk and whole grains, and $p=2$ regular components, sodium and energy, each subject having $m_{i}=4$ replicates. 
The off-diagonal panels show the contour plots of two-dimensional marginals estimated by our method (upper triangular panels) and the method of \cite{Zhang2011b} (lower triangular panels).
The numbers $i,j$ at the top right corners indicate which marginal densities $f_{X_{i},X_{J-1}}$ are plotted in those panels. 
The diagonal panels show the one dimensional marginal densities estimated by our method (in blue) 
and the method of \cite{Zhang2011b} (in red).
}
\label{fig: EATS4}
\end{figure}

\begin{figure}[!ht]
\centering
\includegraphics[height=8cm, width=16cm, trim=0cm 0cm 0cm 0cm, clip=true]{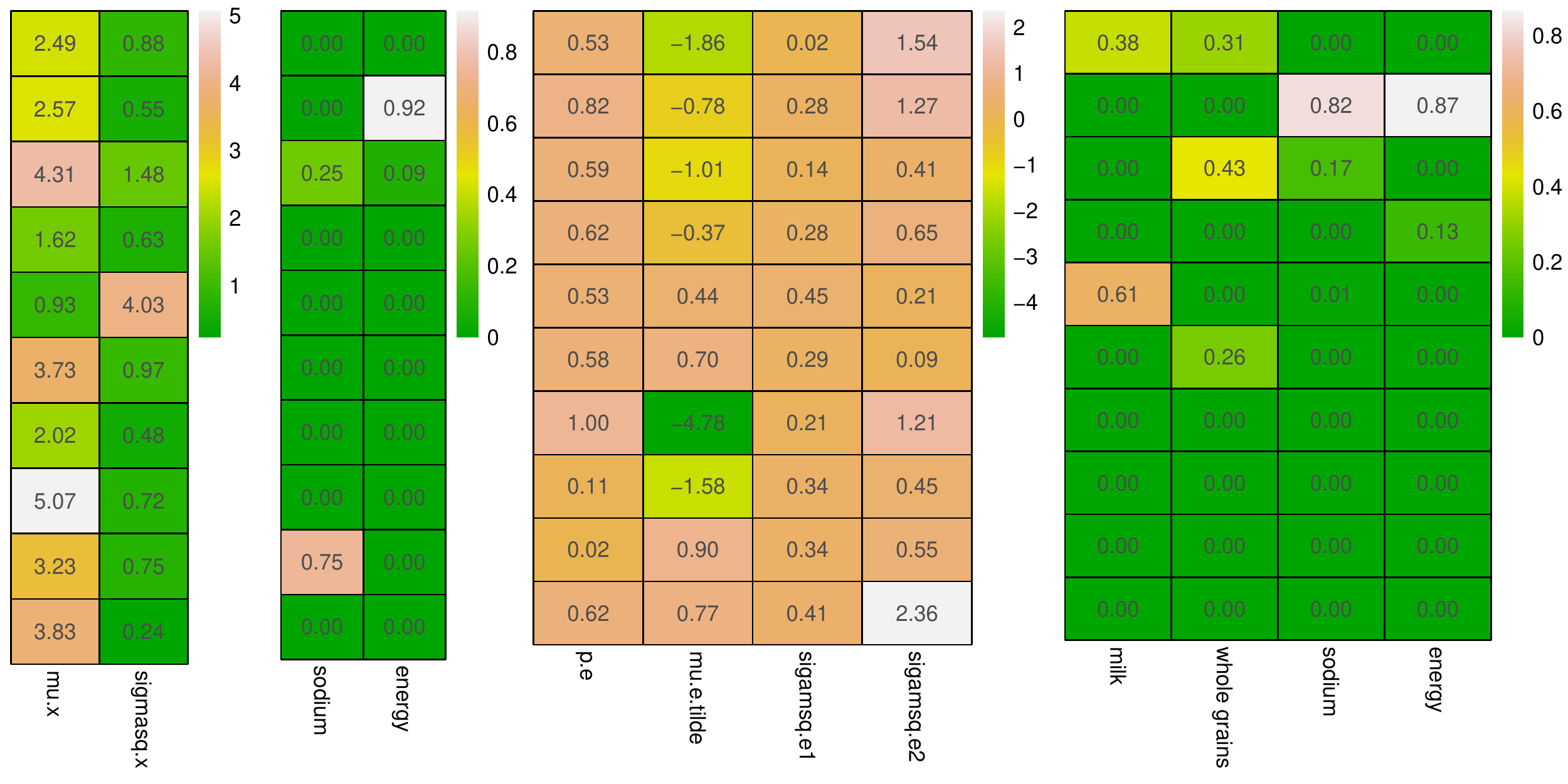}
\caption{\baselineskip=10pt 
Results for the EATS data sets with sample size $n = 965$, $q=2$ episodic components, milk and whole grains, and $p=2$ regular components, sodium and energy, each subject having $m_{i}=4$ replicates. 
These results correspond to the final MCMC iteration but are representative of other iterations in steady state. 
The left panel shows the mixture component specific parameters $\{\mu_{X,k}, \sigma_{X,k}^{2}\}_{k=1}^{K_{X}=10}$ 
used to model the marginal densities $f_{X,\ell}(X_{\ell})$ of the two regular dietary components.  
The middle left panel shows the associated `empirical' mixture probabilities $\wh\pi_{X,\ell,k} = \sum_{i=1}^{n}1\{C_{X,\ell,i}=k\}/n$. 
Only the mixture components $2,3$ and $9$ were actually used to model the densities 
and these mixture components were shared between the two dietary components, 
the other mixture components were redundant. 
The middle right panel shows the mixture component specific parameters $\{p_{\epsilon,k}, \mu_{\epsilon,k}, \sigma_{\epsilon,k,1}^{2}, \sigma_{\epsilon,k,2}^{2}\}_{k=1}^{K_{\epsilon}=10}$ used to model the marginal densities $f_{\epsilon,\ell}(\epsilon_{\ell})$ of the scaled errors of all dietary components. 
The right panel shows the associated `empirical' mixture probabilities $\wh\pi_{\epsilon,\ell,k} = \sum_{i=1}^{n}\sum_{j=1}^{m_{i}}1\{C_{\epsilon,\ell,i,j}=k\}/\sum_{i=1}^{n}m_{i}$. 
Only the mixture components $1,2,3,4,5$ and $6$ were actually used to model the densities 
and these mixture components were shared between the four dietary components, 
the other mixture components were redundant.  
}
\label{fig: EATS5 Z Tables}
\end{figure}

\begin{figure}[!ht]
\centering
\includegraphics[height=16cm, width=16cm, trim=0cm 0cm 0cm 0cm, clip=true]{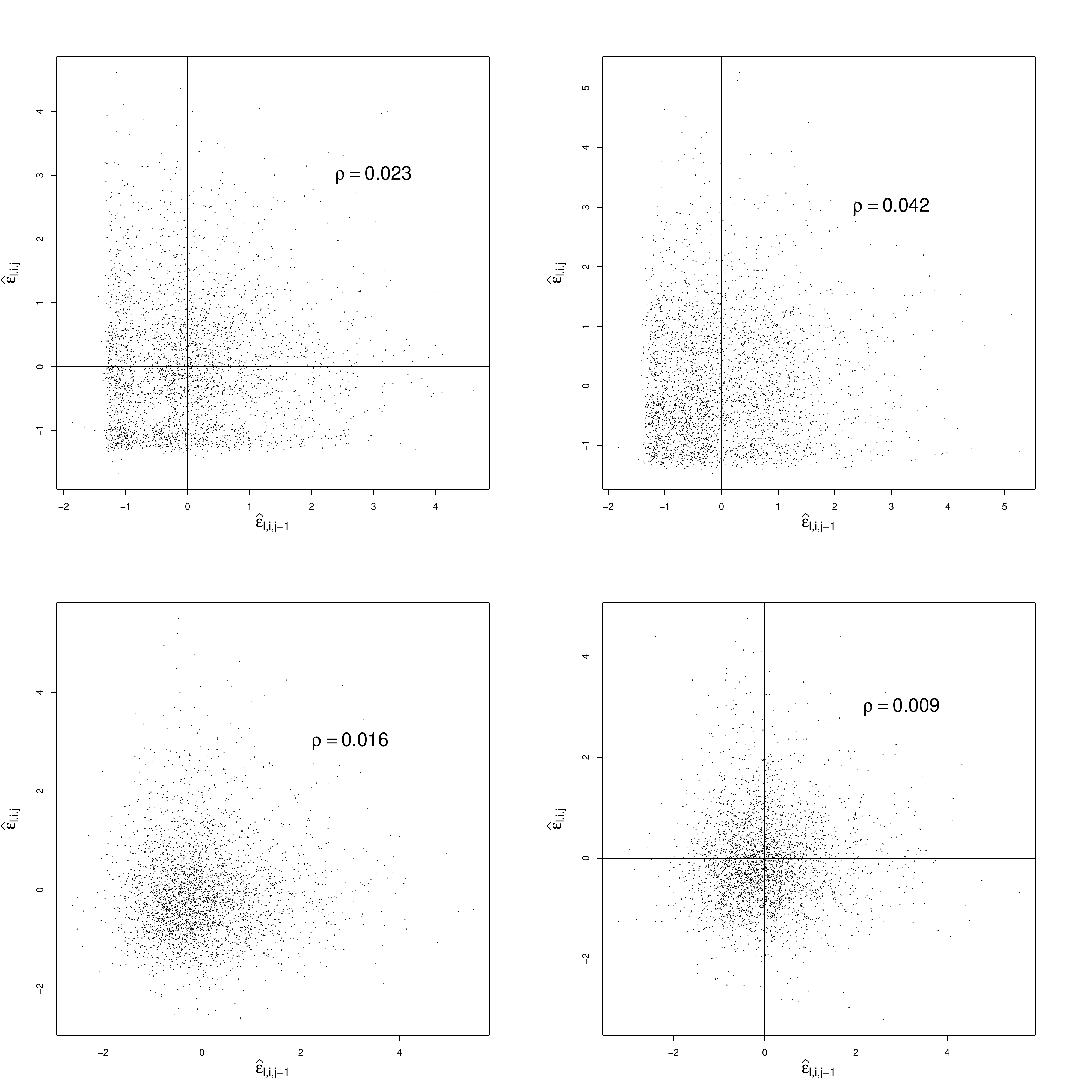}
\caption{\baselineskip=10pt 
Results for the EATS data sets with sample size $n = 965$, $q=2$ episodic components, milk and whole grains, and $p=2$ regular components, sodium and energy, each subject having $m_{i}=4$ replicates. 
These results correspond to the final MCMC iteration but are representative of other iterations in steady state. 
The panels show the scatterplots of scaled `residuals' ($\wh\epsilon_{\ell,i,j-1}$ vs $\wh\epsilon_{\ell,i,j}$, $j=2,3,4$) at adjacent sampling occasions 
and their estimated correlation coefficients for milk (top left), whole grains (top right), sodium (bottom left) and energy (bottom right).  
}
\label{fig: EATS6 residual Tables}
\end{figure}

\clearpage
\section{Additional Simulation Experiments}\label{sec: additional sims}

In this section, we discuss our findings in some additional small scale numerical experiments 
when we tried to simulate from the model of \cite{Zhang2011b}. 

Mimicking the \emph{exact} \cite{Zhang2011a,Zhang2011b} models is, however, a daunting, if not impossible, task. 
Even \cite{Zhang2011a,Zhang2011b} themselves did attempt this. 
We explain below. 

We recall from Section \ref{sec: mvt copula comparison with NCI} that \cite{Zhang2011b} is a transformation-retransformation based method
that assumes additivity of Box-Cox transformed latent consumptions $\bX_{tr}$ and associated errors $\bU_{tr}$, 
hence unbiasedness of the transformed proxies $\bW_{tr}$ for $\bX_{tr}$, 
multivariate normality of $\bX_{tr}$ and $\bU_{tr}$, 
homoscedasticity of $\bU_{tr}$, 
independence of $\bU_{tr}$ from $\bX_{tr}$ etc. 
As discussed in detail in \cite{Sarkar_etal:2014}, there may be at least three different ways to make these transformations, 
none of which may achieve exact multivariate normality, 
additivity, homoscedasticity and independence simultaneously.
Also, the transformation parameters, namely the Box-Cox coefficients $\lambda$, the means $\mu(\lambda)$ and the scales $\sigma(\lambda)$, 
are all determined using the observed recalls. 
It is not clear how we can reverse engineer the process with predetermined values of these parameters that will result in realistic dietary recall data. 

Additionally, the assumption of unbiasedness is most meaningful in the original observed scale, as our method assumes, 
and not in any arbitrarily chosen nonlinear scale, as \cite{Zhang2011b} assumes. 
Simulating from scenarios that assume additivity in the transformed scale will thus be unfair to our proposed model. 
It is mathematically impossible to design a model that satisfies the assumptions of unbiasedness and additivity in both the original and a nonlinearly transformed scale.

\vspace{-10pt}
\subsubsection*{\bf \cite{Zhang2011b} 
Model in a Log-Transformed Scale:} To alleviate these issues while trying to simulate scenarios that conform to both the assumptions of \cite{Zhang2011b} 
as closely as possible, 
we let $\bW_{tr,i,j} = (W_{tr,i,j,1},\dots,W_{tr,i,j,2q+p})\trans$ be related to the observed recalls $\bY_{i,j}$ as 
\vspace{-4ex}\\
\be
& \hskip -20pt Y_{\ell,i,j} = \Ind(W_{\ell,i,j}>0),~~~~~W_{tr,\ell,i,j} = \wt{X}_{tr,\ell,i} + U_{tr,\ell,i,j}, ~~~\hbox{for}~\ell=1,\dots,q, \nonumber\\
& \hskip -25pt \{g_{tr}(Y_{\ell,i,j}) \mid Y_{\ell-q,i,j}=1\} = W_{tr,\ell,i,j} = \wt{X}_{tr,\ell,i} + U_{tr,\ell,i,j} - \sigma_{tr,u,\ell}^{2}/2, ~~~\hbox{for}~\ell=q+1,\dots,2q,    \label{eq: Zhang 2011b sim model}\\
& \hskip -20pt g_{tr}(Y_{\ell,i,j}) = W_{tr,\ell,i,j} = \wt{X}_{tr,\ell,i} + U_{tr,\ell,i,j} - \sigma_{tr,u,\ell}^{2}/2, ~~~~~~~\hbox{for}~\ell=2q+1,\dots,2q+p, \nonumber
\ee 
\vspace{-4ex}\\
where, $g_{tr}(Y) = \log Y$, 
$\wt\bX_{tr,i,j} = (\wt{X}_{tr,1,i,j},\dots,\wt{X}_{tr,2q+p,i,j})\trans \sim \MVN_{2q+p}(\bmu_{tr,x},\bSigma_{tr,x})$ 
with $\bmu_{tr,x} = (\mu_{tr,x,1},\dots,\mu_{tr,x,2q+p})\trans$ and $\bSigma_{tr,x} = ((\sigma_{tr,x,\ell,\ell'})$,  
and $\bU_{tr,i,j} = (U_{tr,1,i,j},\dots,U_{tr,2q+p,i,j})\trans \sim \MVN_{2q+p}(\bzero,\bSigma_{tr,u})$ 
where $\bSigma_{tr,u}=((\sigma_{tr,u,\ell,\ell'}))$, $\diag(\bSigma_{tr,u}) = (\sigma_{tr,u,1}^{2},\dots,\sigma_{tr,u,2q+p}^{2})\trans$ 
with $\sigma_{tr,u,1}^{2}=\dots=\sigma_{tr,u,q}^{2}=1$, 
independently of $\wt\bX_{tr,i,j}$. 
We adjust for the terms $\sigma_{tr,u,\ell}^{2}/2$ while transforming back to the original scales. 
To make the situation further favorable for \cite{Zhang2011b}, 
we also assume that the transformation $g_{tr}(Y) = \log Y$ to be known.

\vspace{-10pt}
\subsubsection*{\bf An Equivalent Model in the Original Scale:} 
The log transformation plays a special role here as we can reformulate model (\ref{eq: Zhang 2011b sim model}) above as  
\vspace{-4ex}\\
\bse
&& \hskip -20pt W_{\ell,i,j} = W_{tr,\ell,i,j} = \wt{X}_{tr,\ell,i} + U_{tr,\ell,i,j} = \wt{X}_{\ell,i} + U_{\ell,i,j},~~~\hbox{for}~\ell=1,\dots,q,\\
&& \hskip -20pt W_{\ell,i,j} = \exp(W_{tr,\ell,i,j}) = \exp(\wt{X}_{tr,\ell,i}) \exp(U_{tr,\ell,i,j} - \sigma_{tr,u,\ell}^{2}/2) = \wt{X}_{\ell,i} \wt{U}_{\ell,i,j},~~~\hbox{for}~\ell=q+1,\dots,2q,\\
&& \hskip -20pt W_{\ell,i,j} = \exp(W_{tr,\ell,i,j}) = \exp(\wt{X}_{tr,\ell,i}) \exp(U_{tr,\ell,i,j} - \sigma_{tr,u,\ell}^{2}/2) = \wt{X}_{\ell,i} \wt{U}_{\ell,i,j},~~~\hbox{for}~\ell=2q+1,\dots,2q+p.  
\ese
\vspace{-4ex}\\
Suppressing the indices $i,j$ for cleaner notation, we have $\eE(\wt{U}_{\ell}) = \eE\{\exp(U_{tr,\ell}-\sigma_{tr,u,\ell}^{2}/2)\}=1$ for $\ell=q+1,\dots,2q+p$. 
The $W_{\ell}$'s, as defined above, may thus be viewed as surrogates for $\wt{X}_{\ell} = \wt{X}_{tr,\ell}, \ell=1,\dots,q$, with additive errors $U_{\ell}$, 
and $\wt{X}_{\ell} = \exp(\wt{X}_{tr,\ell}), \ell=q+1,\dots,2q+p$ 
with multiplicative measurement errors $\wt{U}_{\ell}$. 
As shown in \cite{sarkar2018bayesian}, multiplicative measurement error models can be reformulated as additive models with conditionally heteroscedastic errors as 
$W_{\ell} = \wt{X}_{\ell} \wt{U}_{\ell} =  \wt{X}_{\ell} + \wt{X}_{\ell}(\wt{U}_{\ell}-1) = \wt{X}_{\ell} + s_{\ell}(\wt{X}_{\ell})\epsilon_{\ell} =  \wt{X}_{\ell} + U_{\ell}$ 
with $U_{\ell} = s_{\ell}(\wt{X}_{\ell})\epsilon_{\ell}$, $s_{\ell}(\wt{X}_{\ell}) = \wt{X}_{\ell}$, 
$\epsilon_{\ell}=(\wt{U}_{\ell}-1)$, $\epsilon_{\ell}$ independent from $\wt{X}_{\ell}$ and $\eE(\epsilon_{\ell})=0$. 
We have $\var(U_{\ell} \vert X_{\ell}) = X_{\ell}^{2}\var(\wt{U}_{\ell}) = X_{\ell}^{2}\{\exp(\sigma_{tr,u,\ell}^{2})-1\}$. 
We can therefore reformulate model (\ref{eq: Zhang 2011b sim model}) to closely resemble our proposed model from Section \ref{sec: models} in the main paper as 
\vspace{-4ex}\\
\bse
&& \hskip -20pt Y_{\ell,i,j} = \Ind(W_{\ell,i,j}>0), ~~~ W_{\ell,i,j} = \wt{X}_{\ell,i} + U_{\ell,i,j},~~~\hbox{for}~\ell=1,\dots,q,\\
&& \hskip -20pt Y_{\ell,i,j} = Y_{\ell-q,i,j} W_{\ell,i,j}, ~~~ W_{\ell,i,j} = \wt{X}_{\ell,i} + U_{\ell,i,j},~~~\hbox{for}~\ell=q+1,\dots,2q,\\
&& \hskip -20pt Y_{\ell,i,j} = W_{\ell,i,j}, ~~~W_{\ell,i,j} = \wt{X}_{\ell,i} + U_{\ell,i,j},~~~\hbox{for}~\ell=2q+1,\dots,2q+p.  
\ese
\vspace{-6ex}\\
In the original scale, we now also have  
\vspace{-4ex}\\
\bse
&& \hskip -20pt X_{\ell} = \eE(Y_{\ell+q} \vert \wt{X}_{tr,\ell},\wt{X}_{tr,q+\ell}) = \Phi(\wt{X}_{tr,\ell}) \exp(\wt{X}_{tr,q+\ell}) \eE(U_{tr,q+\ell} - \sigma_{tr,u,q+\ell}^{2}/2) = \Phi(\wt{X}_{tr,\ell}) \exp(\wt{X}_{tr,q+\ell}),\\
&& \hspace{13cm}\ell=1,\dots,q,\\
&& \hskip -20pt X_{\ell} = \eE(Y_{\ell+q} \vert \wt{X}_{tr,q+\ell}) = \exp(\wt{X}_{tr,q+\ell}) \eE(U_{tr,q+\ell} - \sigma_{tr,u,q+\ell}^{2}/2) = \exp(\wt{X}_{tr,q+\ell}), ~~~\ell=q+1,\dots,q+p.
\ese
\vspace{-4ex}\\
One difference with our proposed model that still remains is that, for $\ell=1,\dots,q$, the probability of reporting a positive recall is now $\Phi(\wt{X}_{\ell})$, 
where $\wt{X}_{\ell}=\wt{X}_{tr,\ell}$ is correlated with 
$X_{\ell}$, now a function of $\wt{X}_{tr,\ell}$ and $\wt{X}_{tr,q+\ell}$, 
but the probability does not directly depend on $X_{\ell}$.

To find out the true joint and marginal densities of $\bX$, we first note that, by construction, $\exp(\bX_{tr}) \sim \MVLN_{2q+p}(\bmu_{tr,x},\bSigma_{tr,x})$, 
where $\MVLN_{d}(\bmu,\bSigma)$ denotes a $d$-dimensional multivariate lognormal distribution with mean $\bmu$ and covariance $\bSigma$ in the log scale. 
This implies that, marginally, $X_{\ell} = \exp(X_{tr,q+\ell}) \sim \LN(\mu_{tr,x,q+\ell},\sigma_{tr,x,q+\ell,q+\ell})$ for the regular components $\ell=q+1,\dots,q+p$, 
where $\LN(\mu,\sigma^{2})$ denotes a univariate lognormal distribution with mean $\mu$ and variance $\sigma^{2}$ in the log scale. 
Finding the marginal distributions for the episodic components $X_{\ell}, \ell=1,\dots,q$, 
and the joint distribution of the episodic and the regular components $\bX=(X_{1},\dots,X_{q+p})\trans$ in the original scale is, however, not so straightforward. 

Taking the transformation $\wt\bX_{tr} \to \bZ_{tr} = (Z_{tr,1},\dots,Z_{tr,q},Z_{tr,q+1},\dots,Z_{tr,2q+p})\trans$ with 
\vspace{-4ex}\\
\bse
&& Z_{tr,\ell} = \wt{X}_{tr,\ell}, ~~~~~~~\ell=1,\dots,q,\\ 
&& Z_{tr,\ell} = X_{\ell-q} = \Phi(\wt{X}_{tr,\ell-q})\exp(\wt{X}_{tr,\ell}), ~~~~~~~\ell=q+1,\dots,2q,\\ 
&& Z_{tr,\ell} = X_{\ell-q} = \exp(\wt{X}_{tr,\ell}), ~~~~~~~\ell=2q+1,\dots,2q+p
\ese
\vspace{-4ex}\\
with $J(\wt\bX_{tr} \to \bZ_{tr}) = (Z_{tr,q+1} \cdots Z_{tr,2q+p})^{-1} = (X_{1} \cdots X_{q+p})^{-1}$, 
we then have
\vspace{-4ex}\\
\bse
& f_{\bZ_{tr}}(Z_{tr,1},\dots,Z_{tr,q},X_{1},\dots,X_{q+p}) = \frac{1}{(\sqrt{2\pi})^{2q+p} {\abs{\bSigma_{tr,x}}}^{1/2} \prod_{\ell=1}^{q+p} X_{\ell}} \\
&\exp\left\{ - \frac{1}{2} 
\left(\begin{array}{c}
Z_{tr,1}-\mu_{tr,x,1}\\
\cdots \\
Z_{tr,q}-\mu_{tr,x,q}\\
\log \frac{X_{1}}{\Phi(Z_{tr,1})}-\mu_{tr,x,q+1}\\
\cdots\\
\log \frac{X_{q}}{\Phi(Z_{tr,q})}-\mu_{tr,x,2q}\\
\log X_{q+1}-\mu_{tr,x,2q+1}\\
\cdots\\
\log X_{q+p}-\mu_{tr,x,2q+p}\\
\end{array}\right)\trans 
\bSigma_{tr,x}^{-1}
\left(\begin{array}{c}
Z_{tr,1}-\mu_{tr,x,1}\\
\cdots \\
Z_{tr,q}-\mu_{tr,x,q}\\
\log \frac{X_{1}}{\Phi(Z_{tr,1})}-\mu_{tr,x,q+1}\\
\cdots\\
\log \frac{X_{q}}{\Phi(Z_{tr,q})}-\mu_{tr,x,2q}\\
\log X_{q+1}-\mu_{tr,x,2q+1}\\
\cdots\\
\log X_{q+p}-\mu_{tr,x,2q+p}\\
\end{array}\right) 
\right\}
\ese 
\vspace{-4ex}\\
The true joint distribution of the long-term latent consumptions of the episodic and the regular components $\bX=(X_{1},\dots,X_{q+p})\trans$ in the original scale may then be obtained as  
\vspace{-4ex}\\
\bse
&& f_{\bX}(X_{1},\dots,X_{q+p}) = \int_{\Z_{tr,1}} \cdots \int_{\Z_{tr,q}} f_{\bZ_{tr}}(Z_{tr,1},\dots,Z_{tr,q},X_{1},\dots,X_{q+p}) dZ_{tr,1}\cdots dZ_{tr,q}.
\ese
\vspace{-4ex}\\
It is not possible to evaluate this integral in a closed form. 
We can numerically estimate $f_{\bX}(\bX)$ using importance sampling as 
\vspace{-4ex}\\
\bse
&& \hskip -20pt f_{\bX}(\bX) = \int_{\Z_{tr,1}} \cdots \int_{\Z_{tr,q}} \frac{f_{\bZ_{tr}}(Z_{tr,1},\dots,Z_{tr,q},X_{1},\dots,X_{q+p})} {g_{\bZ_{tr,1:q}}(Z_{tr,1},\dots,Z_{tr,q})}  g_{\bZ_{tr,1:q}}(Z_{tr,1},\dots,Z_{tr,q}) dZ_{tr,1}\cdots dZ_{tr,q} \\
&& \wh{=} \frac{1}{M}\sum_{m=1}^{M} \frac{f_{\bZ_{tr,1:q},\bX}(\bZ_{tr,1:q}^{(m)},\bX)} {g_{\bZ_{tr,1:q}}(\bZ_{tr,1:q}^{(m)})},  
\ese
\vspace{-4ex}\\
where $g_{\bZ_{tr,1:q}}(\cdot)$ is an importance sampling density 
and $\{\bZ_{tr,1:q}^{(m)}\}_{m=1}^{M} \sim g_{\bZ_{tr,1:q}}(\cdot)$ independently.

By design, we also have, for $\ell=1,\dots,q$, $(\wt{X}_{tr,\ell},\wt{X}_{tr,q+\ell})\trans \sim \MVN_{2}(\bmu_{tr,x,\ell,q+\ell},\bSigma_{tr,x,\ell,q+\ell})$ 
where $\bmu_{tr,x,\ell,q+\ell},\bSigma_{tr,x,\ell,q+\ell}$ are obtained from $\bmu_{tr,x}$ and $\bSigma_{tr,x}$ as $\bmu_{tr,x,\ell,q+\ell}=(\mu_{tr,x,\ell},\mu_{tr,x,q+\ell})\trans$ etc. 
Proceeding as above, we thus have 
\vspace{-4ex}\\
\bse
& f_{Z_{tr,\ell},Z_{tr,q+\ell}}(Z_{tr,\ell},X_{\ell}) = \frac{1}{(\sqrt{2\pi})^{2} {\abs{\bSigma_{tr,x,\ell,q+\ell}}}^{1/2} X_{\ell}} \\
&\exp\left\{ - \frac{1}{2} 
\left(\begin{array}{c}
Z_{tr,\ell}-\mu_{tr,x,\ell}\\
\log \frac{X_{\ell}}{\Phi(Z_{tr,\ell})}-\mu_{tr,x,q+\ell}\\
\end{array}\right)\trans 
\bSigma_{x,tr,\ell,q+\ell}^{-1}
\left(\begin{array}{c}
Z_{tr,\ell}-\mu_{tr,x,\ell}\\
\log \frac{X_{\ell}}{\Phi(Z_{tr,\ell})}-\mu_{tr,x,q+\ell}\\
\end{array}\right)
\right\}.
\ese 
\vspace{-4ex}\\
The marginal densities of the episodic components $X_{\ell},\ell=1,\dots,q$, can therefore be estimated as before as 
\vspace{-4ex}\\
\bse
&& \hskip -20pt f_{X_{\ell}}(X_{\ell}) = \int_{\Z_{tr,\ell}} \frac{f_{Z_{tr,\ell}, Z_{tr,q+\ell}}(Z_{tr,\ell},X_{\ell})} {g_{Z_{tr,\ell}}(Z_{tr,\ell})}  g_{Z_{tr,\ell}}(Z_{tr,\ell}) dZ_{tr,\ell} 
\wh{=} \frac{1}{M}\sum_{m=1}^{M} \frac{f_{Z_{tr,\ell},X_{\ell}}(Z_{tr,\ell}^{(m)},X_{\ell})} {g_{Z_{tr,\ell}}(Z_{tr,\ell}^{(m)})},  
\ese
\vspace{-4ex}\\
where $g_{Z_{tr,\ell}}(\cdot)$ is an importance sampling density 
and $\{Z_{tr,\ell}^{(m)}\}_{m=1}^{M} \sim g_{Z_{tr,\ell}}(\cdot)$ independently.

\vspace{-10pt}
\subsubsection*{Implementation of \cite{Zhang2011b}:}
Apart from closely mimicking the model of \cite{Zhang2011b}, 
we assume that the transformation $g_{tr}(Y) = \log Y$ is known, making the situation further favorable for \cite{Zhang2011b}. 
In our implementation, as described in Section \ref{sec: mvt copula comparison with NCI}, 
we set $g_{tr}(Y,\lambda) = \sqrt{2}\{g(Y,\lambda)-\mu(\lambda)\}/\sigma(\lambda) = \log~Y$ with $\lambda=0, \mu(\lambda)=0$ and $\sigma(\lambda)=\sqrt{2}$.
We also recall from Section \ref{sec: mvt copula comparison with NCI} that \cite{Zhang2011b} relies on 
first estimating the latent consumptions $\wt\bX_{tr,\ell,i}$'s in the transformed scale, 
then transforming them back to $\bX_{\ell,i}$'s in the original scale,
and then applying separate univariate and multivariate kernel density estimation methods to these estimates of $\bX_{\ell,i}$'s to approximate the marginal and joint densities of interest.  
Applied to data generated from model (\ref{eq: Zhang 2011b sim model}), 
the method would, however, estimate 
$\wt{X}_{tr,\ell,i}^{\prime} = \wt{X}_{tr,\ell,i}, \ell=1,\dots,q$, 
$\wt{X}_{tr,\ell,i}^{\prime} = \wt{X}_{tr,\ell,i}-\sigma_{tr,u,\ell}^{2}/2, \ell=q+1,\dots,2q+p$ for all $i=1,\dots,n$. 
To adjust for this bias, in the reverse transformation, we set 
\vspace{-4ex}\\
\bse
&& X_{\ell,i} = \Phi(\wt{X}_{tr,\ell,i}) \exp(\wt{X}_{tr,q+\ell,i}) = \Phi(\wt{X}_{tr,\ell,i}) \exp(\wt{X}_{tr,q+\ell,i}^{\prime}+\sigma_{tr,u,q+\ell}^{2}/2), ~~~~~\ell=1,\dots,q,\\
&& X_{\ell,i} = \exp(\wt{X}_{tr,q+\ell,i}) = \exp(\wt{X}_{tr,q+\ell,i}^{\prime}+\sigma_{tr,u,q+\ell}^{2}/2), ~~~~~~\ell=q+1,\dots,q+p.
\ese

\vspace{-20pt}
\subsubsection*{Parameter Choices:}
We focused on a case with dietary components of mixed types, $q=2$ episodic and $p=1$ regular, $i=1,\dots,n=1000$ subjects with $m_{i}=3$ surrogates for each $i$. 
We set $\bmu_{tr,x} = (0.75,1.00,0.15,0.15,1.00)\trans$, 
$\bSigma_{tr,x} = ((\sigma_{tr,x,\ell,\ell'}))$ as $\sigma_{tr,x,1,1}=0.25, \sigma_{tr,x,2,2}=0.15, \sigma_{tr,x,3,3}=\sigma_{tr,x,4,4}=0.25, \sigma_{tr,x,5,5}=0.05$, 
and $\sigma_{tr,x,\ell,\ell'}=0.7^{\abs{\ell-\ell'}}$ for all $\ell \neq \ell'$. 
We set $\bSigma_{tr,u} = ((\sigma_{tr,u,\ell,\ell'}))$ as $\sigma_{tr,u,\ell,\ell}=1$ for $\ell=1,\dots,q$, 
$\sigma_{tr,u,\ell,\ell}=0.125$ for $\ell=q+1,\dots,2q+p$. 
$\sigma_{tr,u,\ell,\ell'}=0.5^{\abs{\ell-\ell'}}$ for all $\ell \neq \ell'$. 
We had on average approximately $25\%$ and $18\%$ zero recalls for the two episodic components in the simulated data sets.

\vspace{-20pt}
\subsubsection*{Summary of Findings:}

Despite being very careful that the model we simulated from match the assumptions of \cite{Zhang2011b} in a transformed scale 
but also closely conform to our proposed model in the original scale, the method of \cite{Zhang2011b} outperformed our proposed method in the simulation scenario considered here, often significantly. 
A typical case is depicted in Figures \ref{fig: SimStudy from Zhang2011b 1} and \ref{fig: SimStudy from Zhang2011b 2} below.

\begin{figure}[!ht]
\begin{center}
\includegraphics[height=10cm, width=17cm, trim=0cm 0cm 0cm 0cm, clip=true]{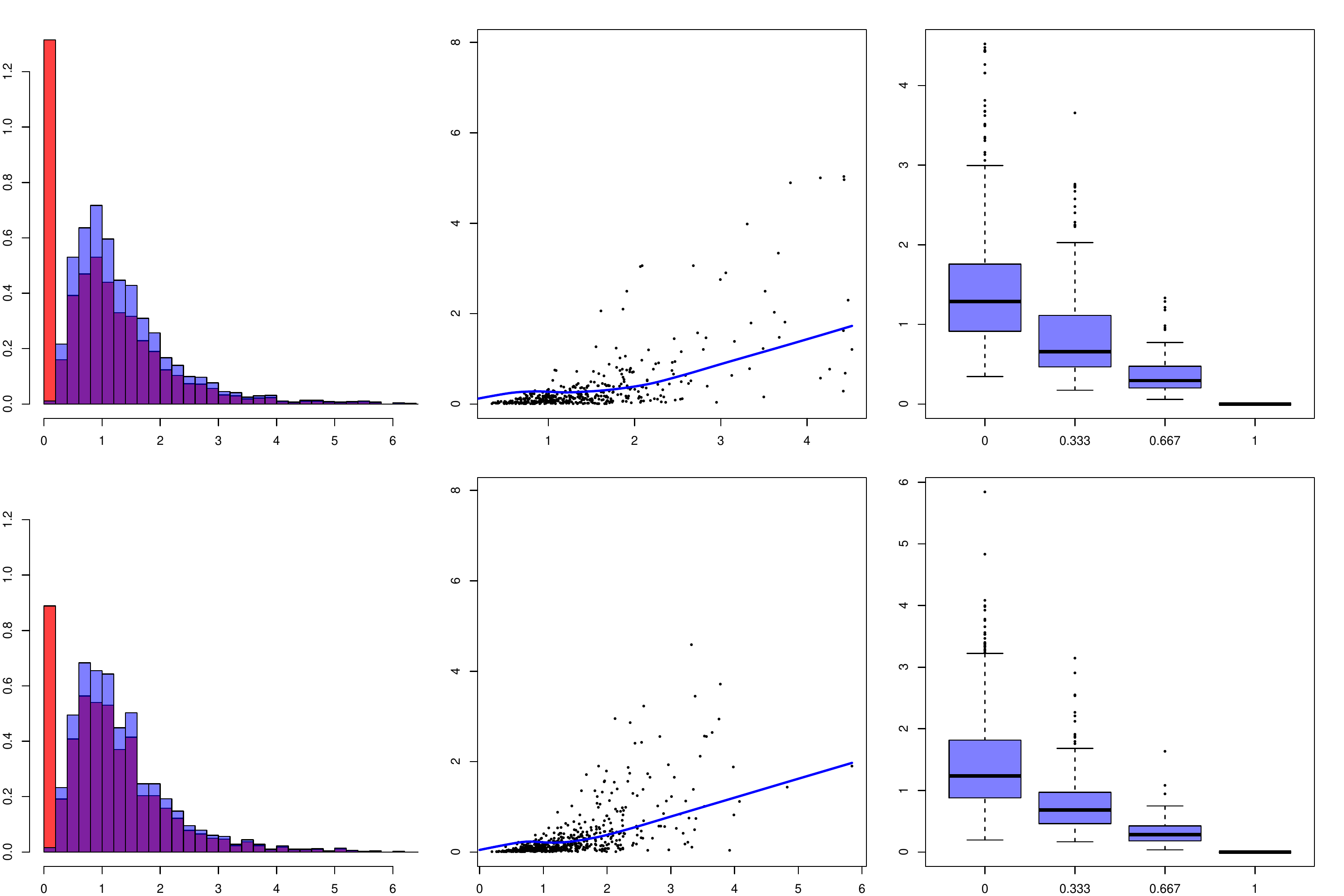}
\end{center}
\caption{\baselineskip=10pt Exploratory a data set simulated according to the process detailed in Section \ref{sec: additional sims} in the Supplementary Material 
with sample size $n = 1000$, $q=2$ episodic components and $p=1$ regular components, each subject having $m_{i}=3$ replicates. 
Left panels: histogram of recalls $Y_{\ell,i,j}$ (red) and histogram of strictly positive recalls $Y_{\ell,i,j}(>0)$ (blue) superimposed on each other; 
middle panels: subject-specific means $\overline{Y}_{\ell,i}$ vs subject-specific variances $S_{Y,\ell,i}^{2}$ when multiple strictly positive recalls are available; 
right panels: box plots of proportion of zero recalls vs corresponding subject-specific means $\overline{Y}_{\ell,i}$.
}
\label{fig: SimStudy from Zhang2011b 1}
\end{figure}

\begin{figure}[!ht]
\centering
\includegraphics[height=15cm, width=16cm, trim=0cm 0cm 0cm 0cm, clip=true]{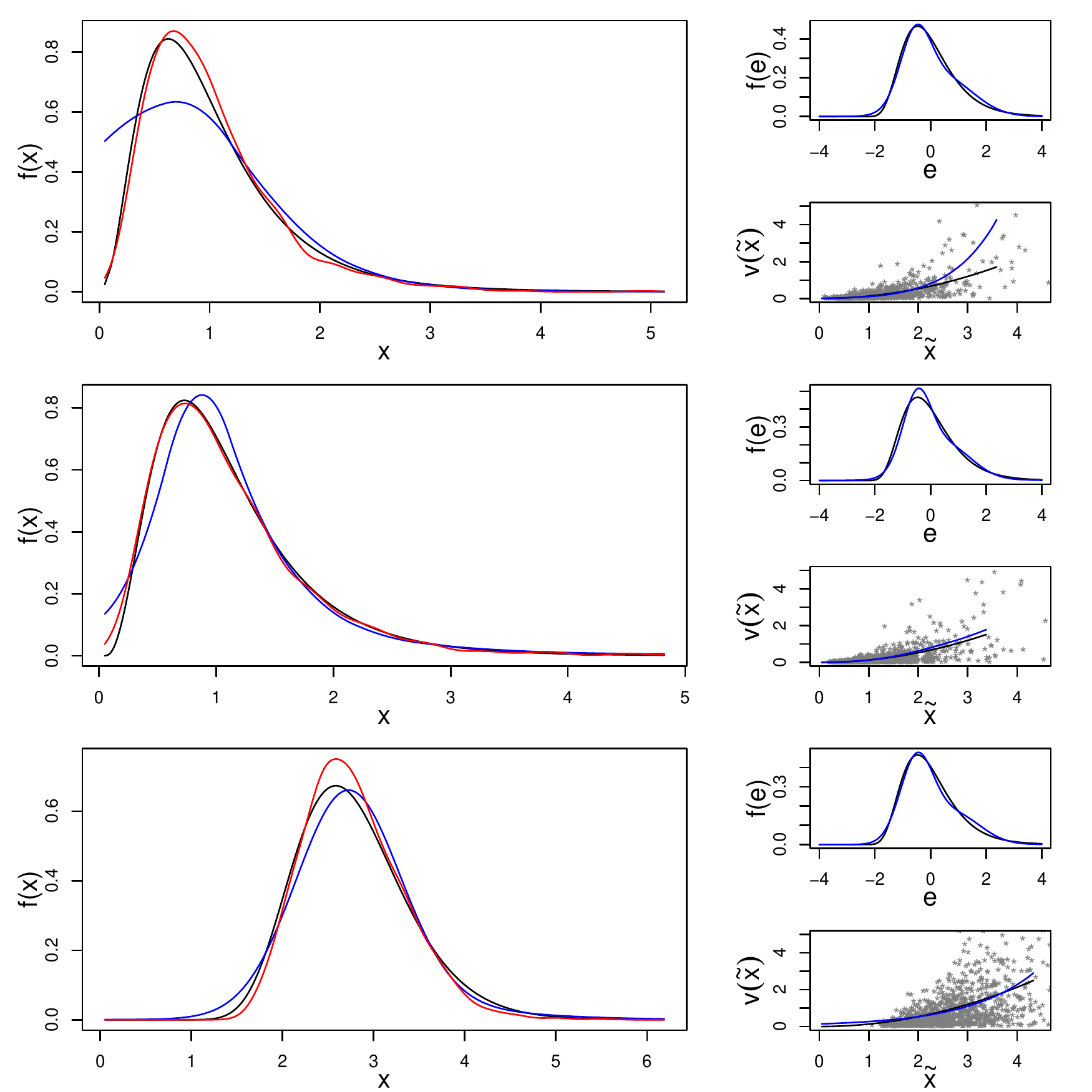}
\caption{\baselineskip=10pt 
Results for a data set simulated according to the process detailed in Section \ref{sec: additional sims} in the Supplementary Material 
with sample size $n = 1000$, $q=2$ episodic components and $p=1$ regular components, each subject having $m_{i}=3$ replicates. 
From top to bottom, the left panels show the estimated densities $f_{X,\ell}(X_{\ell})$ of the two episodic components and the one regular component, respectively, 
obtained by our method (in blue) and the method of \cite{Zhang2011b} (in red).  
The right panels show the estimated distributions of the scaled errors $f_{\epsilon,q+\ell}(\epsilon_{q+\ell})$ and the estimated variance functions $v_{\ell}(\wt{X}_{\ell}) = s_{\ell}^{2}(\wt{X}_{\ell})$, estimated by our method. 
The black lines represent the truths (the right panels) 
or their importance sampling based approximations (the left panels). 
}
\label{fig: SimStudy from Zhang2011b 2}
\end{figure}

The main challenge in modeling dietary recall data for episodic dietary components with exact zero recalls is again the sparsity of informative recalls near the left boundary, 
most of them being exact zeros. 
Our method is well suited to model reflected J-shaped densities for episodic components 
with discontinuities at the left boundaries observed in real data sets 
as well as multimodality, heavy tails etc. for both regular and episodic components. 
Our method, however, is less suited to model unimodal left skewed densities generated in the simulation scenario considered here, as is reflected in Figure \ref{fig: SimStudy from Zhang2011b 2}. 
\cite{Zhang2011b}, on the other hand, assume the densities in Box-Cox transformed scales to be exactly normal, thus perfectly symmetric, unimodal and bell-shaped. 
This means, that these densities, when transformed back to the original scale, will always have a similar shape but only more left skewed with longer right tails. 
This is clearly evident from the left panels of Figure \ref{fig: SimStudy from Zhang2011b 2}. 
This is not specific to the logarithmic transformations assumed here but is more generally true for any Box-Cox transformation. 
The method of \cite{Zhang2011b}, therefore, can never capture the discontinuities at the left boundaries for episodic dietary components. 
As a result, recall data simulated from the model of \cite{Zhang2011b} will also never closely mimick 24 hour recalls for episodic components observed in real data sets. 
This is again clearly evident from the left panels of Figure \ref{fig: SimStudy from Zhang2011b 1}, 
especially in comparison with the left panels of Figure \ref{fig: EATS Milk & Whole Grains} from the main paper. 
The parametric assumptions of \cite{Zhang2011b} are also highly restrictive for modeling other departures from normality in the transformed and hence also the original scale, 
including multimodality, heavy tails etc. 
The simulation scenario described here is thus highly restrictive and unrealistic. 

In real world dietary recall data sets, including our motivating EATS data set, the true densities of the average long-term consumptions for episodic components 
are extremely left skewed reflected J-shaped with discontinuities at the left boundaries.  
In more realistic simulation scenarios considered in Section \ref{sec: simulation studies} of the main paper, 
where the true densities conformed to such shapes, our method vastly outperformed the method of \cite{Zhang2011b}.

\bibliographystylelatex{natbib}
\bibliographylatex{ME,Copula}

\end{document}